\documentclass[12pt]{article}
\usepackage{latexsym}
\usepackage{amsmath}
\usepackage{amsfonts}
\usepackage{amssymb}
\usepackage{amsthm}
\usepackage{setspace}
\usepackage{graphicx}
\usepackage{framed}
\usepackage{natbib}
\usepackage{xcolor}
\definecolor{mygray}{gray}{0.3}
\usepackage{tikz}
\usetikzlibrary{arrows,shapes}
\usetikzlibrary{decorations.markings}
\usetikzlibrary{arrows}
\usepackage{hyperref}
 \hypersetup{
    colorlinks   = true,
     citecolor    = blue,
     linkcolor=blue
}
\newcommand\supp{\mathrm{supp}}

\newcommand\beq{\begin{equation}}
\newcommand\eeq{\end{equation}}
\newcommand\bit{\begin{itemize}}
\newcommand\eit{\end{itemize}}
\newcommand\bea{\begin{eqnarray}}
\newcommand\eea{\end{eqnarray}}
\newcommand\beas{\begin{eqnarray*}}
\newcommand\eeas{\end{eqnarray*}}
\newcommand\beqa{\begin{eqnarray}}
\newcommand\eeqa{\end{eqnarray}}

\newcommand\pfof[1]{{\bf Proof of #1:  }}
\newcommand\eop{\hfill $\square$}

\newtheorem{definition}{Definition}
\newtheorem{theorem}{Theorem}
\newtheorem{lemma}{Lemma}
\newtheorem{proposition}{Proposition}
\newtheorem{corollary}{Corollary}
\theoremstyle{remark}

\newcommand\vmax{\overline{v}}
\newcommand\vmin{\underline{v}}
\begin{document}

\title{An Optimal Distributionally Robust Auction\footnote{This paper is based on Chapter I of my PhD thesis at Stanford University. I would like to thank Dan Iancu, Gabriel Carroll, Andy Skrzypacz, Ilya Segal, Bob Wilson, Dmitry Arkhangelsky, Pavel Krivenko, Alex Bloedel and audience members at 2020 Conference on Mechanism and Institution Design for helpful comments and encouragement.}} 
\author{Alex Suzdaltsev\footnote{Higher School of Economics, Saint Petersburg, Russia, \texttt{asuzdaltsev@gmail.com}
}}
\date{\today}

\maketitle

\begin{abstract}
An indivisible object may be sold to one of $n$ agents who know their valuations of the object. The seller would like to use a revenue-maximizing mechanism but her knowledge of the valuations' distribution is scarce: she knows only the means (which may be different) and an upper bound for valuations. Valuations may be correlated. 

Using a constructive approach based on duality, we prove that a mechanism that maximizes the worst-case expected revenue among all deterministic dominant-strategy incentive compatible, ex post individually rational mechanisms is such that the object  should be awarded to the agent with the highest \emph{linear score} provided it is nonnegative. Linear scores are bidder-specific linear functions of bids. The set of optimal mechanisms includes other mechanisms but all those have to be close to the optimal linear score auction in a certain sense. When means are high, all optimal mechanisms share the linearity property. Second-price auction without a reserve is an optimal mechanism when the number of symmetric bidders is sufficiently high.  

\end{abstract}

\newpage
\normalsize
\section{Introduction}
One of the most basic problems in mechanism design is allocation of an item among $n$ buyers by a revenue-maximizing seller. The classical solutions due to \cite{myerson1981optimal} and \cite{cremer1988full} have been obtained under the assumption of expected revenue maximization, where the expectation is taken with respect to a fixed distribution of bidders' values. This distribution is assumed either to (1) be objectively known by the seller or (2) represent her subjective beliefs. 

While this assumption has led to important insights and does constitute a reasonable approximation to reality in certain cases, it may be less plausible in other situations. For example, a new good may be for sale, or a seller might not be a subjective expected utility maximizer. Even if the good is not new, there may exist purely statistical problems with non-parametric estimation of density functions, especially multiple-dimensional. On the other hand, even if a seller acts as a subjective expected utility maximizer in day-to-day ``intuitive'' decision problems that do not require explicit numerical articulation of beliefs, implementing a classical optimal auction requires first writing down a mathematical formula expressing the beliefs, and this may be infeasible.
 
Motivated by the above observations, in this paper we treat the distribution as unknown and consider a problem of finding an optimal distributionally robust mechanism, i.e. a mechanism that provides the best worst-case expected revenue where the worst case is over all joint distributions of values lying in a large class. (The values are private and known by the bidders.) The class of distributions consists of all distributions with a known vector of means and known bounds for support. One may think of the known vector of means and bounds for support as the best \emph{sparse educated guess} that the seller has about the joint distribution of values, incorporating her beliefs or knowledge about the possible asymmetry of the buyers.  %

We restrict attention to \emph{deterministic ex post} mechanisms, with \emph{ex post} meaning dominant-strategy incentive compatible and ex post individually rational, a term used by \cite{segal2003optimal}. The motivation for this property is standard: one would like, in line with Wilson doctrine \citep{wilson1987game}, to make a mechanism robust to misspecification of bidders' beliefs\footnote{\cite{chung2007foundations} give a maxmin foundation to dominant-strategy incentive compatibility under the assumption of known regular value distribution. We do not know whether an analogous foundation holds in our setting.}. This is even more important in our setting where the seller is concerned about the misspecification of her own prior and thus does not have a natural guess for the bidders' common prior, if one exists. The restriction to deterministic mechanisms may be motivated by the fact that, as noted by \cite{bergemann2016informationally}, randomization may be difficult for a seller to credibly commit to. Also, randomized mechanisms may be quite complex. We note, however, that this restriction is with loss of revenue, as randomization typically yields improvement in maxmin optimization problems. %
Despite the restrictions, the set of mechanisms we consider is still quite large: for example, instead of having an include-all auction, one may first conduct an auction among buyers 1 and 2, and if the object fails to be allocated, then approach buyer 4 with a price that depends on values reported in the first auction, and if the object fails to be allocated still, approach buyers 3 and 5 with an auction designed depending on all previous reports, and so on\footnote{Because we assume that the values are known by the bidders, there are no informational spillovers in the sequential mechanism, so one may model it as a one-shot game.}. 

The main result of the paper says that the considered maxmin problem admits a simple solution that we call a \emph{linear score auction} (LSA). Its defining feature is that the winner of the auction is the bidder whose \emph{linear score} is the highest, provided it is nonnegative. Linear scores are bidder-specific linear\footnote{Here, we use the term ``linear'' in the calculus sense of the word, i.e. the functions are, in fact,
possibly affine.} functions of bids\footnote{Our score auction has nothing to do with \emph{scoring auctions} commonly used in procurement (see, e.g., \cite{che1993design}). In procurement scoring auctions, a score combines information about a bidder's costs and the quality of her good. Here, a score combines information about a bidder's valuation and the prior information on that valuation's distribution.}. Transfers are pinned down by incentive compatibility in the standard way. Note that this mechanism can be regarded as a linear version of the Myersonian optimal auction: indeed, that auction under asymmetric value distributions is effectively a (generically) nonlinear score auction in which the scores equal to bidders' ironed virtual values. The difference between Myersonian optimal auction and linear score auction is illustrated in figure \ref{fig:linvsnonlin}. Also note that when all score functions are identical, the linear score auction reduces to the usual second-price auction with a reserve price. Thus, the distributional robustness of this common auction format be an additional rationale for using it in practice.
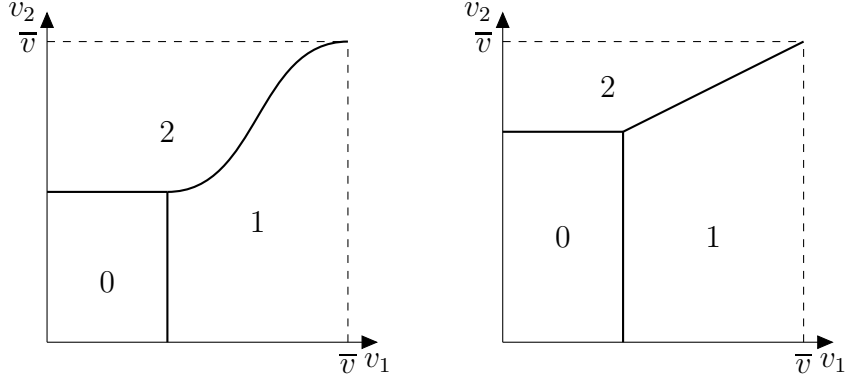
\begin{figure}[h!]
\centering
\begin{tikzpicture}[>=triangle 45, xscale=4,yscale=4]
\draw[->] (0,0) -- (1.1,0) node[below] {$v_1$};
\draw[->] (0,0) -- (0,1.1) node[left] {$v_2$};
\draw[dashed] (1,0)--(1,1);
\draw[dashed] (0,1)--(1,1);
\node[below] at (1,0){$\vmax$};
\node[left] at (0,1){$\vmax$};

\draw[thick] (0.4,0)--(0.4,0.5);
\draw[thick] (0,0.5)--(0.4,0.5);
\draw[thick] (0.4,0.5) to [out=0,in=180] (1,1);
\node at (0.7,0.4){1};
\node at (0.4,0.7){2};
\node at (0.2,0.2){0};

\end{tikzpicture}
\hspace{1em}
\begin{tikzpicture}[>=triangle 45, xscale=4,yscale=4]
\draw[->] (0,0) -- (1.1,0) node[below] {$v_1$};
\draw[->] (0,0) -- (0,1.1) node[left] {$v_2$};
\draw[dashed] (1,0)--(1,1);
\draw[dashed] (0,1)--(1,1);
\node[below] at (1,0){$\vmax$};
\node[left] at (0,1){$\vmax$};

\draw[thick] (0.4,0)--(0.4,0.7);
\draw[thick] (0,0.7)--(0.4,0.7);
\draw[thick] (0.4,0.7)--(1,1);
\node at (0.7,0.35){1};
\node at (0.35,0.85){2};
\node at (0.2,0.35){0};

\end{tikzpicture}

\caption{In a Myersonian optimal auction with asymmetric value distributions (left), the boundary determining which bidder gets the object respects the intricacies of bidders' value distributions. In a linear score auction (right), the boundary is linear
. Digits 1 and 2 denote areas of value space where the corresponding bidders get the object. Zero denotes areas where the object is kept by the seller.}\label{fig:linvsnonlin}
\end{figure}

The proof of the result uses strong linear programming duality. To get a duality-free intuition for why the result is true, consider the following simple example. Suppose there are two bidders with valuations lying in $[0,1]$ and expected valuations of $m_1$ and $m_2$ satisfying $m_1+m_2\geq 1$. Consider the simple second-price auction without reserve, which is one feasible linear score auction. Consider any distribution $F$ satisfying the above constraints. Under any profile of values $(v_1,v_2)\in \supp (F)$ with $v_1> v_2$ the first bidder wins, and the seller gets $v_2$. Note, however, that if Nature transfers probability mass from $(v_1,v_2)$ to $(1,v_2)$ the seller's revenue does not change, but $\mathbb{E}(v_1)$ increases. The same is true for profiles with $v_2>v_1$. Thus, for any distribution $F$ there exists a distribution $\hat{F}$ satisfying $\mathbb{E}_{\hat{F}}\max\{v_1,v_2\}=1$, $\mathbb{E}_{\hat{F}}v_i\geq m_i$,
  such that the seller's revenue is the same. But then, after moving from $F$ to $\hat{F}$, Nature can redistribute mass within the set $\{v:\max\{v_1,v_2\}=1\}$ so that the means are lowered back to $m_i$ and the seller's revenue decreases. (This is possible exactly when $m_1+m_2\geq 1$.) This shows that a seller caring about the worst-case can restrict attention to distributions $F$ satisfying $\mathbb{E}_{F}\max\{v_1,v_2\}=1$. Call such distributions ``potentially worst-case distributions''. The above argument is valid for any feasible mechanism that always allocates the good. But then, because of the identity
\[\mathbb{E}_F(\min\{v_1,v_2\})+\mathbb{E}_F(\max\{v_1,v_2\})\equiv m_1+m_2,\] 
the expected revenue of the SPA under a potentially worst-case distribution is equal to $\mathbb{E}_F(\min\{v_1,v_2\})=m_1+m_2-1$. That is, when seller uses the SPA (and not some mechanism with nonlinearities that always allocates the good), her expected revenue depends on a (potentially worst-case) distribution $F$ only through the known means $m_i$. This is a clear sign of robustness\footnote{The analysis in this example does not prove that the SPA is optimal -- when means are different, it is not, even among the mechanisms that always allocate the good. It merely shows why mechanisms of a specific form, of which SPA is an example, yield robustness and thus are good candidates for an optimal mechanism.}.

The actual proof starts with an observation that the worst-case expected revenue is equal to the \emph{convex closure} of a mechanism's transfer function, evaluated at the point $m$ where $m$ is the vector of means. Using an analytic representation of convex closures afforded by a classic strong duality result, we then show that for any feasible mechanism there exists a linear score auction that does no worse. This construction is the heart of the paper. While one can ascertain the validity of the construction visually in the case of two bidders, the proof is trickier in higher dimensions. The key to the construction is that the vector of (generalized) ``reserve prices'' for the desired revenue-improving linear score auction comes from a fixed point of a certain piecewise-linear map derived from the original mechanism. The existence of such a fixed point is guaranteed by Brouwer theorem. However, not any fixed point can yield the desired construction. Thus, analyzing the structure of the set of fixed points constitutes a major part of analysis.

Solving a finite-dimensional optimization problem, we then characterize  the set of optimal parameters for the LSA. Qualitatively, the optimal auction is similar to the Myersonian solution: the optimal auction discriminates against stronger bidders (those with larger means) to reduce their information rents and intensify competition among all bidders. This is manifested in the fact that for a fixed bid $b$, the score of a stronger bidder is lower than the score for a weaker bidder. However, there are new features. First, there may be multiple optimal vectors of parameters. This a common scenario in maxmin optimization problems. Second, when bidders are symmetric, the set of optimal vectors of reserve prices depends on $n$: it decreases in $n$ in strong set order and includes zero for all sufficiently high $n$. This contrasts with the Myersonian solution where the reserve price does not depend on $n$ with symmetric bidders. This may be interpreted as evidence that competition is a more powerful device to protect against low-revenue distributions than a reserve price, so the latter is no longer needed when the number of bidders is sufficiently large. Finally, when bidders are sufficiently \emph{a}symmetric, a new phenomenon of \emph{weak exclusion} arises: it may be weakly optimal to exclude a (low-mean) bidder from an auction completely. This may never happen in the classical model when value distributions have overlapping supports. This is because the weak bidder's value is expected to be so low that including the bidder does not change the worst-case expected revenue. Interestingly, weak exclusion may happen only if $n\geq 3$ and only if the means of other bidders' valuations are sufficiently high. %

The main result only implies that \emph{some} optimal mechanism takes the form of linear score auction. However, as alluded to above, in general many mechanisms can share the same worst-case expected revenue. In section \ref{set},  we characterize the whole set of optimal mechanisms for the case of two bidders. We show that a mechanism is optimal if and only if it is ``sufficiently close'' to the optimal linear score auction in a certain sense. When means are sufficiently high, the ``safe neighborhood'' of the optimal LSA collapses in a way that \emph{any} optimal mechanism must coincide with the LSA in an area where the values are relatively high. 

\subsection{Relation to literature}
This paper contributes to the growing literature on robust mechanism design. The closest contributions to ours are \cite{carrasco2018optimal}, \cite{koccyiugit2017distributionally}, \cite{neeman2003effectiveness}, \cite{che2019robust} and \cite{suzdal2020distributionally}. \cite{carrasco2018optimal} study the problem of selling the good to a single agent by seller who maximizes worst-case expected revenue while knowing the first $N$ moments of distribution. They characterize the optimal randomized mechanism, and also the optimal deterministic posted price for the case when a mean and upper bound is known, as in our paper. Their approach also effectively uses duality. \cite{koccyiugit2017distributionally} study, among other settings, second-price auctions with reserve price when there are $n$ symmetric bidders with a known lower bound for (the common) $\mathbb{E}(v_i)$. They characterize the optimal reserve price but do not study the question of whether the SPA is an optimal deterministic mechanism. They also show that randomized mechanisms yield strictly more revenue in this setting. \cite{che2019robust} finds an optimal randomized reserve price in a second-price auction when the set of joint value distributions is the same as in this paper. 
\cite{neeman2003effectiveness} finds the optimal reserve price in a similar setting where the set of distributions is the same as in \cite{koccyiugit2017distributionally}, but the criterion is the worst-case ratio of expected revenue to expected full surplus, rather than expected revenue itself. Finally, \cite{suzdal2020distributionally} proves that a maxmin reserve price in a second-price auction is equal to seller's valuation when the values are iid and the mean, and an upper bound on either values or variance is known.

\cite{bose2006optimal} show that when the seller has an arbitrary set of iid priors that includes the bidders' common prior, an auction that equalizes ex post revenue is optimal. However, to implement such an auction, the seller still has to know the bidders' common prior, as the paper assumes Bayesian, rather than dominant-strategy, implementation. The paper also considers the case of maxmin optimization on the part of bidders. 
\cite{azar2013parametric} find a maxmin-optimal posted price mechanism (among posted price mechanisms only) for the case where there is an unlimited supply of the good and mean and variance of every bidder's marginal value distribution are known. \cite{bergemann2011robust} and \cite{bergemann2008pricing} consider minimax-regret pricing under various specifications of incomplete demand knowledge. 
\cite{carrasco2018robust} study mechanism design where a single agent's utility is nonlinear in allocation and all distributions with given bounds for support and mean are possible. \cite{auster2018robust} considers optimal design for a buyer in a lemons problem where all distributions of seller's costs are possible. 
\cite{he2019robustly} seek an optimal reserve price in an SPA when a (common) marginal distribution of every bidder's value is known but the correlation structure is unknown. They find that the optimal deterministic reserve price goes to zero when the number of bidders grows without bound, which is in accord with our results.  \cite{carroll2017robustness} studies the problem of multidimensional screening where again marginals, but not the correlation structure, are known.  \cite{chen2019distribution} study multidimensional screening under moment information.
 
\cite{wolitzky2016mechanism} studies optimal mechanisms for bilateral trade when agents are maxmin optimizers, with their uncertainty sets being similar to the one studied in the present paper (a buyer knows only the mean of a seller's valuation and bounds on its support, and vice versa). He also proposes a foundation for such uncertainty sets via agents' uncertainty over each other's information structures. Unfortunately, such a microfoundation is not directly applicable to the present paper's set-up because of our assumption that correlation between values can be arbitrary. %

Another strand of literature has concentrated on establishing the performance bounds of ``simple'' mechanisms, relative to optimal ones. For example, \cite{azar2013optimal} show, among other results, that when bidders' values are independent, distributions' hazard rates are monotone, and the vector of means $m$ is known, running a second-price auction with the vector of individual reserves equal to $m$ ensures no less than $\frac{1}{2e}$ of the optimal Myersonian revenue. Note that  this approach implicitly tries to minimize the worst-case relative \emph{regret} rather maximize the worst-case performance per se. Other important contributions in this vein include \cite{hartline2009simple}, \cite{dhangwatnotai2015revenue} and \cite{allouah2020prior}, \cite{giannakopoulos2019robust}. Note that the present paper shows that the distinction between ``simple'' and optimal may not exist when optimality is in the maxmin sense: a relatively simple mechanism (an LSA) is optimal. 

A natural idea when type distribution is unknown is to ask the agents about it (see, e.g., \cite{brooks2013surveying}). We rule out such schemes as we assume that bidders may lack such information themselves. Another approach is to try to infer the distribution of a bidder's value from other bidders' reports and then compute the empirical virtual value functions \citep{segal2003optimal}. Such schemes can asymptotically achieve full-distributional-information revenue\footnote{See also \cite{loertscher2020asymptotically} for an asymptotically optimal prior-free clock auction.} They \emph{are} feasible in the present paper but they turn out not to be (strictly) optimal: intuitively, to run them successfully, the seller has to have some prior knowledge of the values' correlation (best if values are conditionally iid) which she lacks by assumption. \cite{segal2003optimal} also characterized an optimal \emph{ex post} mechanism with correlated private values, which is a generalization of the Myersonian optimal mechanism\footnote{\cite{segal2003optimal} characterizes the mechanism under a restrictive assumption on conditional virtual values. See, e.g., \cite{papadimitriou2011optimal} for a treatment of a more general case.}.

A related problem in robust design is finding mechanisms robust to misspecification of agents' information structures, rather than the designer's prior. \cite{brooks2019optimal} identify an optimal mechanism in the common value setting, while \cite{du2018robust} identifies a simple mechanism that asymptotically extracts full surplus. This strand of literature still assumes a shared common prior between the designer and agents and that bidders' information structures are common knowledge among them. The present paper, in contrast, dispenses with both of these assumptions but assumes a rather narrow set of information structures: the seller knows that every buyer's information structure is such that she knows her value. Thus, our approach may be seen as complementary to theirs, albeit only indirectly so as our analysis does not apply to the case of interdependent values, which is the focus of the aforementioned papers. 

Others kinds of robustness explored in the economics literature include robustness to technology or preferences, robustness to strategic behavior and robustness to interaction among agents and are surveyed by \cite{carroll2018robustness}.

There also exists a large literature on distributionally robust non-mechanism-design optimization in operations research, with the case of ``known moments'' being popular.  This literature has been apparently pioneered by an economist \citep{scarf1958min}.  See \cite{wiesemann2014distributionally}, \cite{see2010robust}, \cite{goh2010distributionally}, \cite{delage2010distributionally}, \cite{popescu2007robust} and references therein.

Finally, the seller in our model may be seen as ambiguity-averse. Note that our model conforms to the \cite{gilboa1989maxmin} model of maxmin expected utility when one takes the set of all value profiles as the set of states of the world, and mechanisms as acts, because the set of priors we consider is convex. This would fail had we assumed, for instance, that values are known to be iid but the exact distribution is unknown.

\subsection{Organization of the paper}
This paper is organized as follows. We introduce the model in section~\ref{model}, then set the stage for the proof of the main result in section \ref{prep} and present the proof itself in section~\ref{proof}. We then find and present the parametric solutions and discuss comparative statics in section~\ref{paramsol}. Section~\ref{set} discusses the set of optimal mechanisms, section~\ref{ext} -- some extensions. Section~\ref{concl} concludes. Main missing proofs are stated in the \hyperref[app]{Appendix}. Other missing proofs, along with the discussion of worst-case distributions and certain examples, are in the Online Appendix that one can find at the end of this file.

\section{Model}\label{model}
\textbf{Notation.} Statements of the form ``for all $i$'' should be interpreted as ``for all $i\in\{1,\ldots,n\}$''. Symbols without subscripts may refer to either scalars or vectors. This should not create confusion in most cases. $v'>v$ for vectors means that $v'_i>v_i$ for all $i$, and similarly for $v'\geq v$. $\i_k$ is a vector of ones of size $k$. Value vectors are column vectors. Dual variable $\lambda$ is a row vector.

One indivisible object may be sold to one of $n\geq 2$ potential buyers. Buyers know their values for the object, $v_i$, which may be correlated. The seller knows that buyers know their values; however, she lacks detailed information about the joint distribution $F$. She only knows that (1) the support of $F$ is contained in $[0,\vmax]^n$ for some $\vmax>0$ and (2) $\int v_i dF=m_i$, $i=1,\ldots,n$.  Denote by $\Delta(m,\vmax)$ the set of all Borel distributions on $[0,\vmax]^n$ satisfying these conditions. The seller's valuation of the object is zero. Denote by $V$ the set of possible vectors of values $V\in[0,\vmax]^n$ and by $V_{-i}$ the set of vectors of values of all bidders except $i$ ($V_{-i}=[0,1]^{n-1}$).

We restrict attention to dominant-strategy incentive compatible  and ex post individually rational mechanisms. The definitions of these properties are standard. %
Further, here we restrict attention to deterministic mechanisms, i.e. mechanisms such that each bidder gets the good with probability 0 or 1 under any reported profile of valuations $v$. Finally, we apply Revelation principle, and consider only direct mechanisms
\footnote{Revelation principle is valid in our setting for implementation in dominant strategies.}.

Denote a direct deterministic mechanism by $M=\{x_i(v), t_i(v)\}$, $i=1,\ldots,n$, where $x_i(\cdot)$ are measurable allocation functions $[0,\vmax]^{n}\to \{0,1\}$ satisfying  $\sum_{i=1}^n x_i(v)\leq 1$ and $t_i(\cdot)$ are measurable transfer functions $[0,\vmax]^n\to\mathbb{R}$.

Denote the set of dominant-strategy incentive compatible  and ex post individually rational deterministic mechanisms by $\mathcal{M}$. Denote the seller's expected revenue given mechanism $M$ and distribution $F$ by $$R(M,F):=\int\sum_{i=1}^n t_i^M(v)dF.$$ 
Then, the problem that we consider in this paper (and the seller considers herself) is
\beq\label{problem}
\sup\limits_{M\in\mathcal{M}}\inf\limits_{F\in\Delta(m,\vmax)}R(M,F).
\eeq
Denote the value of the inner problem in \eqref{problem}, by $\underline{R}(M)$. That is, $\underline{R}(M)$ is the revenue guarantee of a mechanism $M$.

The main result is that problem \eqref{problem} admits a solution belonging to a simple class of mechanisms described below.

\begin{definition}\label{lsa-def}
Define bidder $i$'s \emph{linear score} $s_i(v_i)$ by $s_i(v_i):=\beta_iv_i-\alpha_i$ where $\alpha_i\geq 0$, $\beta_i>0$ are bidder-specific parameters. A \emph{linear score auction} is a mechanism 
$\{x(v),t(v)\}\in\mathcal{M}$ such that for $i=1,\ldots,n$:
\beq
x_i(v)=\begin{cases}
1, & s_i(v_i) > \max\{\max\limits_{j\neq i}s_j(v_j),0\}\\
0, & s_i(v_i) < \max\{\max\limits_{j\neq i}s_j(v_j),0\}.
\end{cases} 
\eeq
for some $\alpha_i$, $\beta_i\geq 0$.
\end{definition}

Note that a linear score auction with $\alpha_i=r$, $\beta_i=1
$ corresponds to the usual second-price auction with the reserve $r$. Hence, a linear score auction may be regarded as a simple generalization of the second-price auction that accounts for possible asymmetries of bidders.

\begin{definition}\label{lsa-ch-def}
A linear score auction is a \emph{corner-hitting linear score auction} if and only if, for some $r\in[0,\vmax]^n$, $\beta_i=\frac{1}{\vmax-r_i}$, $\alpha_i=\frac{r_i}{\vmax-r_i}$ if $r_i<\vmax$ and $\beta_i=1$, $\alpha_i=r_i$ if $r_i=\vmax$.
\end{definition}

The qualifier ``corner-hitting'' comes from the fact that in a depiction of such an auction's allocation function, the boundary separating the areas of value space corresponding to different bidders getting the object goes in the corner. Parameters $r$ may be thought of as generalized reserve prices  (see figure \ref{fig:auctions}). 

Equivalently, a corner-hitting linear score auction is a linear score auction such that for every bidder $i$, her maximum possible score is either one (when $r_i<\vmax$) or zero (when $r_i=\vmax$).

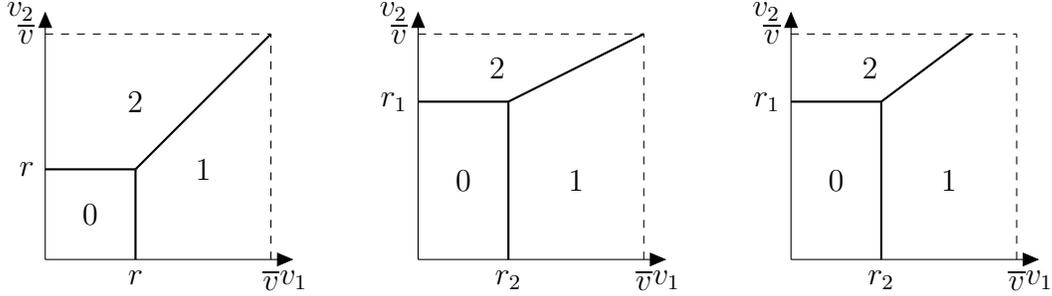
\begin{figure}[h!]
\centering
\begin{tikzpicture}[>=triangle 45, xscale=3,yscale=3]
\draw[->] (0,0) -- (1.1,0) node[below] {$v_1$};
\draw[->] (0,0) -- (0,1.1) node[left] {$v_2$};
\draw[dashed] (1,0)--(1,1);
\draw[dashed] (0,1)--(1,1);
\node[below] at (1,0){$\vmax$};
\node[left] at (0,1){$\vmax$};

\draw[thick] (0.4,0)--(0.4,0.4);
\draw[thick] (0,0.4)--(0.4,0.4);
\draw[thick] (0.4,0.4)--(1,1);
\node[left] at (0,0.4){$r$};
\node[below] at (0.4,0){$r$};

\node at (0.7,0.4){1};
\node at (0.4,0.7){2};
\node at (0.2,0.2){0};

\end{tikzpicture}
\hspace{1em}
\begin{tikzpicture}[>=triangle 45, xscale=3,yscale=3]
\draw[->] (0,0) -- (1.1,0) node[below] {$v_1$};
\draw[->] (0,0) -- (0,1.1) node[left] {$v_2$};
\draw[dashed] (1,0)--(1,1);
\draw[dashed] (0,1)--(1,1);
\node[below] at (1,0){$\vmax$};
\node[left] at (0,1){$\vmax$};

\draw[thick] (0.4,0)--(0.4,0.7);
\draw[thick] (0,0.7)--(0.4,0.7);
\draw[thick] (0.4,0.7)--(1,1);
\node[left] at (0,0.7){$r_1$};
\node[below] at (0.4,0){$r_2$};

\node at (0.7,0.35){1};
\node at (0.35,0.85){2};
\node at (0.2,0.35){0};

\end{tikzpicture}
\hspace{1em}
\begin{tikzpicture}[>=triangle 45, xscale=3,yscale=3]
\draw[->] (0,0) -- (1.1,0) node[below] {$v_1$};
\draw[->] (0,0) -- (0,1.1) node[left] {$v_2$};
\draw[dashed] (1,0)--(1,1);
\draw[dashed] (0,1)--(1,1);
\node[below] at (1,0){$\vmax$};
\node[left] at (0,1){$\vmax$};

\draw[thick] (0.4,0)--(0.4,0.7);
\draw[thick] (0,0.7)--(0.4,0.7);
\draw[thick] (0.4,0.7)--(0.8,1);
\node[left] at (0,0.7){$r_1$};
\node[below] at (0.4,0){$r_2$};

\node at (0.7,0.35){1};
\node at (0.35,0.85){2};
\node at (0.2,0.35){0};

\end{tikzpicture}
\caption{Three linear score auctions with two bidders. A standard second-price auction with a reserve price $r$ (left), and two asymmetric auctions favoring bidder 1 (center and right). Note that the rightmost auction is not a corner-hitting linear score auction, while the other two are.}\label{fig:auctions}
\end{figure}

\begin{theorem}\label{main}
For every mechanism $M\in\mathcal{M}$, there exists a corner-hitting linear score auction $LSA^0$ such that $\underline{R}(LSA^0)\geq \underline{R}(M)$.
\end{theorem}

\begin{proposition}\label{optauction exists}
There exists a corner-hitting linear score auction that maximizes the revenue guarantee $\underline{R}(M)$ among all corner-hitting linear score auctions.
\end{proposition}

\begin{corollary}
There exists a mechanism that solves \eqref{problem} and is a corner-hitting linear score auction.
\end{corollary}

Note that theorem \ref{main} establishes the form of only one optimal mechanism; in general there exist other optimal mechanisms, and there may even be multiple optimal mechanisms that are linear score auctions (see section \ref{set} below).

A well-known characterization of the set $\mathcal{M}$ (given in, e.g., \cite{segal2003optimal}) is as follows.
\begin{lemma}\label{thresholds}
A mechanism $\{x_i(v), t_i(v)\}\in \mathcal{M}$ if and only if for each bidder $i$ there exist a function $p_i$: $[0,\vmax]^{n-1}\to[0,\vmax]$ and a function $z_i: [0,\vmax]^{n-1}\to \mathbb{R}_+$ such that for every valuation profile $v\in[0,\vmax]^n$,
\beq
x_i(v)=
\begin{cases}
1, & v_i > p_i(v_{-i})\\
0, & v_i < p_i(v_{-i});
\end{cases}
\eeq
\beq
t_i(v)=p_i(v_{-i})x_i(v)-z_i(v_{-i}).
\eeq
\end{lemma}

Lemma \ref{thresholds} follows directly from the standard characterization of dominant-strategy incentive-compatibility: monotonicity of the allocation function and the envelope formula that pins down transfers (up to a constant) as a function of allocation.

Second, revenue maximization in \ref{problem} implies that at the optimum $z_i(v_{-i})\equiv 0$.

Lemma \ref{thresholds} leaves room for determining allocation for $v$ such that $v_i=p_i(v_{-i})$ for some $i$. As we allow for discrete distributions (and the worst-case distributions will turn out to be discrete, see section \ref{wcdistrsect}), the determination of allocation on this measure-zero set (``tie-breaking rule'') a priori might matter for the expected revenue (although it won't in fact matter for the optimal mechanism).  

Thus, finding an optimal mechanism in the set $\mathcal{M}$ boils down to (1) finding \emph{threshold functions} $p_i(v_{-i})$ such that the bidder $i$ gets the object and pays $p_i(v_{-i})$ if his report is higher than $p_i(v_{-i})$; (2) determining the tie-breaking rule.

Note that for any corner hitting linear score auction functions $p_i$ have a certain simple form: 
\beq\label{p-lsa}
p_i(v_{-i})=
r_i+(\vmax-r_i)\max\left\{0,\max\limits_{j\in I(r)\setminus \{i\}}\frac{v_j-r_j}{\vmax-r_{j}}\right\};
\eeq
for some $r\in[0,\vmax]^n.$ where the set $I(r)$ is given by $\{i:r_i<\vmax\}$.

In what follows, we will use a shorthand notation $LSA(r)$, with some $r_i$ possibly equal to $\vmax$, to refer to a corner-hitting linear score auction in which $i\in I$ if and only if $r_i<\vmax$, and threshold functions are given by \eqref{p-lsa} for the parameter vector $r$. 

\section{Preparations for the proof: duality}\label{prep}
First, we reformulate the inner problem in \eqref{main} (Nature's problem) using an appropriate linear programming duality result. This allows for a more tractable expression reflecting the dependence of the worst-case expected revenue on the mechanism. 

A revealing characterization of worst-case revenue under a given mechanism $M$ is as follows. Denote by $conv[t]$ the convex closure of the function $t$, i.e. for all $v\in [0,\vmax]^n$, $conv[t](v):=\inf\{R|(v,R)\in conv(graph(t))\}$ where $conv(graph(t))$ is the convex hull of the graph of $t(\cdot)$.

\begin{lemma}\label{convex hull}
For any mechanism $M$, \beq\label{ch-representation}
\underline{R}(M)=conv[t^M](m).
\eeq
\end{lemma} 
\begin{proof}
See \hyperref[app]{Appendix}.\eop
\end{proof}

The intuition behind the representation \eqref{ch-representation} is exactly as in the Bayesian persuasion literature started by \cite{kamenica2011bayesian}: Nature can achieve expected revenue $R$ with some feasible distribution if and only if $(m,R)$ is in the convex hull of graph of $t^M$, so the minimal possible revenue is achieved at the lower boundary of the convex hull, similar to how the highest sender's utility is achieved at the upper boundary of a convex hull of $V(\mu)$ at the point $\mu_0$ in \cite{kamenica2011bayesian}. Whereas in \cite{kamenica2011bayesian}, the constraint is that $\mu_0$, the prior belief, has to be the mean of posterior beliefs, here the constraint is that $m$ has to be the mean of values. Even though the representation \eqref{ch-representation} is intuitive, it is not yet particularly convenient for solving the mechanism design problem. Fortunately, by Fenchel-Moreau-Rockafellar duality theorem (see \cite{bonnans2000perturbation}, theorem 2.113), convex closure of a function can be equivalently represented as the function's biconjugate, that is, 
\beq\label{biconj}
conv[t^M](m)=\sup\limits_{\lambda\in\mathbb{R}^n} \left(\lambda m+\inf\limits_{v\in[0,\vmax]^n} \left\{t^{M}(v)-\lambda v\right\}\right)
\eeq
for all $m\in(0,\vmax)^n$. 

Alternatively, one may have arrived at an expression \eqref{biconj} by considering the dual of Nature's linear problem directly, and invoking strong duality \citep{smith1995generalized}. 

An important co-result useful in the proof of theorem \ref{main} is that the supremum in \eqref{biconj} is achieved:
\begin{lemma}\label{optexistlambda}
The supremum in \eqref{biconj} is achieved at some $\lambda^*(M)\in\mathbb{R}^n$.
\end{lemma}
Lemma \ref{optexistlambda} follows from the results in \cite{smith1995generalized} as the transfer function is bounded from below. In the Online Appendix, we provide a direct proof of this lemma.

The $\lambda^*(M)$ are the optimal dual variables and the Lagrange multipliers on the mean constraints in Nature's problem. If unique, $\lambda^*$ may also be interpreted as the local slope of the convex closure of $t^M(v)$ at $m$. Note that we allow for both positive and negative lambda -- this stems from our specification of mean constraints as equalities. In the Online Appendix, we give examples of mechanisms for which $\lambda^*_i<0$ for some $i$. Note that for such mechanisms the seller's revenue counterintuitively \emph{decreases} in $m_i$ for such $i$. It is a priori not evident that such mechanisms are suboptimal; however, it will follow from the proof of theorem \ref{main} that they indeed are. In section \ref{ext}, we discuss an alternative setting with inequality constraints in which one may restrict $\lambda$ to be nonnegative a priori.

Given a mechanism $M$, partition $[0,\vmax]^n$ into $n+1$ sets $W_0$, $W_1$,$\ldots$,$W_n$ where $W_i$, $i=1,\ldots,n$, are sets of valuations such that bidder $i$ gets the object, and $W_0$ is the set of valuations such that no one gets the object. Then, the value of\eqref{biconj} is equal to
\beq\label{transformedproblem}
\lambda m+\min\{\min\limits_{i=1,\ldots,n}\inf\limits_{v\in W_i(p)}(p_i(v_{-i})-\lambda v),\inf\limits_{v\in W_0(p)}(-\lambda v)\}
\eeq 
The sets $W_i$ are determined, up to tie-breaking, by the functions $p_i(\cdot)$. The tie-breaking rule may be specified by disjoint sets $O_i\subset\{v:v_i=p_i(v_{-i})\}$ such that $W_i=\{v:v_i>p_i(v_{-i})\}\cup O_i$, $i=1,\ldots,n$. If the indifference is between granting the object to a bidder and not granting it at all, it is always harmless to grant the object to someone, so we let sets $O_i$ satisfy $\cup_{i=1}^n O_i=\cup_{i=1}^n\{v:v_i=p_i(v_{-i})\}\setminus\cup_{i=1}^n\{v:v_i>p_i(v_{-i})\}$. Note that this also implies that the set  $W_0$ is defined with strict inequalities: $W_0=\{v:v_i<p_i(v_{-i})\mbox{ for all i}\}$. $W_0$ may be empty.

The following lemma shows that (1) one can safely ignore sets $O_i$ when optimizing over threshold functions $p$; (2) one can 
extend $W_i$ to certain \emph{intersecting} sets such that the infimum \eqref{transformedproblem} does not change. This extension is important for the proof of theorem \ref{main}.  

\begin{lemma}\label{getridofO}
For any $\lambda\in \mathbb{R}^n$ and any mechanism $M\in\mathcal{M}$, the value of \eqref{transformedproblem} is the same as the value of 
\beq
\lambda m+\min\{\min\limits_{i=1,\ldots,n}\inf\limits_{v\in W^{\geq}_i(p)}(p_i(v_{-i})-\lambda v),\inf\limits_{v\in W_0(p)}(-\lambda v)\}
\eeq
where $W_i^{\geq}(p):=\{v:v_i\geq p_i(v_{-i})\}$, $i=1,\ldots,n$.
\end{lemma}

\begin{proof}
See \hyperref[app]{Appendix}.\eop
\end{proof}

Lemma \ref{getridofO} is not trivial because sets  $W_i^{\geq}(p)$ are not merely closures of $W_i$. Whereas the projection of $W_i$ on $V_{-i}$ may be a strict subset of $V_{-i}$, the projection of $W_i^{\geq}(p)$ on $V_{-i}$ is the whole $V_{-i}=[0,\vmax]^{n-1}$ (recall that $p_i(v_{-i})\leq \vmax)$).

As using strong duality converts the inner problem into a maximization problem, one may collapse the outer and the inner problem into a single maximization problem.
Thus, the final reformulation of the problem \eqref{problem} may be stated as 

\begin{framed}
\noindent Choose measurable functions $p_i(v_{-i}):[0,\vmax]^{n-1}\to[0,\vmax]$ and $\lambda\in\mathbb{R}^n$ to maximize 
\beq\label{reducedproblem}
R(p,\lambda):=\lambda m+\min\{\min\limits_{i=1,\ldots,n}\inf\limits_{v\in W^{\geq}_i(p)}(p_i(v_{-i})-\lambda v),\inf\limits_{v\in W_0(p)}(-\lambda v)\}
\eeq
subject to supply constraint: $\forall i\neq j$ and $\forall v\in[0,\vmax]^n$, $v_i> p_i(v_{-i})\mbox{ implies }v_j\leq p_j(v_{-j})$,
where $W_i^{\geq}(p):=\{v:v_i\geq p_i(v_{-i})\}$, $i=1,\ldots,n$, $W_0=\{v:v_i<p_i(v_{-i})\mbox{ for all i}\}$.
\end{framed}
In \eqref{reducedproblem}, we abuse notation slightly by introducing the functional $R(p,\lambda)$ that plays a central role in the analysis to follow. 
Note that $\underline{R}(M)=\sup\limits_\lambda R(p,\lambda)$ where $p$ are the threshold functions representing mechanism $M$. Also note that even though $\sup\limits_\lambda R(p,\lambda)$ equals the worst-case expected revenue only for tuples of functions $p$ satisfying the supply constraint, $R(p,\lambda)$ is well-defined for all tuples of functions $p_i(v_{-i}):[0,\vmax]^{n-1}\to[0,\vmax]$, $i=1,\ldots,n$.  The proof of theorem \ref{main} involves evaluating $R(p,\lambda)$ at a tuple $p$ representing an \emph{infeasible} mechanism.

\section{Proof of theorem \ref{main}}\label{proof}
Suppose we are given a mechanism $M^0\in\mathcal{M}$, represented by threshold functions $p$. Compute an optimal $\lambda^*\in \mathbb{R}^n$ that maximizes $R(p,\lambda)$ given $p$ (such a lambda exists by Lemma \ref{optexistlambda}). We will construct a linear score auction $LSA^0$ with threshold functions $\hat{p}$ such that either $R(\hat{p},\lambda^*)\geq R(p,\lambda^*)$ (in \textbf{Grand case I} below) or $R(\hat{p},\lambda^{**})\geq R(p,\lambda^{**})\geq R(p,\lambda^{*})$ for some $\lambda^{**}$ (in \textbf{Grand case II}). When lambda is reoptimized for the $LSA^0$, the value of $R(\hat{p},\lambda)$ will be even (weakly) higher. 

The construction of the dominating LSA depends on $\lambda^*$ and is qualitatively different depending on whether $\lambda^*$ has negative components or not. 

\textbf{Grand case I.} $\lambda^*_i\geq 0$ for all $i=1,\ldots,n$. 

Note that for every $v_{-i}\in[0,\vmax]^{n-1}$, the value profile $(\vmax,v_{-i})\in W^{\geq}_i(p)$.  Given this fact and the fact that $\lambda^*_i\geq 0$, for every $i$ 
$$\inf\limits_{v\in W^{\geq}_i(p)}(p_i(v_{-i})-\lambda^* v)=\inf\limits_{v_{-i}\in[0,\vmax]^{n-1}}(p_i(v_{-i})-\lambda^*_{-i}v_{-i}-\lambda^*_i\vmax).$$
Thus, \eqref{reducedproblem} may be simplified further to 
\beq\label{reducedproblem1}
R(p,\lambda^*)=\lambda^* m+\min\{\min\limits_{i=1,\ldots,n}\inf\limits_{v_{-i}\in[0,\vmax]^{n-1}}(p_i(v_{-i})-\lambda^*_{-i}v_{-i}-\lambda^*_i\vmax),\inf\limits_{v\in W_0(p)}(-\lambda^* v)\}
\eeq

\textbf{Step 1.} The main idea is to replace each threshold function $p_i(\cdot)$ with a simpler function in such a way that the worst-case revenue \eqref{reducedproblem1} does not decrease. 
To this end, given a function $p_i(v_{-i})$ and $\lambda^*$, compute 
\[b_i:=\inf\limits_{w\in[0,\vmax]^{n-1}}(p_i(w)-\lambda^*_{-i} w)\]
and then
for every $v\in[0,\vmax]^{n-1}$ define 
\beq\label{tildepdef}
\tilde{p}_i(v_{-i}):=\max\{\lambda^*_{-i}v_{-i}+b_i,0\}.
\eeq

$\tilde{p}_i(v_{-i})$ may be thought of as the supporting hyperplane to the graph of $p_i(\cdot)$ with a given slope $\lambda^*_{-i}$ (with an appropriate truncation on the boundary of $[0,\vmax]^{n-1}$). Such a hyperplane exists even if $p_i(\cdot)$ is not convex because it is not a supporting hyperplane at a given point, but rather a supporting hyperplane with a given slope.

We replace each function $p_i(\cdot)$ with $\tilde{p}_i(\cdot).$ Even though the tuple $\tilde{p}$ typically violates the supply constraint, $R(\tilde{p},\lambda^*)$ can still be evaluated. 

\begin{proposition}\label{tildebetter}
$R(\tilde{p},\lambda^*)\geq R(p,\lambda^*)$.
\end{proposition}
\begin{proof}
It is sufficient to prove that each inner infimum in \eqref{reducedproblem1} is weakly greater under $\tilde{p}$ than 
under $p$. 

First, consider $\inf\limits_{v\in W_0(p)}(-\lambda v)$. Note that one can take $w=v_{-i}$ in the inner minimization in \eqref{tildepdef} and hence $p_i(\cdot)\geq \tilde{p}_i(\cdot)$. 
Thus, $W_0(\tilde{p})\subseteq W_0(p)$ and $\inf\limits_{v\in W_0(\tilde{p})}(-\lambda v)\geq \inf\limits_{v\in W_0(p)}(-\lambda v)$.

Now we prove that for all $i$, $\inf\limits_{v_{-i}\in[0,\vmax]^{n-1}}(p_i(v_{-i})-\lambda^*_{-i}v_{-i})=\inf\limits_{v_{-i}\in[0,\vmax]^{n-1}}(\tilde{p}_i(v_{-i})-\lambda^*_{-i}v_{-i})$. In this equation, $LHS\geq RHS$ because $p_i(\cdot)\geq \tilde{p}_i(\cdot)$. However, $LHS\leq RHS$ as well because by \eqref{tildepdef}, for every $v_{-i}$, $\tilde{p}_i(v_{-i})-\lambda^*_{-i}v_{-i}\geq \inf\limits_{w\in[0,\vmax]^{n-1}}(p_i(w)-\lambda^*_{-i} w)=LHS$. Hence, $LHS=RHS$.
\eop
\end{proof}

\begin{figure}[h!]
\centering
\begin{tikzpicture}[>=triangle 45, xscale=4,yscale=4]
\draw[->] (0,0) -- (1.1,0) node[below] {$v_1$};
\draw[->] (0,0) -- (0,1.1) node[left] {$v_2$};
\draw[dashed] (1,0)--(1,1);
\draw[dashed] (0,1)--(1,1);
\node[below] at (1,0){$\vmax$};
\node[left] at (0,1){$\vmax$};

\draw[thick, blue, domain=0:1]
plot (\x, {0.5+pow((\x-0.5),2)});

\draw[thick,red] plot [smooth,tension=1] coordinates{(0.5, 0.25)(0.75,0.5)(1,0.6)};

\draw[thick, red, domain=0.2:0.3]
plot (\x, {0.25-25*pow((\x-0.3),2)});

\draw[thick,red] (1,0.6)--(1,1); 
\draw[->](1.2,0.5)--(1.6,0.5);
\end{tikzpicture}
\begin{tikzpicture}[>=triangle 45, xscale=4,yscale=4]
\draw[->] (0,0) -- (1.1,0) node[below] {$v_1$};
\draw[->] (0,0) -- (0,1.1) node[left] {$v_2$};
\draw[dashed] (1,0)--(1,1);
\draw[dashed] (0,1)--(1,1);
\node[below] at (1,0){$\vmax$};
\node[left] at (0,1){$\vmax$};

\draw[dashed, blue, domain=0:1]
plot (\x, {0.5+pow((\x-0.5),2)});
\draw[thick, blue, domain=0:1]
plot (\x, {0.1875+0.5*\x});

\draw[dashed,red] plot [smooth,tension=1] coordinates{(0.5, 0.25)(0.75,0.5)(1,0.6)};

\draw[dashed, red, domain=0.2:0.3]
plot (\x, {0.25-25*pow((\x-0.3),2)});

\draw[thick,red] (0.05254860746,0)--(1,1);
\draw[dashed, red, domain=0.2:0.2788897449]
plot (\x, {0.25-25*pow((\x-0.3),2)});
\draw[dashed,red] (1,0.6)--(1,1); 

\end{tikzpicture}
\caption{Step 1. Taking  the transformation $p\to\tilde{p}$ for $n=2$. On the left display an arbitrary $p_1(v_2)$ is in red and an arbitrary $p_2(v_1)$ is in blue. On the right display, the resulting $\tilde{p}$-functions are solid and the initial $p$-functions are dashed. $\lambda_1=0.5$ and $\lambda_2\approx 1.055$.}\label{fig:step1}
\end{figure}
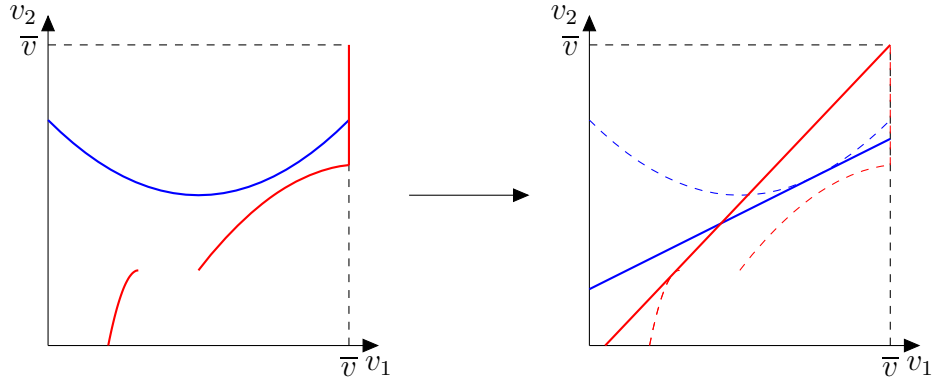

\textbf{Step 2.} In this step, given $\tilde{p}$ we find an LSA that gives revenue no less than $R(\tilde{p},\lambda^*)$ and hence, $R(p,\lambda^*)$. 

To this end, we use a fixed point of the map $v\to \tilde{p}(v)$ (The map defined for each $i$ by $v_i=\tilde{p}_i(v_{-i})$.) Indeed, this is a continuous map from $[0,\vmax]^n$ to itself, and thus, by Brouwer Theorem a fixed point exists. Denote the set of fixed points by $V^*\neq\emptyset$.

We proceed by considering three mutually exclusive and exhaustive cases concerning the structure of $V^*$. 

\textbf{Case 1.} There exists an interior fixed point, but there are no fixed points with at least one of coordinates equal to zero, i.e. $V^*\cap (0,\vmax)^n\neq\emptyset$ and $v^*>0$ for all $v^*\in V^*$.  

\textbf{Case 2.} There exists a fixed point with at least one coordinate equal to zero, i.e. $\exists v^*\in V^*: v^*_i=0$ for some $i\in\{1,\ldots,n\}$.

\textbf{Case 3.} There are no interior fixed points and no fixed points with at least one of coordinates equal to zero, i.e. for all $v^*\in V$ $v^*>0$ and $\exists i: v^*_i=\vmax$. 

\vspace{1.5em}
\textbf{Case 1.} There exists an interior fixed point, but there are no fixed points with at least one of coordinates equal to zero, i.e. $V^*\cap (0,\vmax)^n\neq\emptyset$ and $v^*>0$ for all $v^*\in V^*$. 

Take any interior fixed point $v^*$. It satisfies the equations  $v_i=\tilde{p}_i(v_{-i})$ for all $i$. As $v^*$ is interior, $v^*>0$, and thus, by \eqref{tildepdef}, it satisfies the system of linear equations
\beq\label{system}
A(\lambda^*)v=b,
\eeq
where the matrix $A(\lambda)$ of size $n$ is given by 
\beq\label{Adef}
a_{ij}(\lambda)=
\begin{cases}
1, &  i=j;\\ 
-\lambda_j,  & i\neq j,
\end{cases}
\eeq
while $b_i:=\inf\limits_{w\in[0,\vmax]^{n-1}}(p_i(w)-\lambda^*_{-i} w)$, $i=1,\ldots,n$.

Next Lemma shows that, given the supposition in Case 1, matrix $A(\lambda^*)$ should be of full rank.

\begin{lemma}\label{structsetoffixedpoints}
If matrix $A(\lambda^*)$ is singular and there exists a positive fixed point $v^*\in V^*$, then $V^*$ also contains a fixed point with a zero coordinate. 
\end{lemma}

\begin{proof}
Suppose matrix $A$ is singular. Then it follows from basic linear algebra that the set of solutions to the system \eqref{system} includes the affine subspace $\{v^*+a \hat{v}|a\in\mathbb{R}\}$ where $\hat{v}$ is some nonzero solution to the homogeneous system $Av=0$. 
By subtracting rows of $A$ it is easy to prove that, given $\lambda^*_i\geq 0$, if $\hat{v}$ is a solution to $Av=0$, all coordinates of $\hat{v}$ must be of the same sign. Without loss of generality, take the coordinate-wise positive solution $\hat{v}$ and consider $a=-\min\limits_i(v^*_i/\hat{v}_i)$. Then,  $\tilde{v}$ where $\tilde{v}_i=v^*_i-\min\limits_i(v^*_i/\hat{v}_i)\hat{v}_i$ is a solution to \eqref{system} which is nonnegative in all coordinates and equal to zero in at least one coordinate. Also, it is less or equal to $v^*$ in every coordinate, so $\tilde{v}\in[0,\vmax]^n$. Hence, $\tilde{v}$ is the desired element of $V^*$ that has a zero coordinate.  
\eop
\end{proof}

Thus, as there are no fixed points with at least one of coordinates equal to zero by assumption in Case 1, matrix $A(\lambda^*)$ must be of full rank. Hence, $V^*\cap(0,\vmax)^n$ is a singleton.

Now take the unique point $v^*\in V^*\cap(0,\vmax)^n$ and consider $LSA(v^*)$. Denote by $\hat{p}$ the threshold functions corresponding to $LSA(v^*)$ (they are given by \eqref{p-lsa}).  We prove that $LSA(v^*)$ gives revenue no smaller than the the tuple of threshold functions $\tilde{p}$.

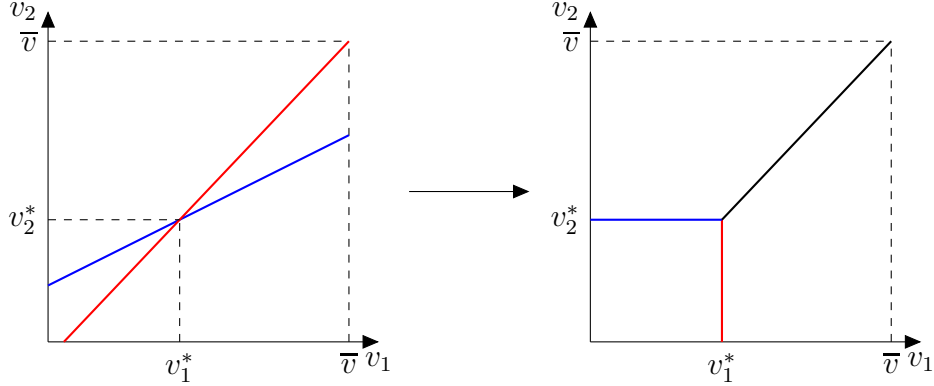
\begin{figure}[h!]
\centering
\begin{tikzpicture}[>=triangle 45, xscale=4,yscale=4]
\draw[->] (0,0) -- (1.1,0) node[below] {$v_1$};
\draw[->] (0,0) -- (0,1.1) node[left] {$v_2$};
\draw[dashed] (1,0)--(1,1);
\draw[dashed] (0,1)--(1,1);
\node[below] at (1,0){$\vmax$};
\node[left] at (0,1){$\vmax$};

\draw[thick, blue, domain=0:1]
plot (\x, {0.1875+0.5*\x});

\draw[thick,red] (0.05254860746,0)--(1,1);

\draw[dashed](0.437406,0)--(0.437406,0.406203);
\draw[dashed](0,0.406203)--(0.437406,0.406203);
\node[below] at (0.437406,0){$v^*_1$};
\node[left] at (0,0.406203){$v^*_2$};

\draw[->](1.2,0.5)--(1.6,0.5);
\end{tikzpicture}
\begin{tikzpicture}[>=triangle 45, xscale=4,yscale=4]
\draw[->] (0,0) -- (1.1,0) node[below] {$v_1$};
\draw[->] (0,0) -- (0,1.1) node[left] {$v_2$};
\draw[dashed] (1,0)--(1,1);
\draw[dashed] (0,1)--(1,1);
\node[below] at (1,0){$\vmax$};
\node[left] at (0,1){$\vmax$};

\draw[red,thick](0.437406,0)--(0.437406,0.406203);
\draw[blue,thick](0,0.406203)--(0.437406,0.406203);
\draw [thick] (0.437406,0.406203)--(1,1);
\node[below] at (0.437406,0){$v^*_1$};
\node[left] at (0,0.406203){$v^*_2$};
\end{tikzpicture}

\caption{Step 2. Taking  the transformation $\tilde{p}\to \hat{p}$ (LSA($v^*$)) for $n=2$. On the left display the fixed point $v^*$ of the map $v\to\tilde{p}(v)$ is identified. On the right display, the auction $LSA(v^*)$ is shown.}\label{fig:step2}
\end{figure}

\begin{proposition}\label{LSAbetter}
$R(\hat{p},\lambda^*)\geq R(\tilde{p},\lambda^*)$ in Case 1.
\end{proposition}
\begin{proof}
Again, it is sufficient to prove that each inner infimum in \eqref{reducedproblem1} is weakly greater under $\hat{p}$ than under $\tilde{p}$.

First, consider $\inf\limits_{v_{-i}\in[0,\vmax]^{n-1}}(\tilde{p}_i(v_{-i})-\lambda^*_{-i}v_{-i})$. It is sufficient to show that $\hat{p}\geq\tilde{p}$ at every point. 
To this end, note that by definition of $v^*$ and $\tilde{p}$ \eqref{tildepdef}, $\tilde{p}_i(v_{-i})$ may be rewritten as
\beq\label{fp-representation}
\tilde{p}_i(v_{-i})=\max\{v_i^*+\lambda_{-i}^*(v_{-i}-v^*_{-i}),0\}.
\eeq

Thus, we need to prove the inequality
\beq\label{ineqtoprove}
\max\{v_i^*+\lambda_{-i}^*(v_{-i}-v^*_{-i}),0\}\leq 
v_i^*+(\vmax-v^*_i)\max\left\{0,\max\limits_{j\neq i}\frac{v_{j}-v^*_{j}}{\vmax-v^*_{j}}\right\}
\eeq
where the right-hand side is the expression for  $\hat{p}$. Note that because $v^*\in (0,\vmax)^n$, no bidder is excluded from the auction LSA($v^*$).

We will prove a slightly more general lemma that will be also of help later.

\begin{lemma}\label{LSAfunctionsgreater}
Suppose $\chi_{-i}^*\in \mathbb{R}^{n-1}_+$ and $v^*\in[0,\vmax]^n$ are such that $v_i^*+\chi^*_{-i}(\vmax\i_{n-1}-v^*_{-i})\leq\vmax$ for all $i$. Denote by $I(v^*)$ the set of indices $i$ such that $v^*_i<\vmax$. Then, for all $i$ and for all $v_{-i}\in[0,\vmax]^{n-1}$,
\beq\label{LSAfunctionsgreater1}
\max\{v_i^*+\chi^*_{-i}(v_{-i}-v^*_{-i}),0\}\leq 
v_i^*+(\vmax-v^*_i)\max\left\{0,\max\limits_{j\in I(v^*)\setminus\{i\}}\frac{v_{j}-v^*_{j}}{\vmax-v^*_{j}}\right\}.
\eeq 
\end{lemma}
\begin{proof}
First, note that if the LHS of \eqref{LSAfunctionsgreater1} is 0, the inequality holds. 

Suppose it is not zero. Then note that
\beq\label{22} 
v_i^*+\chi^*_{-i}(v_{-i}-v^*_{-i})\leq v_i^*+\sum\limits_{j\in I(v^*)\setminus\{i\}}\chi^*_j(v_j-v^*_j)
\eeq
because $\chi^*_j(v_j-v^*_j)\leq 0$ for $j\notin I(v^*)$. 

Now consider an auxiliary linear optimization problem in $|I(v^*)|-1$ variables 
\begin{align}
\max\limits_{\chi_j\geq 0} &\mbox{ } v_i^*+ \sum_{j\in I(v^*)\setminus\{i\}} \chi_j(v_j-v^*_{j})\label{auxproblem}\\
\mbox{s.t. } & v^*_i+\sum_{j\in I(v^*)\setminus\{i\}} \chi_j(\vmax-v_j^*)\leq \vmax \label{auxconstr}
\end{align}
Whenever $(v_j-v^*_{j})>0$ for at least one $j\in I(v^*)\setminus\{i\}$, the problem \eqref{auxproblem} is solved by putting all weight on the variable with the highest ``bang for the buck'', i.e.
$$\chi_j=\frac{\vmax-v^*_i}{\vmax-v^*_j}$$ for 
$$j\in\arg\max_{j\in I(v^*)\setminus\{i\}}\frac{v_j-v_j^*}{\vmax-v_j^*}$$ and
$\chi_j=0$ for other $j$. 
If $(v_j-v^*_{j})\leq 0$ for all $j$, one should set $\chi_j=0$ for all $j$. Hence, the maximum value of problem \eqref{auxproblem} is equal precisely to the RHS of \eqref{LSAfunctionsgreater1}. On the other hand, 
$\chi^*$ is feasible for the problem \eqref{auxproblem} as \eqref{auxconstr} holds for $\chi^*$ due to the lemma's supposition, and the value of the objective function at $\chi^*$ is precisely the RHS of \eqref{22}.\eop
\end{proof}

Now note that the supposition of lemma \ref{LSAfunctionsgreater} holds for $\chi^*=\lambda^*$ and the unique $v^*\in V^*$ due to the fact that $\tilde{p}_i(\vmax\i_{n-1})\leq\vmax$. Also, $I(v^*)=\{1,2,\ldots,n\}$. Thus, applying lemma \ref{LSAfunctionsgreater} to $\chi^*=\lambda^*$ and $v^*$   yields the desired inequality \eqref{ineqtoprove}.

Second, consider  $\inf\limits_{v\in W_0(p)}(-\lambda^* v)$. We will prove that $\inf\limits_{v\in W_0(\hat{p})}(-\lambda^* v)=\inf\limits_{v\in W_0(\tilde{p})}(-\lambda^* v)$.

Because, as we proved above,  $\hat{p}\geq \tilde{p}$, $W_0(\tilde{p})\subseteq W_0(\hat{p})$ and hence, $LHS\leq RHS$. Note that $LHS=-\lambda^*v^*$ by the construction of $LSA(v^*)$ and positivity of $\lambda$. To prove that $LHS\geq RHS$, we prove that $v^*$ is a limit point of the set $W_0(\tilde{p})$. 

Consider again the system \eqref{system}. Because $v^*$ is the unique solution to it, matrix $A$ must be of full rank, and so $A^{-1}$ exists. Consider the sequence $v^k=A^{-1}(b-\i_{n}/k)$ where $\i_n$ is a vector of all ones of dimension $n$. $v^k$ may not lie in $(0,\vmax)^n$ for all $k$. However, $v^k$ approaches $v^*$. Hence, it lies, after some $k$, in $(0,\vmax)^n$ and, as $Av^k<b$ coordinate-wise, in $W_0(\tilde{p})$. 
\eop
\end{proof}

\vspace{1.5em}
\textbf{Case 2.} There exists a fixed point with at least one coordinate equal to zero, i.e. $\exists v^*\in V^*: v^*_i=0$ for some $i$.

In this case, we will again prove that LSA($v^*$) (again, call the respective threshold functions $\hat{p}$) provides a weakly higher revenue than $\tilde{p}$. The major novelty in this case is that because  $v^*_i=0$, $v^*$ may not satisfy the system \eqref{system} (recall that $v_i^*=\tilde{p}_i(v_{-i}^*)=\max\{\lambda^*_{-i}v_{-i}^*+b_i,0\}$ and so we only have that $v_i^*\geq \lambda^*_{-i}v_{-i}^*+b_i$ if $v_i^*=0$). Due to this, the representation \eqref{fp-representation} won't hold for $\tilde{p}$ and one is not able to apply Lemma \ref{LSAfunctionsgreater} directly.

To deal with this issue, we introduce auxiliary functions 
\beq \label{p_aux}
p^{aux}_i(v_{-i}):=\max\{v_i^*+k_i\lambda^*_{-i}(v_{-i}-v^*_{-i}),0\},
\eeq
where $k_i$ are individual-specific coefficients set in a way that the functions $p^{aux}_i$ satisfy conditions of Lemma \ref{LSAfunctionsgreater}, so $\hat{p}\geq p^{aux}$, and simultaneously ensure that $p^{aux}\geq \tilde{p}$. 

That is, we prove the inequality $\hat{p}\geq \tilde{p}$ not directly, as in Case 1, but, if necessary, via the chain $\hat{p}\geq p^{aux}\geq \tilde{p}$.

\begin{proposition}\label{LSAbetter2}
$R(\hat{p},\lambda^*)\geq R(\tilde{p},\lambda^*)$ in Case 2.
\end{proposition}

\begin{proof}
See \hyperref[app]{Appendix}\eop
\end{proof}

\vspace{1.5em}
\textbf{Case 3.} There are no interior fixed points and no fixed points with at least one of coordinates equal to zero, i.e. for all $v^*\in V$ $\exists i: v^*_i=\vmax$, and $v^*>0$. 

This is the hardest case to consider due to problems with comparing $\inf\limits_{v\in W_0(p)}(-\lambda^* v)$ for $\hat{p}$ and $\tilde{p}$. On the other hand, $W_0(\hat{p})\neq \emptyset$ as $v^*>0$, so the comparison via  \eqref{W_0comparison} cannot be made. On the other hand, as $v^*$ is not interior, a sequence approaching $v^*$ and lying in $W_0(\tilde{p})$ might not exist, so the argument made in Case 1 may not be valid. 

Whether or not $v^*$ is a limit point of the set $W_0(\tilde{p})$ now depends on the geometry of the $\tilde{p}_i$ functions in a neighborhood of the point, which is  governed by the properties of the matrix $A(\lambda^*)$. 

In two dimensions, this is illustrated in figure \ref{fig:2geom}.

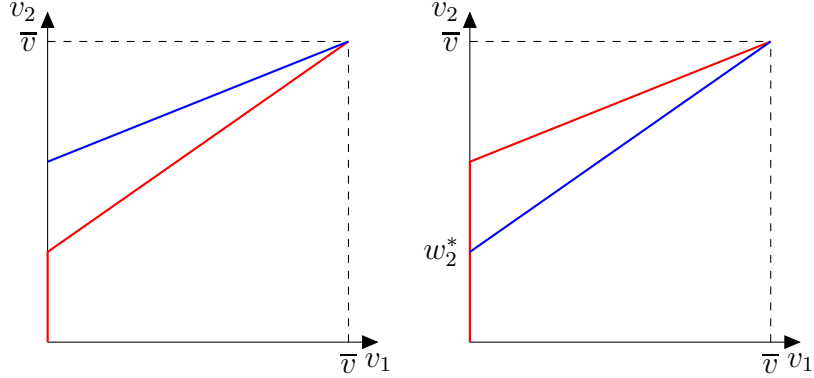
\begin{figure}[h!]
\centering
\begin{tikzpicture}[>=triangle 45, xscale=4,yscale=4]
\draw[->] (0,0) -- (1.1,0) node[below] {$v_1$};
\draw[->] (0,0) -- (0,1.1) node[left] {$v_2$};
\draw[dashed] (1,0)--(1,1);
\draw[dashed] (0,1)--(1,1);
\node[below] at (1,0){$\vmax$};
\node[left] at (0,1){$\vmax$};

\draw[red,thick] (0,0)--(0,0.3)--(1,1);
\draw[blue,thick] (0,0.6)--(1,1);

\end{tikzpicture}
\begin{tikzpicture}[>=triangle 45, xscale=4,yscale=4]
\draw[->] (0,0) -- (1.1,0) node[below] {$v_1$};
\draw[->] (0,0) -- (0,1.1) node[left] {$v_2$};
\draw[dashed] (1,0)--(1,1);
\draw[dashed] (0,1)--(1,1);
\node[below] at (1,0){$\vmax$};
\node[left] at (0,1){$\vmax$};

\draw[red,thick] (0,0)--(0,0.6)--(1,1);
\draw[blue,thick] (0,0.3)--(1,1);
\node[left] at (0,0.3){$w^*_2$};

\end{tikzpicture}

\caption{Two types of geometry of threshold functions around the fixed point $(\vmax,\vmax)$. In both pictures, $\tilde{p}_1(v_2) $ is in red and $\tilde{p}_2(v_1)$ is in blue. In the left picture, $\det(A)>0$ while in the right picture, $\det(A)<0$. Note that in the right picture there exists another fixed point $w^*$ with $w^*_1=0$.}\label{fig:2geom}
\end{figure} 

It turns out that in general the type of geometry is determined by the sign of $\det(A)$. Fortunately, $\det(A)$ can be computed explicitly. The formula, along with other properties of $A$ to be used, is given in the following lemma.
\begin{lemma}\label{Aproperties}%
For a matrix $A(\lambda)$, given by \eqref{Adef}, with $\lambda\geq 0$, the following statements hold:
\begin{enumerate}
\item $\det(A)=\left(1-\sum\limits_{i=1}^n \frac{\lambda_i}{1+\lambda_i}\right)\prod\limits_{i=1}^n(1+\lambda_i)$.
\item If $\det(A)>0$, all elements of $A^{-1}$ are nonnegative. 
\item If $v^*$ is the unique solution to $Av=b$, then $\lambda v^*=\frac{\sum\limits_{i=1}^n \frac{\lambda_i}{1+\lambda_i}b_i}{1-\sum\limits_{i=1}^n \frac{\lambda_i}{1+\lambda_i}}.$
\item The rank of $A$ is at least $n-1$.
\end{enumerate}
\end{lemma}

Now consider two subcases separately.

\vspace{1em}
\emph{Subcase 1.} $\sum\limits_{i=1}^n \frac{\lambda^*_i}{1+\lambda^*_i}<1$. 

In this subcase, $\det(A)>0$ and $\tilde{p}_i$ exhibit the same type of geometry as in figure \ref{fig:2geom}, left.

In this case, we show that, as usual, $LSA(v^*)$ dominates $\tilde{p}$ (one can take any $v^*\in V^*$). For bidders with $v^*_i<\vmax$, the proof that $\hat{p}_i\geq \tilde{p}_i$ is entirely as in Cases 1 and 2 (by Lemma \ref{LSAfunctionsgreater}), since $v^*>0$. If $v^*_i=\vmax$, $\hat{p}_i\equiv\vmax\geq \tilde{p}_i$. 

Now consider $\inf\limits_{v\in W_0(p)}(-\lambda^* v)$. As in Case 1, $\inf\limits_{v\in W_0(\hat{p})}(-\lambda^* v)=-\lambda^*v^*$ by the definition of threshold functions for the linear score auction.  We show that $v^*$ is a limit point of the set $W_0(\tilde{p})$ so $\inf\limits_{v\in W_0(\hat{p})}(-\lambda^* v)\geq \inf\limits_{v\in W_0(\tilde{p})}(-\lambda^* v)$. 

Note that because $v^*>0$, $v^*$ satisfies the system \eqref{system} $A(\lambda^*)v=b$. Consider the sequence $v^k=A^{-1}(b-\i_n/k)$. ($A^{-1}$ exists as $\det(A)>0$.) This sequence satisfies $Av^k<<b$. Most importantly, by Lemma \ref{Aproperties}, part 2, $v^k=v^*-A^{-1}\i_n/k\leq v^*$ as all elements of $A^{-1}$ are nonnegative. Hence, $v^k$ lies in $[0,v]^n$ for all $k$ sufficiently large, and thus $v^k\in W_0(\tilde{p})$. Thus, $v^k\to v^*$ is the desired sequence showing that $v^*$ is a limit point of  $W_0(\tilde{p})$.

\vspace{1em}
\emph{Subcase 2.} $\sum\limits_{i=1}^n \frac{\lambda^*_i}{1+\lambda^*_i}\geq 1$.  In this subcase, $\det(A)\leq 0$ and $\tilde{p}_i$ exhibit the same type of geometry as in figure \ref{fig:2geom}, right.

We will show that in this subcase, there must, along with $v^*$, exist a fixed point $w^*\in V^*$ such that at least one coordinate of $w^*$ is equal to zero. Hence, this subcase is incompatible with the premise in Case 3. In two dimensions, the existence of the fixed point $w^*$ is illustrated in figure \ref{fig:2geom}, right.

Suppose that $\sum\limits_{i=1}^n \frac{\lambda^*_i}{1+\lambda^*_i}=1$. Then $\det(A)=0$ and thus the desired fixed point $w^*$ exists by Lemma \ref{structsetoffixedpoints}. 

From now on, suppose that $\sum\limits_{i=1}^n \frac{\lambda^*_i}{1+\lambda^*_i}>1$.

First, we show that when constructing a desired fixed point $w^*$ , one can, in a certain sense, ignore all bidders with $v^*_i<\vmax$ (where $v^*$ is the initial fixed point considered in Case 3). Indeed, by the premise of Case 3 there exists bidder $j$ such that $v^*_j=\vmax$. Thus, 
\beq
\vmax=\tilde{p}_j(v^*_{-j})=\lambda^*_{-j}v^*_{-j}+b_j\leq\lambda^*_{-j}+b_j\vmax\i_{n-1}\leq \vmax,
\eeq
where the first inequality is due to $\lambda^*\geq 0$ and the second is due to $\tilde{p}_j(\vmax)\leq\vmax$. 
Hence, all inequalities hold as equalities and thus $\lambda^*_{-j}(\vmax\i_{n-1}-v^*_{-j})=0$. But every term in this sum is nonnegative, so we must have $\lambda^*_i(\vmax-v^*_i)=0$ for all $i\neq j$. Thus, for all $i$ with $v^*_i<\vmax$ we must have $\lambda^*_i=0$. But this means that the values of all such bidders do not impact $\tilde{p}_j$ for $j$ such that $v^*_j=\vmax$. (Call this set of bidders $J_{max}$ again.) Hence, it is sufficient to find a desired fixed point restricted to the set of bidders $J_{max}$. To see this formally, denote by $[0 \mbox{  } \tilde{w}^*]$ a vector consisting of zeros at coordinates $j\notin J_{max}$ and of $\tilde{w}^*_j$ for $j\in J_{max}$. Then, if  a point $\tilde{w}^*\in[0,\vmax]^{|J_{max}|}$, satisfies the system $\tilde{w}_i=\tilde{p}_i([0\mbox{  } \tilde{w}]_{-i})$  for all $i\in J_{max}$, with some $\tilde{w}^*_i=0$, the point $w^*$ given by 
\[w^*_j=\begin{cases}
\tilde{w}^*_j, & j\in J_{max};\\
\tilde{p}_j([0\mbox{  } \tilde{w}^*]), & j\notin J_{max} 
\end{cases}\]
will satisfy the whole system $v_i=\tilde{p}_i(v_{-i})$  for all $i\in\{1,\ldots,n\}$, with some $w^*_i=0$.

Hence, it is with loss of generality to consider $v^*=\vmax\i_n$. Since $\vmax\i_n$ is a positive fixed point, it must satisfy $Av=b$, so for all $i$ we get 
\beq\label{biexpr}
b_i=(1-\sum_{j\neq i}\lambda^*_j)\vmax.
\eeq

The main idea to construct a fixed point $w^*$ with $w^*_i=0$ whenever $b_i<0$. So we first show that such an $i$ is guaranteed to exist in this subcase. This is achieved by the following lemma:

\begin{lemma}\label{bnegative}
Suppose $\lambda_1,\ldots,\lambda_n$ are nonnegative numbers satisfying 
\beq\label{Sums}
\sum_{i\neq k}\lambda_i\leq 1
\eeq
for every $k\in \{1,\ldots, n\}$.
Then, $\sum_{i=1}^n\frac{\lambda_i}{1+\lambda_i}\leq 1$. Moreover, if at least one of the inequalities \eqref{Sums} is strict, then $\sum_{i=1}^n\frac{\lambda_i}{1+\lambda_i}< 1$.
\end{lemma}

Indeed, suppose $b_i\geq 0$ for all $i$. Then, by \eqref{biexpr}, \[\sum_{i\neq k}\lambda^*_i\leq 1\]
for every $k\in \{1,\ldots, n\}$, so $\lambda^*_i$ satisfy the premise of Lemma \ref{bnegative}. Hence, $\sum_{i=1}^n\frac{\lambda_i}{1+\lambda_i}\leq 1$ but this contradicts our premise that $\sum\limits_{i=1}^n \frac{\lambda^*_i}{1+\lambda^*_i}>1$. Hence, there exists $j$ such that $b_j<0$.
From now on, enumerate bidders in such a way that $b_1\geq b_2\geq\ldots\geq b_n$.

If $b_1\leq 0$, $0\cdot\i_n$ is the desired fixed point. Indeed, then we have $0=\max\{b_i,0\}=\tilde{p}_i(0\i_{n-1})$ for all $i$. So assume $b_1>0$. Hence, there exists $s^*\in\{1,2,\ldots, n-1\}$ such that $b_s\geq 0$ for all $s\leq s^*$ and $b_s<0$ for all $s>s^*$. 

We will construct a fixed point $w^*$ such that $w^*_i=0$ for $i>s^*$. Denote by $A^{restr}$ the restricted matrix formed by first $s^*$ rows and columns of $A$. Denote by $b^{restr}$ the vector formed by the first $s^*$ elements of $b$. 

Because, $b_1>0$, $b_i\geq 0$ for $i=2,\ldots, s^*$, by Lemma \ref{bnegative} we must have $\sum\limits_{i=1}^{s^*}\frac{\lambda^*_i}{1+\lambda^*_i}<1$. Thus, by Lemma \ref{Aproperties}, part 1, matrix $A^{restr}$ is of full rank. 

So define $w^*$ by 
\[w^*_j=
\begin{cases}
((A^{restr})^{-1}b^{restr})_j, & j\leq s^*;\\
0, & j>s^*.
\end{cases}\]

\begin{proposition}\label{newfixedpoint}
$w^*\in V^*$ and $w^*_i=0$ for some $i$.
\end{proposition}

\begin{proof}
See \hyperref[app]{Appendix}.\eop
\end{proof}

We have constructed a fixed point $w^*\in V^*$ with a zero coordinate which shows that \emph{Subcase 2} is impossible within \textbf{Case 3.}    

This finishes the consideration of \textbf{Grand case I}.

\vspace{1.5em}
\textbf{Grand case II.} $\lambda^*_i< 0$ for some $i=1,\ldots,n$.

Here, given $p$ and $\lambda^*$, we pinpoint a pair of another mechanism and a nonegative vector $\lambda^{**}$ which gives a strictly higher $R$ thus reducing the problem to Grand Case I.

Denote the set  $\{i|\lambda^*_i\geq 0\}$ by $+(\lambda^*)$ and the set $\{i|\lambda^*_i< 0\}$ by $-(\lambda^*)$. Define $q:=|\{i|\lambda^*_i< 0\}|$. Redenote threshold functions by $p_i(v^+,v^-)$  where the profile of values of bidders from $+(\lambda^*)$ is $v^+$ and the profile of values of bidders from $-(\lambda^*)$ is $v^-$. 

For all bidders $i\in +(\lambda^*)$, define $p^{new}_i(v_{-i}):=p_i(v^+,0)$ and for all bidders $i\in -(\lambda^*)$ define $p^{new}_i(v_{-i}):=\vmax$. The transformation works as if bidders from $-(\lambda^*)$ are removed from the auction. 

\begin{lemma}\label{pnewbetter}
$R(p^{new},\lambda^*)\geq R(p,\lambda^*)$.
\end{lemma}

\begin{proof}
Again, we consider separately different infima in \eqref{reducedproblem}. 

First, note that all infima $\inf\limits_{v\in W^{\geq}_i(p)}(p_i(v_{-i})-\lambda^* v)$ for $i\in-(\lambda^*)$ are increased as $p^{new}_i(v_{-i})=\vmax\geq p_i(v_{-i})$ (the increase will be both due to direct rise in $p_i$ and the shrinkage of the set $W^{\geq}_i(p)$.)

Second, consider  $\inf\limits_{v\in W^{\geq}_i(p^{new})}(p^{new}_i(v_{-i})-\lambda^* v)$ for $i\in+(\lambda)$. As before, $\inf\limits_{v\in W^{\geq}_i(p)}(p_i(v_{-i})-\lambda v)=\inf\limits_{v_{-i}\in[0,\vmax]^{n-1}}(p_i(v^+_{-i},v_{-i}^-)-\lambda^*_{-i}v_{-i}-\lambda^*_i\vmax)$. 
For $p^{new}$, this infimum is equal to 
$$\inf\limits_{v_{-i}\in[0,\vmax]^{n-1}}(p_i(v^+_{-i},0)-\lambda^{*+}_{-i}v^+_{-i}-\lambda^{*-}_{-i}v^-_{-i}-\lambda^*_i\vmax).$$ But note that the minimand in latter infimum is nondecreasing in $v_{-i}^-$, due to the sign of $\lambda_{-i}^{*-}$. Hence, $\inf\limits_{v_{-i}\in[0,\vmax]^{n-1}}(p_i(v^+_{-i},0)-\lambda^{*+}_{-i}v^+_{-i}-\lambda^{*-}_{-i}v^-_{-i}-\lambda^*_i\vmax)=\inf\limits_{v_{-i}\in[0,\vmax]^{k-1}}(p_i(v^+_{-i},0)-\lambda^{*+}_{-i}v^+_{-i}-\lambda^*_i\vmax)=:R_0$.
On the other hand, the point $(v^+_{-i},0)$ is feasible in the minimization in $\inf\limits_{v_{-i}\in[0,\vmax]^{n-1}}(p_i(v^+_{-i},v_{-i}^-)-\lambda^*_{-i}v_{-i}-\lambda^*_i\vmax)$. Hence, this second infimum is not greater than $\inf\limits_{v_{-i}\in[0,\vmax]^{k-1}}(p_i(v^+_{-i},0)-\lambda^{*+}_{-i}v^+_{-i}-\lambda^*_i\vmax)=R_0$. Thus, the values of such infima do not decrease upon replacing $p_i(v_{-i})$ with $p_i(v^+,0)=p^{new}$. 

Finally, consider  $\inf\limits_{v\in W_0(p^{new})}(-\lambda^* v)$. 
We will prove that, $\inf\limits_{v\in W_0(p^{new})}(-\lambda^* v)$ is no less than at least one infimum in the expression for $R(p,\lambda^*)$, written as in \eqref{reducedproblem1}, and the statement of this lemma will follow. Note that unlike all previous steps, here we will not show that $\inf\limits_{v\in W_0(p^{new})}(-\lambda^* v)$ necessarily dominates a similar infimum, $\inf\limits_{v\in W_0(p)}(-\lambda^* v)$. Rather, we will show that it is \emph{either} no less than $\inf\limits_{v\in W_0(p)}(-\lambda^* v)$ \emph{or} no less than $\inf\limits_{v\in W^{\geq}_j(p)}(p_j(v_{-j})-\lambda^* v)$ for some $j$.

If $W_0(p^{new})$ is empty, the infimum in the new mechanism is equal to $+\infty$ and thus is trivially higher than $\inf\limits_{v\in W_0(p)}(-\lambda^* v)$. So suppose $W_0(p^{new})$ is nonempty. Consider minimizing sequences $v^k\in W_0(p^{new})$ for the problem $\inf\limits_{v\in W_0(p^{new})}(-\lambda^* v)$.  
Note that there exists a minimizing sequence $v^k$ such that for all $i\in-(\lambda^*)$, $v^k_i=0$ for all $k$, since for all $v\in W_0(p^{new})$, $(v^+,0)\in W_0(p^{new})$ and $\lambda_i^*< 0$. Consider such a sequence. 

\textbf{Case 1}. $v^k\in W_0(p)$ starting from some $k$. This implies that 
$\inf\limits_{v\in W_0(p)}(-\lambda^* v)\leq \lim\limits_{k\to\infty}(-\lambda^*v^k)=\inf\limits_{v\in W_0(p^{new})}(-\lambda^* v)$ which implies the Lemma.

\textbf{Case 2}. $v^k\notin W_0(p)$ for infinitely many $k$. Thus, as there is only a finite number of bidders, there exists $j^*$ such that $v^k_{j^*}\geq p_{j^*}(v^k_{-j^*})$ for infinitely many $k$. Consider hereafter only this subsequence. Recall that $v^k\in W_0(p^{new})$ for all $k$. If $j^*\in +(\lambda^*)$, we would have $p_{j^*}(v^k_{-j^*})=p^{new}(v^k_{-j^*})>v^k_{j^*}\geq p_{j^*}(v^k_{-j^*})$, a contradiction. So $j^*\in-(\lambda^*)$ and thus for all $k$ $v^k_{j^*}=0\geq p_{j^*}(v^k_{-j^*})\geq 0$. So $p_j(v^k_{-j^*})=0$ for all $k$. Moreover, $v^k\in W^{\geq}_{j^*}(p)$ for all $k$. Hence, $\inf\limits_{v\in W^{\geq}_{j^*}(p)}(p_j(v^k_{-j^*})-\lambda^* v)\leq \lim\limits_{k\to\infty}(0-\lambda^*v^k)=\inf\limits_{v\in W_0(p^{new})}(-\lambda^* v)$ which again implies the lemma. \eop
\end{proof}

Now consider a vector $\lambda^{**}$ obtained by replacing  $\lambda^*_i$ for $i\in -(\lambda^*)$ with zeros.

\begin{lemma}\label{newlambdapositive}
$R(p^{new},\lambda^{**})> R(p^{new},\lambda^*)$.
\end{lemma}
\begin{proof}
We will prove that the minimal of all infima in \eqref{reducedproblem} won't change when $\lambda^*_i$ for $i\in -(\lambda^*)$ are replaced with zeros. Since $m_i>0$, revenue will strictly increase. 

Note that the infima except $\inf\limits_{v\in W^{\geq}_i(p^{new})}(p^{new}_i(v_{-i})-\lambda^* v)$ for $i\in -(\lambda^*)$ are minimized by setting $v_j=0$ for all $j\in -(\lambda^*)$. This will continue to hold when $\lambda^*_i$ is replaced with $\lambda^{**}$. Thus, replacing $\lambda^*$ with $\lambda^{**}$ will not change any infimum except $\inf\limits_{v\in W^{\geq}_i(p^{new})}(p^{new}_i(v_{-i})-\lambda^* v)=\vmax(1-\sum_{j\in+(\lambda^*)}\lambda^*_j-\lambda^*_i)$ for $i\in-(\lambda^*)$. But note that none of $\inf\limits_{v\in W^{\geq}_i(p^{new})}(p^{new}_i(v_{-i})-\lambda v)$, $i\in-(\lambda^*)$, is strictly minimal of all the infima, for both $\lambda=\lambda^*$ and $\lambda=\lambda^{**}$. Indeed, if $+(\lambda^*)\neq\emptyset$, one can take any $j\in +(\lambda^*)$ and consider $\inf\limits_{v\in W^{\geq}_j(p^{new})}(p_j^{new}(v_{-j})-\lambda v)$. Since the point $v_s=\begin{cases} \vmax, & s\in +(\lambda^*)\cup \{i\};\\ 0, & \mbox{ o/w}\end{cases}$ is feasible, and $p^{new}_j(\cdot)\leq\vmax$, this infimum is weakly lower than  $\inf\limits_{v\in W^{\geq}_i(p^{new})}(p^{new}_i(v_{-i})-\lambda v)=\vmax(1-\sum_{j\in+(\lambda^*)}\lambda_j-\lambda_i)$ for both $\lambda=\lambda^*$ and $\lambda=\lambda^{**}$. If, on the other hand, $+(\lambda^*)=\emptyset$,  $\inf\limits_{v\in W_0(p^{new})}(-\lambda v)$ is lower, as  it is equal to $0<\vmax(1-\lambda_i)$ for both $\lambda=\lambda^*$ and $\lambda=\lambda^{**}$.
\eop 
\end{proof}
Along with Lemma \ref{pnewbetter}, we get the inequality $R(p^{new},\lambda^{**})> R(p,\lambda^*)$.

Finally, to find an LSA that dominates $p^{new}$, and thus, $p$, feed the pair $[p^{new},\lambda^{**}]$ to \textbf{Grand Case I.} This is possible because all entries of $\lambda^{**}$ are nonnegative. Note that now the tilde-transformation will be applied to $p^{new}$. Also, even though the whole vector $\lambda^{**}$ might not be an optimal one for the mechanism $p^{new}$, this does not create a problem since optimality of $\lambda^*$ is not used throughout constructions in \textbf{Grand Case 1}. \eop

The proof of theorem \ref{main} is complete.

\section{Parametric solutions}\label{paramsol}
Theorem \ref{main} identifies \emph{the form} of an optimal mechanism. It remains to identify optimal values of the parameters $r_i$. To do this, one has to maximize $R(p,\lambda)$ where functions $p$ are given by \eqref{p-lsa} for some $r$, both with respect to $r$ and $\lambda$. Note that it follows from lemma \ref{newlambdapositive} that any optimal mechanism $M^*$ is such that its optimal Lagrange multiplier $\lambda^*(M^*)$ is nonnegative. Together with proposition \ref{optauction exists}, this implies that that in finding the optimal linear score auction, one can restrict oneself to $\lambda\geq 0$.

Before one optimizes over $r$ and $\lambda$, one has to solve the inner minimization problems in \eqref{transformedproblem}.  This is done in the proof of the following lemma:

\begin{lemma}\label{innerproblsol}
Suppose threshold functions $p$ are given by \eqref{p-lsa} and $\lambda\geq 0$.  Then, 
\beq\label{mlambda}
R(p,\lambda)=\begin{cases}
\lambda m+\min\left\{\vmax\left(1-\sum\lambda_i\right),\min\limits_i(r_i-\lambda_{-i}r_{-i}-\lambda_i\vmax),-\lambda r \right\}, & r>0,\\
\lambda m+\min\left\{\vmax\left(1-\sum\lambda_i\right),\min\limits_i(r_i-\lambda_{-i}r_{-i}-\lambda_i\vmax)\right\}, & \mbox{o/w}.
\end{cases}
\eeq
\end{lemma}

The two cases arise because when $r_i=0$ for some $i$, $W_0(p)=\emptyset$, so $\inf\limits_{W_0(p)}(-\lambda v)$ jumps to $+\infty$ at all such points. 

One approach to maximize this function of $2n$ variables would be to maximize first over $\lambda$ and then over $r$. This corresponds to finding explicitly the worst-case revenue for every linear score auction and then finding the best auction. Even though this might seem a more natural approach, it turns out to be more cumbersome. 

So we optimize first over $r$ and then over $\lambda$ instead. Abusing notation, denote the revenue function by $R(r,\lambda)$. The function $r\to R(r,\lambda)$ is piecewise-linear and potentially discontinuous at points where $r_i=0$ for some $i$. Maximization suggests that at least two of the arguments of the minimum operator in \eqref{innerproblsol} should be equal to each other at an optimum, but it is unclear \emph{a priori}, which two (or more). Fortunately, one can show that (1) when $\sum_i\frac{\lambda_i}{1+\lambda_i}\leq 1$, all arguments but $\vmax\left(1-\sum\lambda_i\right)$, must be equal to each other; (2) the case $\sum_i\frac{\lambda_i}{1+\lambda_i}> 1$ is uninteresting from the standpoint of overall optimization over $r$ and $\lambda$. This is stated in the following two lemmas.

Given $\lambda\geq 0$, define the point $r^*(\lambda)$ by 
\[r_i^*(\lambda):=\frac{\lambda_i}{1+\lambda_i}\vmax.\] 
 
\begin{lemma}\label{rstaroptimal}
If $\sum_i\frac{\lambda_i}{1+\lambda_i}\leq 1$, $r^*(\lambda)$ solves the problem $\max\limits_r R(r,\lambda)$.
\end{lemma}

\begin{lemma}\label{geomconstraintexists}
Suppose $(\lambda^0,r^0)$ is a solution to $\max\limits_{r,\lambda}R(r,\lambda)$. Then, $\sum_i\frac{\lambda^0_i}{1+\lambda^0_i}\leq 1$.
\end{lemma}

In the proof of lemma \ref{rstaroptimal}, we build an iterative procedure that at each iteration moves the price vector $r$ to equalize a growing number of arguments of the minimum operator in \eqref{mlambda}, terminating when all but $\vmax\left(1-\sum\lambda_i\right)$ are equal to each other, implying that it finishes at $r^*(\lambda)$. Using the condition $\sum_i\frac{\lambda_i}{1+\lambda_i}\leq 1$, we show that at each step the value of \eqref{mlambda} does not decrease. The prices $r$ with $r_i=0$ for some $i$ are treated in a special way.

It follows from lemmas \ref{rstaroptimal}, \ref{geomconstraintexists} that any optimal Lagrange multiplier $\lambda^*$ solves the following problem:
\begin{align}
\max\limits_{\lambda\geq 0} &\mbox{ } \sum_{i=1}^n m_i\lambda_i-\frac{\lambda_i^2}{1+\lambda_i}\vmax\label{finalproblem}\\
\mbox{s.t. } &\sum_{i=1}^n\frac{\lambda_i}{1+\lambda_i}\leq 1  \label{geomconstr}
\end{align}

Recall from the proof of theorem \ref{main} that the constraint \eqref{geomconstr} ensures that the determinant of the matrix $A(\lambda)$ is nonnegative. It is no coincidence that this constraint appears also in the present part of analysis. Indeed, if this constraint is violated, the transformed threshold functions ($\tilde{p}$) are mutually located as in figure \ref{fig:2geom}, right. But in this case, proceeding to the revenue-improving LSA improves revenue strictly as threshold functions are moved strictly up (except at one point), and the object is allocated at every value profile. Hence, the initial mechanism (which could be itself an LSA) was suboptimal\footnote{This is a geometric intuition behind the (algebraic) proof of Lemma \ref{geomconstraintexists}.}. Thus, the constraint \eqref{geomconstr} eliminates situations with the ``wrong'' implied geometry of $\tilde{p}$. We refer to it as ``geometry constraint''.

The solution to the problem \eqref{finalproblem} differs depending on whether the geometry constraint \eqref{geomconst} binds. If the constraint does not bind, the solution can be easily computed to be 
\[\lambda_i^*=\sqrt{\frac{\vmax}{\vmax-m_i}}-1.\]
Thus, the geometry constraint binds whenever $\sum_i\frac{\lambda^*_i}{1+\lambda_i^*}\geq 1$, that is, when
\[\sum_{i=1}^n\sqrt{1-m_i/\vmax}\leq n-1,\]
i.e., when the means are relatively high. Whether this inequality holds will determine two different regimes for the solution of the optimal auction problem.

When the geometry constraint \eqref{geomconstr} binds, the nonnegativity constraints might also start to bind for bidders whose means are relatively low (even though the means are high overall.) The full solution to the problem \eqref{finalproblem}-\eqref{geomconstr} is given in the following lemma.

\begin{lemma}\label{optimallambdas}
Enumerate bidders in such a way that $m_1\leq m_2\leq\ldots\leq m_n$. The solution to the problem \eqref{finalproblem}-\eqref{geomconstr} is as follows:
\begin{enumerate}
\item If $\sum\limits_{i=1}^n\sqrt{1-m_i/\vmax}>  n-1$, 
\beq\label{lambdalowmeans}
\lambda_i^*=\sqrt{\frac{\vmax}{\vmax-m_i}}-1\mbox{ for all }i.
\eeq
\item If $\sum\limits_{i=1}^n\sqrt{1-m_i/\vmax}\leq n-1$, 
\beq\label{lambdahighmeans}
\lambda_i^*=\begin{cases}
0, & i<k^*;\\
\frac{\sum\limits_{j=k^*}^n\sqrt{\vmax-m_j}/(n-k^*)}{\sqrt{\vmax-m_i}}-1, & i\geq k^*,
\end{cases}
\eeq
where $k^*=\min\left\{k\leq n-1|\sum\limits_{i=k+1}^n\frac{\sqrt{\vmax-m_i}}{\sqrt{\vmax-m_k}}>n-k-1\right\}$.
\end{enumerate}
\end{lemma}

The bidders with $\lambda^*_i=0$ will be the ones who can be excluded from the auction without loss of worst-case revenue. For a given bundle of parameters $(m,\vmax)$, denote the set $\{i: \lambda_i^*=0\}$ by $WE(m,\vmax)$ -- those are bidders that are \emph{weakly excluded}. Denote the complement set of bidders by $SI(m,\vmax)$ -- those that are \emph{strictly included}.  Note that no more than $n-2$ bidders can be weakly excluded in any case. This may be explained by the need to keep at least a minimal level of competition. It follows that no one can be weakly excluded if $n=2$.

Given the definition of $k^*$, the condition that ensures that no bidder gets weakly excluded for $n\geq 3$ is 
\[\sqrt{\vmax-m_1}<\frac{n-1}{n-2}\frac{\sum\limits_{i=2}^n\sqrt{\vmax-m_i}}{n-1}.\] 
That is, the lowest of the means has to be relatively high still.

Now, given the optimal $\lambda^*$ stated in lemma \ref{optimallambdas}, we want to recover optimal reserve prices $r^*$ for the linear score auction. Recall that by Lemma \ref{rstaroptimal}, $r^*(\lambda)$ given by $r^*_i=\frac{\lambda_i}{1+\lambda_i}\vmax$ is \emph{an} optimal vector of prices. However, there may be others. 
Indeed, at some of the steps in the proof of Lemma \ref{rstaroptimal}, the inequalities may have been only weak. Carefully tracing this leads to the following answer. 

\begin{proposition}\label{optimal prices}
The set of optimal generalized reserve prices for the corner-hitting linear score auction is as follows: 
\begin{enumerate}
\item If $\sum\limits_{i=1}^n\sqrt{1-m_i/\vmax}> n-1$ (means are low), there is a unique vector of optimal prices $r^*$ given by
\beq\label{priceslowmeans}
r^*_i=\vmax-\sqrt{\vmax(\vmax-m_i)}\mbox{ for all } i;
\eeq
\item If $\sum\limits_{i=1}^n\sqrt{1-m_i/\vmax}\leq n-1$ (means are high), a vector of reserve prices $r^*\in[0,\vmax]^n$ is optimal if and only if it is satisfies the following system
\beq \label{priceshighmeans}
\begin{cases}
\sum\limits_{i\in SI(m,\vmax)} r^*_i\leq \vmax;\\
\frac{\vmax-r^*_i}{\vmax-r^*_j}=\sqrt{\frac{\vmax-m_i}{\vmax-m_j}}\mbox{ for all }i,j\in SI(m,\vmax),
\end{cases}
\eeq
where $SI(m,\vmax)$ is the set of bidders for whom the optimal Lagrange multipliers $\lambda^*$, as given by \eqref{lambdahighmeans}, are positive.
\end{enumerate}
\end{proposition}

Note that if the means are high, the set of reserve prices for strictly included bidders is one-dimensional. 
If some bidders are weakly excluded, the reserve prices $r_i$ for such bidders can be set arbitrarily (and the exclusion is achieved when $r_i=\vmax$). Figure \ref{fig:parametric}, left, depicts areas of parameter space that correspond to low/high means and to the presence/absence of weak exclusion for a specific example with $n=3$. Figure \ref{fig:parametric}, right, shows the set of optimal price vectors for the case of high means if $n=2$.

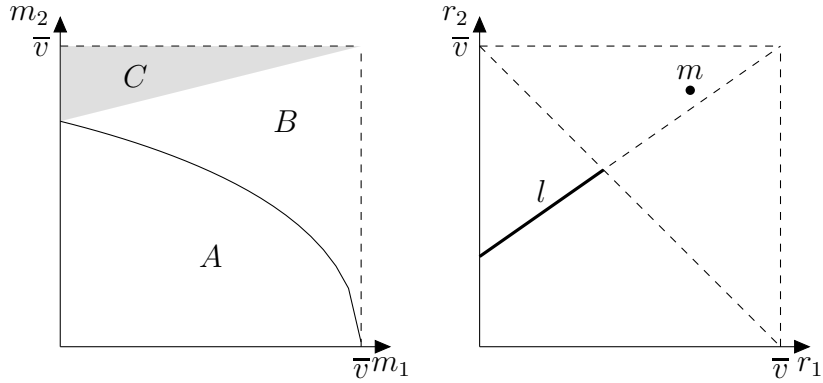
\begin{figure}[h!]
\centering

\begin{tikzpicture}[>=triangle 45, xscale=4,yscale=4]
\draw[->] (0,0) -- (1.1,0) node[below] {$m_1$};
\draw[->] (0,0) -- (0,1.1) node[left] {$m_2$};
\draw[dashed] (1,0)--(1,1);
\draw[dashed] (0,1)--(1,1);
\node[below] at (1,0){$\vmax$};
\node[left] at (0,1){$\vmax$};
\fill[gray, nearly transparent] (0,0.75)--(1,1)--(0,1);
\draw[domain=0:1]
plot (\x, {1-pow(1-sqrt(1-\x)/2,2)});
\node at (0.5,0.3){$A$};
\node at (0.75,0.75){$B$};
\node at (0.25,0.9){$C$};
\end{tikzpicture}
\begin{tikzpicture}[>=triangle 45, xscale=4,yscale=4]
\draw[->] (0,0) -- (1.1,0) node[below] {$r_1$};
\draw[->] (0,0) -- (0,1.1) node[left] {$r_2$};
\draw[dashed] (1,0)--(1,1);
\draw[dashed] (0,1)--(1,1);
\node[below] at (1,0){$\vmax$};
\node[left] at (0,1){$\vmax$};
\draw[dashed] (0,1)--(1,0);
\draw[dashed] (0,0.3)--(1,1);

\node[circle,fill=black,inner sep=1.25pt,minimum size=1pt] at (0.7,0.853) {};
\node[above] at (0.7,0.853){$m$};

\draw[very thick] (0,0.3)--(7/17,10/17);
\node[above] at (7/34,151/340){$l$};
\end{tikzpicture}
\caption{(Left.) Different regimes for the optimal solution depending on the means, $n=3$. One bidder's value has a mean of $m_1$ and the other two  bidders' values have a mean of $m_2$. In area $A$ means are low and the vector of optimal reserve prices is unique. In areas $B$ and $C$ means are high and there are multiple solutions. In addition, in area $C$ the bidder with the lowest mean is weakly excluded from the auction. (Right.) The set of optimal price vectors in the high-means case if $n=2$, $m_2>m_1$.}\label{fig:parametric}
\end{figure}

As noted in the introduction, the solution features discrimination against stronger bidders. Indeed, for any two bidders $i,j\in SI(m,\vmax)$, $m_i>m_j$ implies $r_i^*>r^*_j$ for any optimal price vector $r^*$. This is consistent with the standard result in auction theory \citep{myerson1981optimal}. The differences from the classic solution include weak exclusion of bidders discussed above, multiplicity of solutions, the fact that the set of optimal optimal reserve prices depends on $n$ in the symmetric case.

The multiplicity of solutions in the high-means case stems from the fact that many auctions share the same relevant region of $conv[t^M(v)]$ and induce the same worst-case distribution. This distribution is such that the sale always happens (see more on this in section \ref{wcdistrsect}). Note, however, that the (generalized) reserve prices also control the slope of the boundary determining which bidder gets the object, which is important. This is why the prices must belong to a certain line segment. %

In the symmetric case, we obtain the following corollary of proposition \ref{optimal prices}:

\begin{corollary}
Suppose bidders are symmetric ($m_i=m_j$). Then, a second-price auction with reserve price $r$ solves problem \eqref{main} if and only if $r=\vmax-\sqrt{\vmax(\vmax-m)}$ and $n < \frac{1}{1-\sqrt{1-m/\vmax}}$ or $r\in[0,\vmax/n]$ and $n\geq \frac{1}{1-\sqrt{1-m/\vmax}}$.\footnote{This parametric solution has been obtained by \cite{koccyiugit2017distributionally}. However, they take the auction format -- SPA --as given.}
\end{corollary}

Thus, the set of optimal reserve prices decreases in strong set order and converges to point zero as $n\to\infty$. Thus, as competition increases, a reserve price is no longer needed to protect the seller from low-revenue distributions. The result that a second-price auction without a reserve is an optimal mechanism for $n$ sufficiently high echoes results in \cite{he2019robustly}, \cite{che2019robust} and \cite{suzdal2020distributionally} who show that setting a reserve equal to the seller's own value is either exactly or asymptotically optimal in related maxmin settings. 

Note also that with symmetric bidders, the seller never sets the reserve above $\vmax/n$, even if the mean $m$ is arbitrarily close to $\vmax$. This contrasts with the classical case in which the optimal reserve may be arbitrarily close to $\vmax$ if the distribution is concentrated in a small neighborhood of $\vmax$. This is yet another manifestation of the extreme cautiousness exercised by the seller.

\section{The set of optimal mechanisms}\label{set}
The proof of Theorem \ref{main} establishes that a corner-hitting linear score auction is an optimal mechanism in the considered environment. But do there exist other optimal mechanisms?  In other words, to what extent is linearity important for the optimality of the mechanism?

In this section, we study this aspect of the problem for the case of two bidders. We characterize the whole set of optimal (deterministic) mechanisms. This set is far from being a singleton. The features of the set turn out to depend significantly on whether the means are low ($\sqrt{1-m_1/\vmax}+\sqrt{1-m_2/\vmax}>1$) or high ($\sqrt{1-m_1/\vmax}+\sqrt{1-m_2/\vmax}\leq 1$). When the means are high, the optimal mechanism is pinned down uniquely for sufficiently high values, that is, linearity of the relative boundary is indeed necessary for optimality. 

We start with a lemma that says that all optimal mechanisms induce  the same Lagrange multipliers on the means constraints in Nature's problem. Denote by $R(M,\lambda)$ the revenue under a mechanism $M$ and Lagrange multipliers $\lambda$. Let $\lambda^*(M)$ be any Lagrange multipliers solving the problem \eqref{biconj} (equivalently, $\max\limits_\lambda R(M, \lambda)$ for a mechanism $M$.) Denote by $LSA^{opt}$ any optimal linear score auction.

\begin{lemma}\label{lambdasame}
If a mechanism $M^0$ solves problem \eqref{main}, any solution $\lambda^*(M^0)$ to the dual of Nature's problem \eqref{biconj} given $M^0$ coincides with $\lambda^*(LSA^{opt})$, given by \eqref{lambdalowmeans} and \eqref{lambdahighmeans}, 
\end{lemma}

\begin{proof}
Denote by $R^*$ the optimal revenue achieved by any mechanism. As $M_0$ is optimal,  $R^*=R(M^0,\lambda^*(M^0))$. By the proof of theorem \ref{main}, there exists a linear score auction $LSA^0$ such that $R(M^0,\lambda^*(M^0))\leq R(LSA^0,\lambda^*(M^0))$. Hence, $R(LSA^0,\lambda^*(M^0))\geq R^*$ and this $LSA^0$ has to be, in a fact, an optimal LSA, $LSA^{opt}$. Thus, $\lambda^*(M^0)$ has to be a maximizer of $R(LSA^{opt},\lambda)$. Such a maximizer is unique and is given by \eqref{lambdalowmeans} and \eqref{lambdahighmeans}.
\eop
\end{proof}

The set of optimal mechanism for $n=2$ is characterized in the following two propositions.

\begin{proposition}\label{setmlow}
Suppose $\sqrt{1-m_1/\vmax}+\sqrt{1-m_2/\vmax}>1$. Let $r^*_i$ be the reserve prices of the optimal LSA given by \eqref{priceslowmeans} and $\lambda^*_i$ be the optimal Lagrange multipliers given by \eqref{lambdalowmeans}. Then, a mechanism $(p_1(v_2),p_2(v_1))$ solves problem \eqref{main} if and only if the following conditions all hold:
\begin{enumerate}
\item $p_i(v_{-i})\geq r^*_i+\lambda^*_i(v_{-i}-r^*_{-i})$ for all $v_{-i}\in[0,\vmax]$;
\item $p_i(v_{-i})\leq \left(\lambda^*_1r^*_1+\lambda_2^*r_2^*-\lambda^*_{-i}v_{-i}\right)/\lambda^*_i$ for $v_{-i}\leq r^*_{-i}$;
\item $p_i(v_{-i})$ are weakly increasing for $v_{-i}\geq r^*_{-i}$. 
\item $p_1(v_2)$ and $p_2(v_1)$ are inverse to each other for $v_i\geq r^*_i$ whenever possible, i.e. for any interval $(v_1',v_1'')$, $v_1'\geq r_1^*$ such that $p_2(v_1)$ is strictly increasing on it, $p_1(p_2(v_1))=v_1$ for all $v_1\in (v_1',v_1'')$, and similarly for $p_1(v_2)$.
\end{enumerate}
\end{proposition}

\begin{proposition}\label{setmhigh}
Suppose $\sqrt{1-m_1/\vmax}+\sqrt{1-m_2/\vmax}\leq 1$. Let $\lambda^*_i$ be the optimal Lagrange multipliers given by \eqref{lambdahighmeans}. Note that $\lambda^*_1\lambda^*_2=1$. Let $r^*_i=\lambda_i^*/(1+\lambda_i^*)\vmax$ be the highest optimal reserve prices for the LSA.  Then, a mechanism $(p_1(v_2),p_2(v_1))$ solves problem \eqref{main} if and only if
the following conditions all hold:
\begin{enumerate}
\item  $p_i(v_{-i})= r^*_i+\lambda^*_i(v_{-i}-r^*_{-i})$ for all $v_{-i}\geq r^*_{-i}$;
\item $p_i(v_{-i})\geq r^*_i+\lambda^*_i(v_{-i}-r^*_{-i})$ for all $v_{-i}\in[0,\vmax]$;
\item $p_i(v_{-i})\leq \left(\lambda^*_1r^*_1+\lambda_2^*r_2^*-\lambda^*_{-i}v_{-i}\right)/\lambda^*_i$ for $v_{-i}\leq r^*_{-i}$.
\end{enumerate}
\end{proposition}

In both propositions, sufficiency is easy to prove from the fact that for any mechanism satisfying the conditions it can be verified that $R(M,\lambda^*)=R^*$. Necessity is trickier. The main idea is that if one submits any optimal mechanism to the proof of theorem \ref{main}, the resulting LSA should be an optimal one, characterized by proposition \ref{optimal prices}. This means that one of fixed points of $\tilde{p}$, derived from $M$, must be an optimal vector of prices. Given that we know the slopes of $\tilde{p}$ from lemma \ref{lambdasame}, this reconstructs $\tilde{p}$ for any optimal mechanism $M$, thus yielding condition 1 in proposition \ref{setmlow} (per property $\tilde{p}\leq p$). Condition 2 in the same proposition stems from the fact that for any optimal mechanism, $\inf\limits_{v\in W_0}(-\lambda^*v)$ may not be lower than that for the optimal LSA. Finally, conditions 3 and 4 in proposition \ref{setmlow} stem from the fact that the threshold functions for the optimal mechanism must both satisfy supply constraint and be tight to each other, not allowing the ``holes'' that would expand the set $W_0$. 

In proposition \ref{setmhigh}, the linearity of threshold functions (condition 1) is just a consequence of the fact that that reconstructed $\tilde{p}$ are tight to each other, and thus, per the property $\tilde{p}\leq p$ and the fact that one shouldn't expand the set $W_0$ too much, threshold functions must coincide with $\tilde{p}$ when values are sufficiently high. Because the threshold functions are pinned down for $v_i\geq r^*_i$, one does not have to prove analogs of conditions 3 and 4 of proposition \ref{setmlow} in proposition \ref{setmhigh}. 

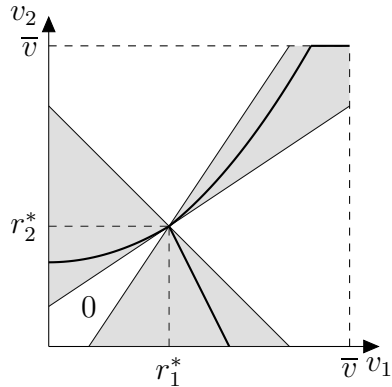
\begin{figure}[h!]\label{fig:setlowmeans}
\centering
\begin{tikzpicture}[>=triangle 45, xscale=4,yscale=4]
\draw[->] (0,0) -- (1.1,0) node[below] {$v_1$};
\draw[->] (0,0) -- (0,1.1) node[left] {$v_2$};
\draw[dashed] (1,0)--(1,1);
\draw[dashed] (0,1)--(1,1);
\node[below] at (1,0){$\vmax$};
\node[left] at (0,1){$\vmax$};

\draw[dashed] (0.4,0)--(0.4,0.4);
\draw[dashed] (0,0.4)--(0.4,0.4);
\node[left] at (0,0.4){$r^*_2$};
\node[below] at (0.4,0){$r^*_1$};

\draw (0,0.1333333)--(1,0.8);
\draw (0.1333333,0)--(0.8,1);
\draw (0,0.8)--(0.8,0);

\fill[gray,nearly transparent] (0.4,0.4)--(1,0.8)--(1,1)--(0.8,1);
\fill[gray,nearly transparent] (0.4,0.4)--(0,0.8)--(0,0.1333333);
\fill[gray,nearly transparent] (0.4,0.4)--(0.8,0)--(0.1333333,0);

\draw[thick,domain=0.4:0.87177] plot(\x,{0.24+\x*\x});
\draw[thick,domain=0:0.4] plot(\x,{0.28+0.75*\x*\x});
\draw[thick] (0.4,0.4)--(0.6,0);
\draw[thick] (0.87177,1)--(1,1);

\node at (0.1333333,0.1333333){0};

\end{tikzpicture}
\caption{Optimal mechanisms in a low-means case. $n=2$, $\vmax=1$, $m_i=0.64$. $r_i^*=0.4$ and $\lambda_i^*=2/3$. The graphs of all optimal threshold functions must lie within the shaded areas. A sample optimal mechanism is shown.}
\label{fig:setlowmeans}
\end{figure}

The sets of optimal mechanisms in the low-means and high-means case are depicted in figures \ref{fig:setlowmeans} and \ref{fig:sethighmeans}. We see that in the low means case (figure \ref{fig:setlowmeans}) it is necessary that the graphs of threshold functions of any optimal mechanism pass exactly through the point $(r_1^*,r_2^*)$. This signifies the fact that when means are relatively low, the main trade-off that the seller faces is the trade-off between not selling the good and the prices upon selling and not the trade-off between favoring bidder 1 or 2 when both values are above the reserve prices.  The existence of the area denoted by ``0'' shows that it is strictly optimal not to sell the good sometimes. 
Note, however, that the threshold functions should be still ``close enough'' to those of the optimal linear score auction, with the ``safe neighborhood'' of the optimal linear score auction given by $\lambda^*$.

In the high-means case (figure \ref{fig:sethighmeans}), the ``safe neighborhood'' collapses so that the any optimal mechanism should coincide with the optimal linear score auction for $v_i\geq r^*_i$. This stems from the fact that with high means, the probability of not selling the good will be small anyway (provided that prices are not very high), so the main trade-off is about designing the right mode of bidders' competition. Note also that, as the ``0'' set collapses, it becomes weakly optimal to always sell the object. 

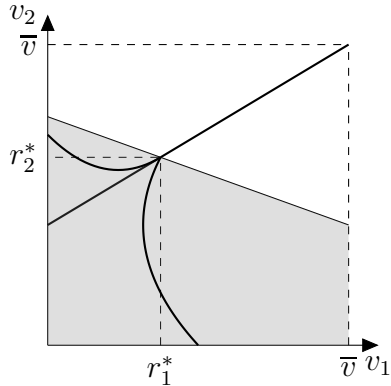
\begin{figure}[h!]\label{fig:sethighmeans}
\centering
\begin{tikzpicture}[>=triangle 45, xscale=4,yscale=4]
\draw[->] (0,0) -- (1.1,0) node[below] {$v_1$};
\draw[->] (0,0) -- (0,1.1) node[left] {$v_2$};
\draw[dashed] (1,0)--(1,1);
\draw[dashed] (0,1)--(1,1);
\node[below] at (1,0){$\vmax$};
\node[left] at (0,1){$\vmax$};

\draw[dashed] (0.375,0)--(0.375,0.625)--(0,0.625);
\node[left] at (0,0.625){$r^*_2$};
\node[below] at (0.375,0){$r^*_1$};

\draw[thick] (0.375,0.625)--(1,1);
\draw[thick] (0,0.4)--(0.375,0.625);
\draw (0,0.76)--(1,0.4);
\fill[gray,nearly transparent] (0,0.76)--(1,0.4)--(1,0)--(0,0);
\draw[domain=0:0.625, smooth, variable=\y, thick]  
plot ({8*\y*\y/7-(6.4)*\y/7+0.5}, {\y});
\draw[thick, domain=0:0.375] plot(\x,{0.7-\x+(32/15)*\x*\x});

\end{tikzpicture}
\caption{Optimal mechanisms in a high-means case. $n=2$, $\vmax=1$, $m_1=0.75$, $m_2=0.91$. $r_1^*=3/8$, $r^*_2=5/8$, $\lambda^*_1=1/\lambda_2^*=0.6$. The graphs of all optimal threshold functions must coincide with the thick line within the unshaded area, and must lie above it within the shaded area. The thick line itself represents an optimal mechanism -- a linear score auction that always allocates the good. Two curves correspond to another mechanism.}
\label{fig:sethighmeans}
\end{figure}

\section{Extensions}\label{ext}
\subsection{Different upper bounds}
In this section, we analyze the case where the seller may put different upper bounds on different bidders' values. For simplicity, we focus on the case of two bidders and assume that the support of the joint distribution $F$ is contained in $[0,\vmax_1]\times[0,\vmax_2]$ for some $\vmax_1\geq\vmax_2>0$. 

The optimality of linear score auctions still holds. The difference is that the optimal linear score auction ceases to be a corner-hitting linear score auction. 

\begin{theorem}\label{main2}
Suppose $\vmax_1\geq\vmax_2$. Then there exists a mechanism $M^*$ that solves \eqref{problem} and is a linear score auction with parameters $\beta_1=\gamma$, $\alpha_1=\gamma r_1$, $\beta_2=1$, $\alpha_2=r_2$ for some $r_i\in[0,\vmax_i]$ and $\gamma\in\mathbb{R}_+$.
\end{theorem}

\begin{figure}[h!]
\centering
\begin{tikzpicture}[>=triangle 45, xscale=3,yscale=3]
\draw[->] (0,0) -- (2.1,0) node[below] {$v_1$};
\draw[->] (0,0) -- (0,1.1) node[left] {$v_2$};
\draw[dashed] (2,0)--(2,1);
\draw[dashed] (0,1)--(2,1);
\node[below] at (2,0){$\vmax_1$};
\node[left] at (0,1){$\vmax_2$};

\draw[thick] (0.4,0)--(0.4,0.7);
\draw[thick] (0,0.7)--(0.4,0.7);
\draw[thick] (0.4,0.7)--(2,1);
\node[left] at (0,0.7){$r_1$};
\node[below] at (0.4,0){$r_2$};

\node at (0.7,0.35){1};
\node at (0.35,0.85){2};
\node at (0.2,0.35){0};

\end{tikzpicture}
\hspace{1em}
\begin{tikzpicture}[>=triangle 45, xscale=3,yscale=3]
\draw[->] (0,0) -- (2.1,0) node[below] {$v_1$};
\draw[->] (0,0) -- (0,1.1) node[left] {$v_2$};
\draw[dashed] (2,0)--(2,1);
\draw[dashed] (0,1)--(2,1);
\node[below] at (2,0){$\vmax_1$};
\node[left] at (0,1){$\vmax_2$};

\draw[thick] (0.4,0)--(0.4,0.7);
\draw[thick] (0,0.7)--(0.4,0.7);
\draw[thick] (0.4,0.7)--(0.8,1);
\node[left] at (0,0.7){$r_1$};
\node[below] at (0.4,0){$r_2$};

\node at (0.7,0.35){1};
\node at (0.35,0.85){2};
\node at (0.2,0.35){0};

\end{tikzpicture}
\caption{Two linear score auctions with two bidders. The first auction may be suboptimal with upper-bound asymmetric bidders. Unlike the case $\vmax_1=\vmax_2$, one may stictly prefer the boundary ``hitting the wall'' (right picture) rather than ``hitting the corner" (left).  Digits 1 and 2 denote areas of value space where the corresponding bidders get the object. Zero denotes areas where the object is kept by the seller.}\label{fig:auctions2}
\end{figure}
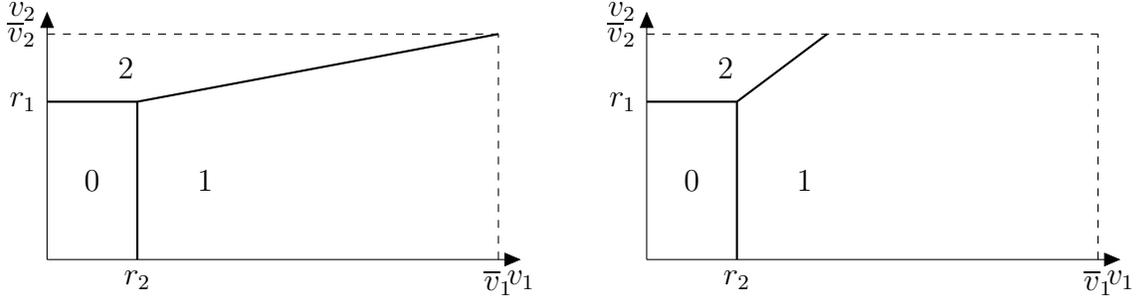

The parametric solution is given in the following proposition:
\begin{proposition}\label{optparam}
Suppose $\vmax_1\geq \vmax_2$. Then,
\begin{enumerate}
\item If $\sqrt{1-m_1/\vmax_1}+\sqrt{1-m_2/\vmax_2}>1$, the optimal prices $r_i^*$ are unique and given by
\beq \label{optpriceslowmeans}r_i^*=\vmax_i-\sqrt{\vmax_i(\vmax_i-m_i)},\mbox{ }i=1,2.
\eeq
The optimal slope $\gamma^*$ is not unique; $\gamma^*$ is optimal if and only if
\beq\label{optgammameanslow}
\gamma^*\in\left[\frac{r_1^*}{\vmax_1-r_1^*},\frac{\vmax_2-r_2^*}{r_2^*}\right].\eeq
\item If $\sqrt{1-m_1/\vmax_1}+\sqrt{1-m_2/\vmax_2}\leq 1$, the optimal slope $\gamma^*$ is unique and given by the (unique) positive solution to the equation 
\beq\label{gamma}
\frac{\vmax_1-\vmax_2}{(\gamma+1)^2}+\frac{\vmax_2-m_2}{\gamma^2}=\vmax_1-m_1.
\eeq
The optimal pair of prices is not unique; $(r_1^*,r_2^*)$ is optimal if and only if 
$$\begin{cases}
\frac{\vmax_2-r_2^*}{\tilde{v}_1-r_1^*}=\gamma^*\\
\frac{r_1^*}{\vmax_1}+\frac{r_2^*}{\vmax_2}\leq 1,
\end{cases}$$
where $\tilde{v}_1=\frac{\gamma^*\vmax_1+\vmax_2}{\gamma^*+1}\leq \vmax_1$ is the minimum reported value of bidder 1 such that she wins regardless of the second bidder's report.
\end{enumerate}
\end{proposition}

\subsection{Different lower bounds}
The main result and its proof is almost unchanged if the seller puts different lower bounds $\vmin_i$ on bidder's values that can also differ from the seller's own valuation $c$ (but the upper bound is still the same). %
Note, however, that in the ``gap case'' ($c<\max_i\vmin_i$) the parametric solution will be substantially different from the one identified in section \ref{paramsol}. As setting $r_i$  lower than $\vmin_i$ raises the worst-case probability of sale discontinuously to one, the seller will sometimes strictly prefer such ``sure-sale'' linear score auctions to the ones identified in the baseline case. The part of the above analysis that fails in this case is lemma \ref{rstaroptimal}, as now discontinuity in \eqref{mlambda} starts playing a role. %

\subsection{Inequality constraints}
Throughout the analysis, we have maintained the assumption that the mean of the valuation vector distribution is known exactly, i.e. the mean constraint is an equality constraint. This assumption has required additional work to rule out negative Lagrange multipliers in \textbf{Grand case II.} of the proof of the main result; in contrast, had we assumed that only a lower bound for the mean is known (as in \cite{koccyiugit2017distributionally}), we would get 
the nonnegativity of $\lambda$ for free. We think that the equality constraint is a more plausible modeling choice under the ``educated guess'' interpretation of the known mean assumption. Note, however, that the inequality constraint may be a better choice when the mean information comes from data obtained from previous auctions with the same bidders. It is well-recognized that in that case bidders, anticipating that their reports will affect the design of a future auction, may strategically shade their bids in an otherwise truthful mechanism\footnote{\cite{kanoria2017dynamic} propose an approximate solution to the incentive problem  for the case of iid values. It is less clear how to alleviate it in the case where bidders are ex-ante asymmetric and may have correlated values. }. Thus, such data will indeed provide only a lower bound on the valuations' means. It is therefore warranted to articulate the following result.

\begin{corollary}
The set of optimal mechanisms is the same regardless of whether the seller knows that $\mathbb{E}(v)=m$ or $\mathbb{E}(v)\geq m$.
\end{corollary}

\begin{proof}
Recall that under the equality constraint the transformed problem was to maximize $R(p,\lambda)$ given by \eqref{reducedproblem} over threshold functions $p$ and $\lambda\in\mathbb{R}^n$. Under the constraint  $\mathbb{E}(v)\geq m$ transformed problem becomes maximize $R(p,\lambda)$ over $p$ and $\lambda\in\mathbb{R}^n_+$. Then the result follows from the fact that any solution $(p^*,\lambda^*)$ to the original problem involves $\lambda^*\geq 0$ by lemma \ref{newlambdapositive}.
\eop
\end{proof}

It might be also interesting to consider the problem in which the seller knows that $\mathbb{E}(v_i)\in [\underline{m}_i,\overline{m}_i]$ for all $i$. It follows from our result that $\lambda^*(M^*)\geq 0$ for any optimal mechanism $M^*$ in the known means case, that in such a problem the set of optimal mechanisms is exactly the same as the set of optimal mechanisms for the case when the means are known to be $\underline{m}_i$.  This result is not trivial, as there exist mechanisms under which worst-case revenue may be lower for \emph{higher} $m_i$ (see the Online Appendix), and it is nor a priori clear that such mechanisms are suboptimal. 

\section{Conclusion}\label{concl}
In this paper, we presented a solution to a basic distributionally robust mechanism design problem -- the problem of allocating an indivisible good among $n$ buyers in a manner that maximizes worst-case expected revenue when the seller knows only means of value distributions and an upper bound on their support. The identified solution is simple and may be thought of as a linear version of the classic solution for the case where value distributions are known \citep{myerson1981optimal}. The proof is based on strong linear programming duality, a geometric construction and an analysis of the set of fixed points of a certain piecewise-affine map. We then solved the parameter-tuning problem and identified two regimes for the parametric solution. We compared and contrasted the identified solution to the classic one, and characterized the full set of solutions in the two-bidder case. The linearity of the boundary determining the bidder getting the object is indeed necessary for optimality if the known means are sufficiently high. 

The analysis is certainly not free of limitations that are simultaneously avenues for future research. Some of those are as follows:
\begin{itemize}
\item \textbf{Randomized mechanisms.} It is well-known that randomized solutions perform strictly better in robust optimization problems thanks to their ability to provide hedging against various Nature's strategies. Thus, our restriction to deterministic mechanisms is certainly with loss of revenue. In fact, it follows from our results and theorem 11 in \cite{koccyiugit2017distributionally} that 
\[\inf\limits_m\frac{R^*_{det}(m)}{R^*_{rand}(m)}=0,\]
where $R^*_{det}(m)$ and $R^*_{rand}(m)$ are the best revenue guarantees of a deterministic and randomized mechanism when all bidders' values have the same mean $m$. In this sense, the loss of revenue can be significant and it is certainly interesting to know what an optimal randomized mechanism is. However, there are two caveats to the above ratio analysis. First, it follows from the analysis in \cite{koccyiugit2017distributionally} that the infimum is achieved only when $m\to 0$, but then both $R^*_{det}(m)\to 0$ and $R^*_{rand}(m)\to 0$ which may not represent an interesting scenario: it is not of great practical importance that one revenue is infinitely smaller than the other if both are infinitely small. Second, and relatedly, it follows from results in \cite{carrasco2018optimal} that for the case of one bidder $\sup_m \left(R^*_{rand}(m)-R^*_{det}(m)\right)\approx 0.1331\vmax$. Thus, if there is a cost of $C>0.1331\vmax$ of using a randomized mechanism (say, due to its complexity --- and numerical analysis shows that it is quite complex), it will be an optimal decision for the seller to use a deterministic mechanism regardless of mean $m$ even if randomization is feasible. We expect a similar argument hold for the case of multiple bidders as well.%

\item \textbf{Higher moments.} The proof of the main result in the present paper cannot be easily modified to accommodate higher moments constraints. It follows again from strong duality  that the natural candidate optimal threshold functions if $k$ moments are known are  polynomials of degree $k$ (see \cite{carrasco2018optimal} for the case of one bidder). Note, however, that using such polynomials on the full domain is very likely to be infeasible due to the supply constraint. To resolve the conflict and determine the winner of the object it may be necessary to draw a boundary that may be linear still. One natural higher-moment constraint is the nonnegative covariance constraint. Because the function $v_{2}\to v_1v_{2}$ is linear, a similar tilde-transformation can be used to establish the optimality of a linear score auction under known mean and nonnegative covariance if there are two bidders. However, it is unclear whether the result is true for $n\geq 3$ (due to the presence of terms $p_i(v_i)-\zeta v_jv_k$, $i\neq j\neq k$, in the dual problem).
\end{itemize}

\bibliographystyle{te}

\bibliography{bibliography}

\begin{thebibliography}{44}
\newcommand{\enquote}[1]{``#1''}
\providecommand{\natexlab}[1]{#1}
\providecommand{\url}[1]{\texttt{#1}}
\providecommand{\urlprefix}{URL }
\providecommand{\bibAnnoteFile}[1]{%
  \IfFileExists{#1}{\begin{quotation}\noindent\textsc{Key:} #1\\
  \textsc{Annotation:}\ \input{#1}\end{quotation}}{}}
\providecommand{\bibAnnote}[2]{%
  \begin{quotation}\noindent\textsc{Key:} #1\\
  \textsc{Annotation:}\ #2\end{quotation}}

\bibitem[{Allouah and Besbes(2020)}]{allouah2020prior}
Allouah, Amine and Omar Besbes (2020), \enquote{Prior-independent optimal
  auctions.} \emph{Management Science}.
\bibAnnoteFile{allouah2020prior}

\bibitem[{Auster(2018)}]{auster2018robust}
Auster, Sarah (2018), \enquote{Robust contracting under common value
  uncertainty.} \emph{Theoretical Economics}, 13, 175--204.
\bibAnnoteFile{auster2018robust}

\bibitem[{Azar et~al.(2013)Azar, Daskalakis, Micali, and
  Weinberg}]{azar2013optimal}
Azar, Pablo, Constantinos Daskalakis, Silvio Micali, and S~Matthew Weinberg
  (2013), \enquote{Optimal and efficient parametric auctions.} In
  \emph{Proceedings of the twenty-fourth annual ACM-SIAM symposium on Discrete
  algorithms}, 596--604, Society for Industrial and Applied Mathematics.
\bibAnnoteFile{azar2013optimal}

\bibitem[{Azar and Micali(2013)}]{azar2013parametric}
Azar, Pablo~Daniel and Silvio Micali (2013), \enquote{Parametric digital
  auctions.} In \emph{Proceedings of the 4th conference on Innovations in
  Theoretical Computer Science}, 231--232, ACM.
\bibAnnoteFile{azar2013parametric}

\bibitem[{Bergemann et~al.(2016)Bergemann, Brooks, and
  Morris}]{bergemann2016informationally}
Bergemann, Dirk, Benjamin~A Brooks, and Stephen Morris (2016),
  \enquote{Informationally robust optimal auction design.}
\bibAnnoteFile{bergemann2016informationally}

\bibitem[{Bergemann and Schlag(2011)}]{bergemann2011robust}
Bergemann, Dirk and Karl Schlag (2011), \enquote{Robust monopoly pricing.}
  \emph{Journal of Economic Theory}, 146, 2527--2543.
\bibAnnoteFile{bergemann2011robust}

\bibitem[{Bergemann and Schlag(2008)}]{bergemann2008pricing}
Bergemann, Dirk and Karl~H Schlag (2008), \enquote{Pricing without priors.}
  \emph{Journal of the European Economic Association}, 6, 560--569.
\bibAnnoteFile{bergemann2008pricing}

\bibitem[{Bonnans and Shapiro(2000)}]{bonnans2000perturbation}
Bonnans, J~Frederic and Alexander Shapiro (2000), \emph{Perturbation Analysis
  of Optimization Problems}. Springer Science \& Business Media.
\bibAnnoteFile{bonnans2000perturbation}

\bibitem[{Bose et~al.(2006)Bose, Ozdenoren, and Pape}]{bose2006optimal}
Bose, Subir, Emre Ozdenoren, and Andreas Pape (2006), \enquote{Optimal auctions
  with ambiguity.} \emph{Theoretical Economics}, 1, 411--438.
\bibAnnoteFile{bose2006optimal}

\bibitem[{Brooks(2013)}]{brooks2013surveying}
Brooks, Benjamin~A (2013), \enquote{Surveying and selling: Belief and surplus
  extraction in auctions.} \emph{Unpublished manuscript}.
\bibAnnoteFile{brooks2013surveying}

\bibitem[{Brooks and Du(2019)}]{brooks2019optimal}
Brooks, Benjamin~A and Songzi Du (2019), \enquote{Optimal auction design with
  common values: An informationally-robust approach.} \emph{Available at SSRN
  3137227}.
\bibAnnoteFile{brooks2019optimal}

\bibitem[{Carrasco et~al.(2018{\natexlab{a}})Carrasco, Luz, Kos, Messner,
  Monteiro, and Moreira}]{carrasco2018optimal}
Carrasco, Vinicius, Vitor~Farinha Luz, Nenad Kos, Matthias Messner, Paulo
  Monteiro, and Humberto Moreira (2018{\natexlab{a}}), \enquote{Optimal selling
  mechanisms under moment conditions.} \emph{Journal of Economic Theory}, 177,
  245--279.
\bibAnnoteFile{carrasco2018optimal}

\bibitem[{Carrasco et~al.(2018{\natexlab{b}})Carrasco, Luz, Monteiro, and
  Moreira}]{carrasco2018robust}
Carrasco, Vinicius, Vitor~Farinha Luz, Paulo~K Monteiro, and Humberto Moreira
  (2018{\natexlab{b}}), \enquote{Robust mechanisms: the curvature case.}
  \emph{Economic Theory}, 1--20.
\bibAnnoteFile{carrasco2018robust}

\bibitem[{Carroll(2017)}]{carroll2017robustness}
Carroll, Gabriel (2017), \enquote{Robustness and separation in multidimensional
  screening.} \emph{Econometrica}, 85, 453--488.
\bibAnnoteFile{carroll2017robustness}

\bibitem[{Carroll(2018)}]{carroll2018robustness}
Carroll, Gabriel (2018), \enquote{Robustness in mechanism design and
  contracting.} \emph{Annual Reviews}.
\bibAnnoteFile{carroll2018robustness}

\bibitem[{Che(2019)}]{che2019robust}
Che, Ethan (2019), \enquote{Robust reserve pricing in auctions under mean
  constraints.} \emph{Available at SSRN 3488222}.
\bibAnnoteFile{che2019robust}

\bibitem[{Che(1993)}]{che1993design}
Che, Yeon-Koo (1993), \enquote{Design competition through multidimensional
  auctions.} \emph{The RAND Journal of Economics}, 668--680.
\bibAnnoteFile{che1993design}

\bibitem[{Chen et~al.(2019)Chen, Hu, and Perakis}]{chen2019distribution}
Chen, Hongqiao, Ming Hu, and Georgia Perakis (2019), \enquote{Distribution-free
  pricing.} \emph{Available at SSRN 3090002}.
\bibAnnoteFile{chen2019distribution}

\bibitem[{Chung and Ely(2007)}]{chung2007foundations}
Chung, Kim-Sau and Jeffrey~C Ely (2007), \enquote{Foundations of
  dominant-strategy mechanisms.} \emph{The Review of Economic Studies}, 74,
  447--476.
\bibAnnoteFile{chung2007foundations}

\bibitem[{Cremer and McLean(1988)}]{cremer1988full}
Cremer, Jacques and Richard~P McLean (1988), \enquote{Full extraction of the
  surplus in bayesian and dominant strategy auctions.} \emph{Econometrica:
  Journal of the Econometric Society}, 1247--1257.
\bibAnnoteFile{cremer1988full}

\bibitem[{Delage and Ye(2010)}]{delage2010distributionally}
Delage, Erick and Yinyu Ye (2010), \enquote{Distributionally robust
  optimization under moment uncertainty with application to data-driven
  problems.} \emph{Operations research}, 58, 595--612.
\bibAnnoteFile{delage2010distributionally}

\bibitem[{Dhangwatnotai et~al.(2015)Dhangwatnotai, Roughgarden, and
  Yan}]{dhangwatnotai2015revenue}
Dhangwatnotai, Peerapong, Tim Roughgarden, and Qiqi Yan (2015),
  \enquote{Revenue maximization with a single sample.} \emph{Games and Economic
  Behavior}, 91, 318--333.
\bibAnnoteFile{dhangwatnotai2015revenue}

\bibitem[{Du(2018)}]{du2018robust}
Du, Songzi (2018), \enquote{Robust mechanisms under common valuation.}
  \emph{Econometrica}, 86, 1569--1588.
\bibAnnoteFile{du2018robust}

\bibitem[{Giannakopoulos et~al.(2019)Giannakopoulos, Po{\c{c}}as, and
  Tsigonias-Dimitriadis}]{giannakopoulos2019robust}
Giannakopoulos, Yiannis, Diogo Po{\c{c}}as, and Alexandros
  Tsigonias-Dimitriadis (2019), \enquote{Robust revenue maximization under
  minimal statistical information.} \emph{arXiv preprint arXiv:1907.04220}.
\bibAnnoteFile{giannakopoulos2019robust}

\bibitem[{Gilboa and Schmeidler(1989)}]{gilboa1989maxmin}
Gilboa, Itzhak and David Schmeidler (1989), \enquote{Maxmin expected utility
  with non-unique prior.} \emph{Journal of Mathematical Economics}, 18,
  141--153.
\bibAnnoteFile{gilboa1989maxmin}

\bibitem[{Goh and Sim(2010)}]{goh2010distributionally}
Goh, Joel and Melvyn Sim (2010), \enquote{Distributionally robust optimization
  and its tractable approximations.} \emph{Operations research}, 58, 902--917.
\bibAnnoteFile{goh2010distributionally}

\bibitem[{Hartline and Roughgarden(2009)}]{hartline2009simple}
Hartline, Jason~D and Tim Roughgarden (2009), \enquote{Simple versus optimal
  mechanisms.} In \emph{Proceedings of the 10th ACM conference on Electronic
  commerce}, 225--234.
\bibAnnoteFile{hartline2009simple}

\bibitem[{He and Li(2020)}]{he2019robustly}
He, Wei and Jiangtao Li (2020), \enquote{Correlation-robust auction design.}
\bibAnnoteFile{he2019robustly}

\bibitem[{Kamenica and Gentzkow(2011)}]{kamenica2011bayesian}
Kamenica, Emir and Matthew Gentzkow (2011), \enquote{Bayesian persuasion.}
  \emph{American Economic Review}, 101, 2590--2615.
\bibAnnoteFile{kamenica2011bayesian}

\bibitem[{Kanoria and Nazerzadeh(2017)}]{kanoria2017dynamic}
Kanoria, Yash and Hamid Nazerzadeh (2017), \enquote{Dynamic reserve prices for
  repeated auctions: Learning from bids.} \emph{Available at SSRN 2444495}.
\bibAnnoteFile{kanoria2017dynamic}

\bibitem[{Ko{\c{c}}yi{\u{g}}it et~al.(2020)Ko{\c{c}}yi{\u{g}}it, Iyengar, Kuhn,
  and Wiesemann}]{koccyiugit2017distributionally}
Ko{\c{c}}yi{\u{g}}it, {\c{C}}a{\u{g}}{\i}l, Garud Iyengar, Daniel Kuhn, and
  Wolfram Wiesemann (2020), \enquote{Distributionally robust mechanism design.}
  \emph{Management Science}, 66, 159--189.
\bibAnnoteFile{koccyiugit2017distributionally}

\bibitem[{Loertscher and Marx(2020)}]{loertscher2020asymptotically}
Loertscher, Simon and Leslie~M Marx (2020), \enquote{Asymptotically optimal
  prior-free clock auctions.} \emph{Journal of Economic Theory}, 105030.
\bibAnnoteFile{loertscher2020asymptotically}

\bibitem[{Myerson(1981)}]{myerson1981optimal}
Myerson, Roger~B (1981), \enquote{Optimal auction design.} \emph{Mathematics of
  operations research}, 6, 58--73.
\bibAnnoteFile{myerson1981optimal}

\bibitem[{Neeman(2003)}]{neeman2003effectiveness}
Neeman, Zvika (2003), \enquote{The effectiveness of english auctions.}
  \emph{Games and Economic Behavior}, 43, 214--238.
\bibAnnoteFile{neeman2003effectiveness}

\bibitem[{Papadimitriou and Pierrakos(2011)}]{papadimitriou2011optimal}
Papadimitriou, Christos~H and George Pierrakos (2011), \enquote{On optimal
  single-item auctions.} In \emph{Proceedings of the forty-third annual ACM
  symposium on Theory of computing}, 119--128.
\bibAnnoteFile{papadimitriou2011optimal}

\bibitem[{Popescu(2007)}]{popescu2007robust}
Popescu, Ioana (2007), \enquote{Robust mean-covariance solutions for stochastic
  optimization.} \emph{Operations Research}, 55, 98--112.
\bibAnnoteFile{popescu2007robust}

\bibitem[{Scarf(1958)}]{scarf1958min}
Scarf, Herbert (1958), \enquote{A min-max solution of an inventory problem.}
  \emph{Studies in the mathematical theory of inventory and production}.
\bibAnnoteFile{scarf1958min}

\bibitem[{See and Sim(2010)}]{see2010robust}
See, Chuen-Teck and Melvyn Sim (2010), \enquote{Robust approximation to
  multiperiod inventory management.} \emph{Operations research}, 58, 583--594.
\bibAnnoteFile{see2010robust}

\bibitem[{Segal(2003)}]{segal2003optimal}
Segal, Ilya (2003), \enquote{Optimal pricing mechanisms with unknown demand.}
  \emph{American Economic Review}, 93, 509--529.
\bibAnnoteFile{segal2003optimal}

\bibitem[{Smith(1995)}]{smith1995generalized}
Smith, James~E (1995), \enquote{Generalized chebychev inequalities: theory and
  applications in decision analysis.} \emph{Operations Research}, 43, 807--825.
\bibAnnoteFile{smith1995generalized}

\bibitem[{Suzdaltsev(2020)}]{suzdal2020distributionally}
Suzdaltsev, Alex (2020), \enquote{Distributionally robust pricing in
  independent private value auctions.} \emph{Working paper}.
\bibAnnoteFile{suzdal2020distributionally}

\bibitem[{Wiesemann et~al.(2014)Wiesemann, Kuhn, and
  Sim}]{wiesemann2014distributionally}
Wiesemann, Wolfram, Daniel Kuhn, and Melvyn Sim (2014),
  \enquote{Distributionally robust convex optimization.} \emph{Operations
  Research}, 62, 1358--1376.
\bibAnnoteFile{wiesemann2014distributionally}

\bibitem[{Wilson(1987)}]{wilson1987game}
Wilson, Robert (1987), \enquote{Game theoretic approaches to trading
  processessin tru (man bewley, ed., advances in economic theory: Fifth world
  congress.}
\bibAnnoteFile{wilson1987game}

\bibitem[{Wolitzky(2016)}]{wolitzky2016mechanism}
Wolitzky, Alexander (2016), \enquote{Mechanism design with maxmin agents:
  Theory and an application to bilateral trade.} \emph{Theoretical Economics},
  11, 971--1004.
\bibAnnoteFile{wolitzky2016mechanism}

\end{thebibliography}

\section*{Appendix: main missing proofs}\label{app}
\pfof{lemma \ref{convex hull}}{
As $m\in(0,\vmax)^n$, results of \cite{smith1995generalized} apply. In particular, $\inf\limits_{F\in\Delta(m,\vmax)}R(M,F)=\inf\limits_{F\in\Delta_{fin}(m,\vmax)}R(M,F)$ where $\Delta_{fin}(m,\vmax)$ is a subset of $\Delta(m,\vmax)$ consisting of all distributions with finite support. So it is sufficient to prove that $\inf\{R|(m,R)\in conv(graph(t))\}=\inf\limits_{F\in\Delta_{fin}(m,\vmax)}R(M,F)$. But this is certainly true, as one of definitions of a convex hull of a subset $S$ of a Euclidean space is the set of all finite convex combinations of points in $S$.
\eop}

\pfof{lemma \ref{getridofO}}{
Consider some bidder $i$ and some profile of values $v^0\in W_i^{\geq}(p)\setminus W_i(p)$. It must be that $v^0_i=p_i(v^0_{-i})$. Either $p_i(v^0_{-i})<\vmax$ or $p_i(v^0_{-i})=\vmax$. Suppose first that $p_i(v^0_{-i})<\vmax$. Then, any neighborhood of $v^0$ has a nonempty intersection with $\{v:v_i>p_i(v_{-i})\}$, and hence, with $W_i(p)$. Thus, $p_i(v^0_{-i})-\lambda v^0\geq \inf\limits_{v\in W_i(p)}(p_i(v_{-i})-\lambda v)$ which means that adding $v^0$ to $W_i(p)$ won't change the value of \eqref{transformedproblem}.

Now suppose that $p_i(v^0_{-i})=\vmax$. Because $v^0\in W_i^{\geq}(p)\setminus W_i(p)$, $v^0\in\{v:v_i=p_i(v_{-i})\}$ and thus $v^0\in\cup_{i=1}^n\{v:v_i=p_i(v_{-i})\}$. Also, $v^0\notin O_i$. Suppose first that $v^0\in \cup_{i=1}^n\{v:v_i=p_i(v_{-i})\}\setminus\cup_{i=1}^n\{v:v_i>p_i(v_{-i})\}$. By condition 3, $v^0\in \cup_{i=1}^n O_i$ and hence $v^0\in O_j$ for some $j\neq i$. Hence, $v^0\in W_j(p)$ for some $j\neq i$. Now suppose that $v^0\in \cup_{i=1}^n\{v:v_i>p_i(v_{-i})\}$. Again, as $v^0_i=p_i(v^0_{-i})$, it means that $v^0\in W_j(p)$ for some $j\neq i$. This ensures the second inequality in the following chain:
\beq\label{cangetridofo}
 p_i(v^0_{-i})-\lambda v^0=\vmax-\lambda v^0\geq p_j(v^0_j)-\lambda v^0\geq \inf\limits_{v\in W_j(p)}(p_j(v_{-j})-\lambda v),
\eeq
Thus, adding $v^0$ to $W_i(p)$ won't change the value of \eqref{transformedproblem}.\eop
}

\pfof{proposition \ref{LSAbetter2}}{
Consider four sets of bidders: (i) those with $v^*_i=\vmax$; (ii) those with $v^*_i\in(0,\vmax)$; (iii) those with $v^*_i=0$ and $\lambda^*_{-i}(\vmax\i_{n-1}-v^*_{-i})>0$; (iv) those with $v^*_i=0$ and $\lambda^*_{-i}(\vmax\i_{n-1}-v^*_{-i})=0$.

We prove that for every type of bidder, $\hat{p}_i\geq \tilde{p}_i$. 

First, if $v^*_i=\vmax$, then $\hat{p}_i(v_{-i})\equiv\vmax\geq \tilde{p}_i(v_{-i})$ $\forall v_{-i}$.

Second, if $v^*_i\in(0,\vmax)$, then the representation \eqref{fp-representation} is valid for $\tilde{p}_i$ and thus Lemma \ref{LSAfunctionsgreater} is applicable directly, with $\chi^*_{-i}=\lambda^*_{-i}$, as in Case 1. Thus, $\hat{p}_i\geq \tilde{p}_i$. 

Third, suppose $v_i^*=0$ and $\lambda^*_{-i}(\vmax\i_{n-1}-v^*_{-i})>0.$ Define 
\beq
k_i:=\frac{\lambda^*_{-i}\vmax\i_{n-1}+b_i}{\lambda^*_{-i}(\vmax\i_{n-1}-v^*_{-i})}.
\eeq
We prove that, with such a choice of $k_i$, $p^{aux}_i(v_{-i})\geq \tilde{p}_i(v_{-i})$ at any point $v_{-i}$. 

First, for $v_{-i}$ such that $\lambda^*_{-i}v_{-i}<\lambda^*_{-i}v^*_{-i}$, $p^{aux}_i(v_{-i})=0$, so the inequality holds.

Now consider $v_{-i}$ such that $\lambda^*_{-i}v_{-i}\geq\lambda^*_{-i}v^*_{-i}$. Any such point can be represented as a convex combination of $\vmax\i_{n-1}$ and a point $u$ satisfying $\lambda^*_{-i}u=\lambda^*_{-i}v^*_{-i}$. Namely,
\[
v_{-i}=t\cdot u +(1-t)\cdot\vmax\i_{n-1},
\]
where 
\[t=\frac{\lambda^*_{-i}(\vmax\i_{n-1}-v_{-i})}{\lambda^*_{-i}(\vmax\i_{n-1}-v^*_{-i})}\] 
and \[u=\vmax\i_{n-1}\frac{\lambda^*_{-i}(v^*_{-i}-v_{-i})}{\lambda^*_{-i}(\vmax\i_{n-1}-v_{-i})}
+ v_{-i}\frac{\lambda^*_{-i}(\vmax\i_{n-1}-v^*_{-i})}{\lambda^*_{-i}(\vmax\i_{n-1}-v_{-i})}.\] 
$t\leq 1$ due to the fact that $\lambda^*_{-i}v_{-i}\geq\lambda^*_{-i}v^*_{-i}$.

$\tilde{p}_i$ is a convex function and thus 
\beq\label{ineq_convex}
\tilde{p}_i(v_{-i})=\tilde{p}_i(t\cdot u +(1-t)\cdot\vmax\i_{n-1})\leq t\tilde{p}_i(u)+(1-t)\tilde{p}_i(\vmax\i_{n-1}).
\eeq
However, $\tilde{p}_i(u)=\max\{\lambda^*_{-i}u+b_i,0\}=\max\{\lambda^*_{-i}v^*_{-i}+b_i,0\}=v_i^*=0$. Moreover, 
\[(1-t)\tilde{p}_i(\vmax\i_{n-1})=\max\left\{\frac{\lambda^*_{-i}(v_{-i}-v^*_{-i})}{\lambda^*_{-i}(\vmax\i_{n-1}-v^*_{-i})}(\lambda^*_{-i}\vmax\i_{n-1}+b_i),0\right\}=p^{aux}_i(v_{-i}).\]
Combining this with \eqref{ineq_convex} yields the desired inequality $p^{aux}_i\geq \tilde{p}_i$.

The second inequality $\hat{p}_i\geq p^{aux}_i$ follows from lemma \ref{LSAfunctionsgreater} with $\chi^*_{-i}=k_i\lambda^*_{-i}$. Note that the supposition of the Lemma \ref{LSAfunctionsgreater} holds due to the fact that $p^{aux}_i(\vmax\i_{n-1})=\tilde{p}_i(\vmax\i_{n-1})$ by construction and $\tilde{p}_i(\vmax\i_{n-1})\leq \vmax$.

Hence, $\hat{p}_i\geq p^{aux}_i\geq \tilde{p}_i$ as desired.

Fourth, suppose $v_i^*=0$ and $\lambda^*_{-i}(\vmax\i_{n-1}-v^*_{-i})=0.$ As every term $\lambda^*_j(\vmax-v^*_j)$ is nonnegative, we must have $\lambda^*_j(\vmax-v^*_j)=0$ for all $j\neq i$.

Denote by $J_{max}$ the set of indices $j\neq i$ such that $v^*_j=\vmax$. Then it follows that 
\[0=v^*_i=\max\{\sum_{j\in J_{max}}\lambda^*_j\vmax+b_i,0\}.\]
But then, due to nonnegativity of $\lambda^*_j$, 
\[\tilde{p}_i(v_{-i})\leq \max\{\sum_{j\in J_{max}}\lambda^*_j\vmax+b_i,0\}=0.\]
Hence, $\tilde{p}_i(v_{-i})\equiv 0$, and thus the inequality $\hat{p}_i\geq \tilde{p}_i$ holds.

\vspace{1em}
We have shown that for every $i$, $\hat{p}_i\geq \tilde{p}_i$. This implies that all $\inf\limits_{v_{-i}\in[0,\vmax]^{n-1}}(p_i(v_{-i})-\lambda^*_{-i}v_{-i})$ increase when $\tilde{p}_i$ is replaced by $\hat{p}_i$. To cover the remaining infimum, $\inf\limits_{v\in W_0(p)}(-\lambda^* v)$, note that 
\beq\label{W_0comparison}
\inf\limits_{v\in W_0(\hat{p})}(-\lambda^* v)=+\infty>\inf\limits_{v\in W_0(\tilde{p})}(-\lambda^* v),
\eeq
as $W_0(\hat{p})=\emptyset$ due to the fact that $v^*_i=0$ for some $i$, and so the object is always allocated in LSA($v^*$).
\eop
}

\pfof{proposition \ref{newfixedpoint}}{
The fact that $w^*_i=0$ for some $i$ follows from $s^*<n$.

Now we show that $w^*\in [0,\vmax]^n$. $w^*\geq 0$ by the fact that $b^{restr}\geq 0$ and Lemma \ref{Aproperties}, part 2. Due to the fact that $\vmax\i_n\in V^*$, it must be that
\[(A^{restr})^{-1}\vmax\i_{s^*}-\vmax\i_{s^*}\sum_{i>s^*}\lambda^*_i=b^{restr}\]
or
\[\vmax\i_{s^*}=(A^{restr})^{-1}\left(b^{restr}+\vmax\i_{s^*}\sum_{i>s^*}\lambda^*_i\right),\]
so again by Lemma \ref{Aproperties}, part 2,  $(A^{restr})^{-1}b^{restr}\leq \vmax\i_{s^*}$, so $w^*\leq \vmax\i_{n}$. 

By construction, $w^*$ satisfies first $s^*$ equations of the system $v_i=\tilde{p}_i(v_{-i})$. It remains to show that it satisfies the equations $s^*+1,\ldots, n$, that is
\[0=\max\left\{\sum_{i=1}^{s^*}\lambda^*_i w^*_i+b_j,0\right\},\]
for $j>s^*$. By the ranking of $b_i$, it is sufficient to show that 
\[\sum_{i=1}^{s^*}\lambda^*_i w^*_i\leq -b_{s^*+1}.\] 
Writing the LHS using Lemma \ref{Aproperties}, part 3, we get
\[\frac{\sum\limits_{i=1}^{s^*} \frac{\lambda^*_i}{1+\lambda^*_i}b_i}{1-\sum\limits_{i=1}^{s^*} \frac{\lambda^*_i}{1+\lambda^*_i}}\leq -b_{s^*+1}.\] 
Substituting for $b_i$ from \eqref{biexpr}, after simplifications one gets
\beq\label{32}
\lambda^*_{s^*+1}\geq \frac{\sum\limits_{i=1}^{s^*} \frac{\lambda^{*2}_i}{1+\lambda^*_i}-\sum\limits_{i\neq s^*+1}\lambda^*_i+1}{\sum\limits_{i=1}^{s^*} \frac{\lambda^*_i}{1+\lambda^*_i}}.
\eeq
But \eqref{32} is true because 
\[\lambda^*_{s^*+1}>\frac{1}{\sum\limits_{i\neq s^*+1} \frac{\lambda^*_i}{1+\lambda^*_i}}-1\geq \frac{\sum\limits_{i=1}^{s^*} \frac{\lambda^{*2}_i}{1+\lambda^*_i}-\sum\limits_{i\neq s^*+1}\lambda^*_i+1}{\sum\limits_{i=1}^{s^*} \frac{\lambda^*_i}{1+\lambda^*_i}},\]
where the first inequality is equivalent to $\sum\limits_{i=1}^n \frac{\lambda^*_i}{1+\lambda^*_i}>1$, the supposition in \emph{Subcase 2}, and the second is equivalent to 
\[\sum_{i=1}^{s^*}\lambda^*_i\leq \sum_{i\neq s^*+1}\lambda^*_i. \]
\eop 
}

\renewcommand*{\thesection}{\Alph{section}}
\setcounter{section}{0}
\newpage
\begin{center}
\Large Online Appendix
\end{center}
\section{More missing proofs}
\pfof{proposition \ref{optauction exists}}{
Denote by $R(r,\lambda)$ the value of \eqref{reducedproblem} when the threshold functions $p$ are given by \eqref{p-lsa} for some $r\in[0,\vmax]^n$. The revenue guarantee of the auction is given by $\underline{R}(r):=\max\limits_{\lambda\in\mathbb{R}^n}R(r,\lambda)$ (supremum is achieved by lemma \ref{optexistlambda}). We will prove that $\underline{R}(r)$ is upper semi-continuous, which implies the result. Note that $R(r,\lambda)$ is continuous on $(0,\vmax]^n\times \mathbb{R}^n$, 
so, by Berge's Maximum theorem argument $\underline{R}(r)$ is continuous on $(0,\vmax]^n$. Also, $R(r,\lambda)$ is continuous on $\mathcal{R}_0\times \mathbb{R}^n$ (when its domain is restricted to $\mathcal{R}_0\times \mathbb{R}^n$ ) so $\underline{R}(r)$ is continuous on $\mathcal{R}_0$ when its domain is restricted to $\mathcal{R}_0$. Hence, to show that $\underline{R}(r)$ is upper semi-continuous on full domain, it is sufficient to show that for every sequence of points $r^k\in(0,\vmax]^n$ with limit $r\in\mathcal{R}_0$, we have $\lim\limits_{k\to\infty}\underline{R}(r^k)\leq \underline{R}(r)$.   

To this end, extend the function $R(r,\lambda)$ from $(0,\vmax]^n$ to the whole $[0,\vmax]^n$ by continuity and call the extension $R^{ext}(r,\lambda)$ with $\underline{R}^{ext}(r)=\max\limits_{\lambda\in\mathbb{R}^n}R^{ext}(r,\lambda)$. Note that $R^{ext}(r,\lambda)\leq R(r,\lambda)$ and thus $\underline{R}^{ext}(r)\leq \underline{R}(r)$ if $r\in\mathcal{R}_0$, and $R^{ext}(r,\lambda)= R(r,\lambda)$, $\underline{R}^{ext}(r)= \underline{R}(r)$  otherwise. By the usual argument,  $\underline{R}^{ext}(r)$ is continuous. Thus, $\lim\limits_{k\to\infty}\underline{R}(r^k)=\lim\limits_{k\to\infty}\underline{R}^{ext}(r^k)=\underline{R}^{ext}(r)\leq \underline{R}(r)$.
\eop}

\pfof{lemma \ref{optexistlambda}}{
Denote the function being maximized in \eqref{biconj} by $G(\lambda)$: $\mathbb{R}^n\to\mathbb{R}$. This function is concave; thus it is continuous. First, observe that for any unbounded sequence $\lambda^k$, the sequence $G(\lambda^k)$ is unbounded from below. To see this, note that if $\lambda^k_i$ is not bounded from above, one can set $v_i=\vmax$ in minimization in \eqref{biconj}, whereas if $\lambda^k_i$ is not bounded from below, one can set $v_i=0$; hence, $G(\lambda^k)$ is majorized by a sequence that is unbounded from below. 

Now take any maximizing sequence $\lambda^k$ (i.e. a sequence such that $\lim\limits_{k\to\infty}G(\lambda^k)=\sup\limits_{\lambda}G(\lambda)$). As the sequence $G(\lambda^k)$ is clearly bounded from below, by the above observation we know that $\lambda^k$ must lie in a bounded set $Q$. Therefore it has a subsequence $\lambda^{k_s}$ converging to some $\lambda^*\in cl(Q)\subset \mathbb{R}^n$. But then, by continuity of $G$, $\sup\limits_{\lambda}G(\lambda)=\lim\limits_{s\to\infty}G(\lambda^{k_s})=G(\lambda^*)$. \eop
}

\pfof{lemma \ref{Aproperties}}{
Consider the homogeneous system $Av=0$. By subtracting row $j>1$ from row 1 one gets $(1+\lambda_1)v_1=(1+\lambda_j)v_j$ for all $j>1$. Plugging this in equation 1, one gets $\left(1-\sum_{i=1}^n\frac{\lambda_i}{1+\lambda_i}\right)v_1=0$. Thus, the system has a unique (trivial) solution if and only if  $1\neq \sum_{i=1}^n\frac{\lambda_i}{1+\lambda_i}$. Thus, $\det(A)=0$ if and only if $1= \sum_{i=1}^n\frac{\lambda_i}{1+\lambda_i}$. Because determinant is a multilinear function of rows of $A$,  it must be equal to $C\left(1-\sum\limits_{i=1}^n \frac{\lambda_i}{1+\lambda_i}\right)\prod\limits_{i=1}^n(1+\lambda_i)$ for some $C\in\mathbb{R}$. As $\det(A(0))=1$, $C=1$. This proves part 1. It also follows that the set of solutions to the system is at most one-dimensional. This proves part 4.

To prove parts 2 and 3, consider a non-homogeneous system $Av=b$. By subtracting rows as above one obtains that its unique (provided that $\det(A)\neq 0$) solution $v(b)$ is given by 
\[v_i(b)=\frac{b_i\left(1-\sum\limits_{k\neq i}\frac{\lambda_k}{1+\lambda_k}\right)+\sum\limits_{k\neq i}\frac{\lambda_k}{1+\lambda_k}b_k}{(1+\lambda_i)\left(1-\sum\limits_{i=1}^n \frac{\lambda_i}{1+\lambda_i}\right)}.\]
As $A^{-1}_{ij}=\partial v_i/\partial b_j$, part 3 follows from $\lambda\geq 0$ and $\det(A)>0$. Part 4 follows from the above analytic solution and the fact that $\lambda v(b)=(1+\lambda_1)v_1(b)-b_1$. 
\eop}

\pfof{lemma \ref{bnegative}}{The proof is by solving the optimization problem $\max\limits_{\lambda}\sum_{i=1}^n\frac{\lambda_i}{1+\lambda_i}$ subject to \eqref{Sums} and the nonnegativity constraint. The solution is $\lambda_i=\lambda_j=\frac{1}{n-1}$.\eop}

\pfof{lemma \ref{innerproblsol}}{
$\inf\limits_{W_0(p)}(-\lambda v)=-\lambda r$ whenever $W_0(p)\neq\emptyset$. Now consider the problem $\inf\limits_{v_{-i}\in [0,\vmax]^{n-1}}(p_i(v_{-i})-\lambda_{-i}v_{-i})$. Consider the minimization over some $v_j$, $j\neq i$ with $v_s$, $s\notin\{i,j\}$ fixed. Given \eqref{p-lsa}, the objective is convex and piecewise-affine and $\lambda_j\geq 0$, so the optimal choice of $v_j$ is either $\vmax$ (if $\lambda_j\geq \frac{\vmax-r_i}{\vmax-r_j}$) or such that $\frac{v_j-r_j}{\vmax-r_j}=\max\{0,\max\limits_{s\neq i,j}\frac{v_s-r_s}{\vmax-r_s}\}$ (otherwise). But this is true for every $j$, $v_j=\vmax$ for some $j$ implies that $v_s=\vmax$ for all $s\neq i$. If, on the other hand, for every $j$ the second possibility materializes, we obtain that $\frac{v_j-r_j}{\vmax-r_j}=\eta$, $\eta\in[0,1]$ for all $j\neq i$. Note in  the first case ($v_j=\vmax$ for all $j\neq i$) this condition is also satisfied. Thus, it is sufficient to optimize over $\eta$. Given the linearity of the objective, the two possible corner solutions will yield either an infimum of $r_i-\lambda_{-i}r_{-i}$ or $\vmax-\vmax\sum\limits_{j\neq i}\lambda_j$ from which the expression \eqref{mlambda} follows.
\eop}

\pfof{Lemma \ref{rstaroptimal}}{

For brevity, we will write $r^*$ instead of $r^*(\lambda)$.

Define 
\beq\label{g}
g(r):=\begin{cases}
\min\left\{\min\limits_i(r_i-\lambda_{-i}r_{-i}-\lambda_i\vmax),-\lambda r \right\}, & r>0,\\
\min\limits_i(r_i-\lambda_{-i}r_{-i}-\lambda_i\vmax), & \mbox{o/w}.
\end{cases}
\eeq
We will prove that if $\sum_i\frac{\lambda_i}{1+\lambda_i}\leq 1$, $r^*$ maximizes $g(r)$. Then it will follow that $r^*$ maximizes $R(r,\lambda)$. 

Note that $g(r)$ is discontinuous at any point of the set $\mathcal{R}_0:=\{r\in[0,\vmax]^n:r_i=0\mbox{ for some }i\}$. We will show directly that $g(r^*)\geq g(r)$ for any other point $r\in[0,\vmax]^n$. To this end, we will construct a finite sequence of points starting with any point $r\in(0,\vmax]^n$ and ending with $r^*$ such that the value of $g$ weakly increases at every step. Denote $E_i(r):=r_i-\lambda_{-i}r_{-i}-\lambda_i\vmax$. We will proceed by considering two cases: (1) the starting point $r\notin \mathcal{R}_0$; (2) $r\in\mathcal{R}_0$. 

\textbf{Case 1.} The starting point $r\notin \mathcal{R}_0$. Note that if $r_i>r^*_i$ for some $i$, $g(r)$ is improved by changing $r_i$ to $r^*_i$. Indeed, $r_i>r^*_i$ is equivalent to $E_i(r)>-\lambda r$, so  $E_i$ is higher than the overall minimum in the first line in \eqref{g}. Decreasing $r_i$ to $r^*_i$ will weakly increase $-\lambda r$ and $E_j(r)$ for all $j\neq i$, thus weakly increasing the value of $g$. If the new point $\tilde{r}$ is in $\mathcal{R}_0$, proceed directly to \textbf{Case 2} below. If after coordinates $r_i$ such that $r_i>r_i^*$ were all lowered to $r_i^*$, we are still not in $\mathcal{R}_0$ proceed as follows.  

Consider any $r\in(0,\vmax]^n$ such that $r_i\leq r^*_i$ for all $i$. Note that $E_i(r)\leq -\lambda r$ for all $i$. Denote by $E_{(1)}(r)$ the lowest of $E_i$ and by $E_{(2)}(r)$ the second lowest. Assume that the numbers $E_i(r)$ are not all identical, i.e. $E_{(1)}(r)<E_{(2)}(r)$ (if they are not, skip right away to the next paragraph). Denote by $L_1(r)$ the set $\{i: E_i(r)=E_{(1)}(r)\}$ and by $L_2(r)$ the set $\{i: E_i(r)=E_{(2)}(r)\}$. Now change all $r_i$ for $i$ in $L_1(r)$ in such a way that 
the new point $\tilde{r}$ satisfies $L_1(\tilde{r})=L_1(r)\cup L_2(r)$, i.e. the values of $E_i$ for $i\in L_1(r)$  and $i\in L_2(r)$ are equalized. The key is to show that $g(r)$ won't decrease after this move. First, all prices $r_i$ for $i\in L_1(r)$ will increase after this move, as the inequality $E_{(1)}(r)< E_{(2)}(r)$ is equivalent to $(1+\lambda_i)r_i-\lambda_i\vmax< (1+\lambda_j)r_j-\lambda_j\vmax$ for all $i\in L_1(r)$ and $j\in L_2(r)$. Second, because $E_i=E_j$ for all $i,j\in L_1(r)$, $r_j=\frac{(1+\lambda_i)r_i+(\lambda_j-\lambda_i)\vmax}{1+\lambda_j}$  for all $i,j\in L_1(r)$ and thus
\beq\label{gexpr}
g(r)=E_{(1)}(r)=r_i(1+\lambda_i)\left(1-\sum_{j\in L_1(r)}\frac{\lambda_j}{1+\lambda_j}\right)+const,
\eeq
for all $i\in L_1(r)$, where $const$ does not depend on $r_i$ for $i\notin L_1(r)$. As $r_i$ will increase during the move and, by supposition, $\sum\frac{\lambda_i}{1+\lambda_i}\leq 1$, $g(r)$ will weakly increase. Also, note that the new prices $\tilde{r}$ are still weakly lower than $r^*$ because $E_i(r)\leq -\lambda r$ still holds for all $i$ after the move.

Now iterate this procedure until all $E_i(r)$ are equalized. At each step, the value of $g(r)$ will weakly increase, and the eventual price $\tilde{r}$ satisfies $\tilde{r}\leq r^*$. At the last step, increase $\tilde{r}$ to $r^*$. $g(r)$ will grow again due to the representation \eqref{gexpr}. This finishes \textbf{Case 1.}

\textbf{Case 2.} The starting point $r\in \mathcal{R}_0$. 
We will extend $g(r)$ to the set $\mathcal{R}_{ext}:=\{r\in(-\infty,\vmax]^n: r_i=0\mbox{ for some }i\}$ by the same formula as on $\mathcal{R}_0$ and show that for every $r\in \mathcal{R}_0$ there is a point $r^{**}\in \mathcal{R}_{ext}$ such that $g(r^{**})\geq g(r)$. After that we show that $g(r^*)\geq g(r^{**})$ for all such $r^{**}$. 

Given a point $r\in\mathcal{R}_0$, take any $k$ such that $r_k=0$. If $E_i(r)>E_k(r)$ one can lower $r_i$ down to the point at which  $E_i(r)=E_k(r)$, without harming $g$ (note that it is possible because $\mathcal{R}_{ext}$ is unbounded from below.) If, one the other hand, $E_i(r)\leq E_k(r)$ for all $i$, the same construction as in Case 1 shows that one can weakly increase prices in a certain fashion such that $g$ is again unharmed. The final point of this process will be a point $r^{**}(k)\in\mathcal{R}_{ext}$ such that $E_i(r^{**})= E_k(r^{**})$ for all $i$ and $g(r^{**}(k))\geq g(r)$. 

$r^{**}(k)$ is given by $r^{**}_i(k)=\frac{\lambda_i-\lambda_k}{1+\lambda_i}\vmax$ for $i$. It remains to prove that $g(r^{**}(k))\leq g(r^*)$ for all $k$. The inequality $g(r^{**}(k))\leq g(r^*)$ reads as 
\[-\sum_{j\neq k}\lambda_j\frac{\lambda_j-\lambda_k}{1+\lambda_j}\vmax-\lambda_k\vmax\leq -\sum_{i=1}^n\frac{\lambda_i^2}{1+\lambda_i}\vmax,\]
which is equivalent to 
\[\lambda_k\left(\sum_{i=1}^n\frac{\lambda_i}{1+\lambda_i}-1\right)\leq 0,\] 
which is certainly true due to the supposition of the lemma. 
\eop}

\pfof{Lemma \ref{geomconstraintexists}}{
We will prove that for all $\lambda$ such that $\sum_i\frac{\lambda_i}{1+\lambda_i}\geq 1$, $\max\limits_{r}R(r,\lambda)=m\lambda+\vmax\left(1-\sum_i\lambda_i\right)$. Then the result will follow because when the inequality is strict, one can lower some $\lambda_j$ to make it equality, and the value of revenue will strictly increase. If lowering one $\lambda_j$ is not enough (i.e., $\sum_i\frac{\lambda_i}{1+\lambda_i}> 1$
even when $\lambda_j=0$), one can start lowering another $\lambda_j$, etc. One will always reach the equality $\sum_i\frac{\lambda_i}{1+\lambda_i}=1$ because its LHS is zero when all $\lambda_i$ are equal to zero.

Indeed, by \eqref{mlambda}, we have that $R(r,\lambda)\leq  
m\lambda+\vmax\left(1-\sum_i\lambda_i\right)$ for all $r$. On the other hand, 
\[R(r^*(\lambda),\lambda)=m\lambda+\min\left\{\vmax\left(1-\sum_{i=1}^n\lambda_i\right),-\sum_{i=1}^n\frac{\lambda_i^2}{1+\lambda_i}\vmax\right\}.\]
If $\sum_i\frac{\lambda_i}{1+\lambda_i}\geq 1$, the first argument of the minimum is weakly lower, so $R(r^*(\lambda),\lambda)=m\lambda+\vmax\left(1-\sum_i\lambda_i\right)$. Hence, the upper bound is achieved, and thus $\max\limits_{r}R(r,\lambda)=m\lambda+\vmax\left(1-\sum_i\lambda_i\right)$.
\eop}

\pfof{Lemma \ref{optimallambdas}}{
First, solve the relaxed problem without the constrain \eqref{geomconstr}. The solution is $\lambda_i^*=\sqrt{\frac{\vmax}{\vmax-m_i}}-1.$ This solution is the solution to the unrelaxed problem whenever is satisfies the constraint \eqref{geomconstr}, which is exactly when $\sum\limits_{i=1}^n\sqrt{1-m_i/\vmax}\geq  n-1$. 

When $\sum\limits_{i=1}^n\sqrt{1-m_i/\vmax}< n-1$. Denoting by $\xi$ the Lagrange multiplier on \eqref{geomconstr} and by $\kappa_i$ on nonnegativity constraints, one can easily show that the multipliers 
\[\xi^*=\vmax-\frac{1}{(n-k^*)^2}\left(\sum\limits_{i=k^*}^n\sqrt{\vmax-m_i}\right)^2,\]
\[\kappa_i^*=\max\{\xi^*-m_i,0\},\]
together with the proposed solution \eqref{lambdahighmeans}, satisfy first-order and complementary slackness conditions. As there is no other solution to first-order and complementary slackness conditions, the objective function is continuous, and the feasible set is compact, we conclude that $\lambda^*$ is the solution to the problem \eqref{finalproblem}-\eqref{geomconstr}, as desired.
\eop}

\pfof{Proposition \ref{optimal prices}}{
Consider the case of low means first. In this case, all $\lambda_i^*$ are strictly positive and $\sum_i\frac{\lambda^*_i}{1+\lambda_i^*}<1$. Hence, at all steps in the proof of Lemma \ref{rstaroptimal}, the value of $g(r)$ increases strictly. Furthermore, $g(r^*(\lambda^*))=-\sum_i \frac{\lambda_i^{*2}}{1+\lambda^*_i}\vmax<\vmax\left(1-\sum_i\lambda^*_i\right)$. Hence, $R(r,\lambda^*)=g(r)$ at all steps in the proof of Lemma \ref{rstaroptimal}. Thus, $r^*(\lambda^*)$ is the unique maximizer of  $R(r,\lambda^*)$. Plugging $\lambda^*$ given by 
\eqref{lambdalowmeans} to $r^*_i=\frac{\lambda^*_i\vmax}{1+\lambda^*}$, one gets the desired answer \eqref{priceslowmeans}.

Now suppose means are high. First, note that if $\lambda_i^*=0$, then lowering $r_i$ to $r^*_i=0$ won't change the value of $g(r)$, hence, any price $r_i\in[0,\vmax]$ is optimal. Without loss of generality, now suppose that all bidders are strictly included. At all other steps in the proof of Lemma \ref{rstaroptimal} $g(r)$ will increase strictly except the last step where all $E_i(r)$ are already equalized and one is ready to equalize them with $-\lambda r$. Indeed, as $L_1(r)$ now includes all bidders, and $\sum_i\frac{\lambda^*_i}{1+\lambda_i^*}=1$, the expression in parentheses in \eqref{gexpr} is zero. Thus, a vector of prices $r$ is optimal if and only if  $E_i(r)=E_j(r)=g(r^*(\lambda^*))=\vmax\left(1-\sum_i \lambda^*_i\right)$ for all $i,j$ and $r_i\leq r^*_i(\lambda^*_i)$ for all $i$. Thus, $r$ should satisfy a non-homogeneous system of linear equations given by $A(\lambda^*)r=b$ where $b_i=\vmax\left(1-\sum\limits_{j\neq i}\lambda_j^*\right)$. As $\sum_i\frac{\lambda^*_i}{1+\lambda_i^*}=1$, by Lemma \ref{Aproperties}, part 1, $A(\lambda^*)$ is singular and its rank is $n-1$. Because $r^*(\lambda^*)$ is a solution to the non-homogeneous system and $\vmax\i_n-r^*(\lambda^*)$ is a solution to the homogeneous system, the set of solutions is given by $\{r^*(\lambda^*)-\alpha (\vmax\i_n-r^*(\lambda^*))|\alpha\in\mathbb{R}\}$. Thus, every vector of optimal prices $r$ satisfies $\vmax\i_n-\hat{r}=(1+\alpha)(\vmax\i_n-r^*(\lambda^*)$, thus 
\[\frac{\vmax-r_i}{\vmax-\hat{r}_j}=\frac{\vmax-r^*_i}{\vmax-r^*_j}=\frac{1+\lambda_j^*}{1+\lambda_i^*}=\sqrt{\frac{\vmax-m_i}{\vmax-m_j}},\] 
where we used the expressions for $\lambda^*_i$ from \eqref{lambdahighmeans}. Finally, summing the inequalities $r_i\leq r^*_i(\lambda^*)$, one gets
\[\sum_{i=1}^n r_i\leq \sum_{i=1}^n \frac{\lambda_i^*}{1+\lambda_i^*}\vmax=\vmax.\]
\eop}

\pfof{Proposition \ref{setmlow}}{First, take any mechanism $M=(p_1(v_2),p_2(v_1))$ that satisfies the conditions in the proposition. It is straightforward to check, using the representation \eqref{transformedproblem}, that $R(M, \lambda^*(LSA^{opt}))=R^*$ where $\lambda^*(LSA^{opt})$ are the optimal Lagrange multipliers for the optimal linear score auction, given by \eqref{lambdalowmeans}. As $R^*=R(M, \lambda^*(LSA^{opt}))\leq R(M, \lambda^*(M))\leq R^*$, it must be that $R(M, \lambda^*(M))= R^*$, so $M$ is optimal.

Now, take any optimal mechanism $M$. By lemma \ref{lambdasame}, $\lambda^*_i(M)>0$. Hence, when $M$ is submitted as a starting mechanism to the proof of theorem \ref{main}, one gets in \textbf{Grand case I}. Denote by $\tilde{p}(M)$ the modified threshold functions for $M$ as constructed in the proof of theorem \ref{main}, given by \eqref{tildepdef}. We will reconstruct $\tilde{p}(M)$ from optimality of $M$. First, note, that the LSA that weakly dominates $M$ should in fact be the optimal LSA. Because any dominating LSA is such that its vector of reserve prices $r$ is a fixed point of $\tilde{p}$, it must by that $r^*$, as defined by \eqref{priceslowmeans} should be a fixed point of $\tilde{p}(M)$. On the other hand, by lemma \ref{lambdasame}, $\tilde{p}_i$ must have slope of $\lambda^*_i$ given by \eqref{lambdalowmeans}. Thus,  $\tilde{p}_i$ should be given by the RHS in condition 1 in the proposition. As $\tilde{p}_i\leq p_i$, condition 1 in the proposition must be fulfilled.

To establish condition 2, note that because $M$ is optimal, revenue may not strictly increase at every step in the proof of theorem \ref{main}, and thus, $\inf\limits_{W_0(p)}(-\lambda^*(M)v)= \inf\limits_{W_0(\hat{p})}(-\lambda^*(M)v)$ where $p$ are the threshold functions corresponding to $M$ and $\hat{p}$ are the threshold functions for the optimal linear score auction. Thus, the set $W_0(p)$ must lie weakly below the line $\lambda^*_1v_1+\lambda^*_2v_2=\lambda^*_1r^*_1+\lambda_2^*r_2^*$. However, if condition 2 is violated, this is not true. Indeed, suppose 
$p_1(v_{2}^0)> \left(\lambda^*_1r^*_1+\lambda_2^*r_2^*-\lambda^*_{2}v_{2}^0\right)/\lambda^*_1$ for some $v_2^0\leq r^*_2$. 
Take any $v_1^0\in(\left(\lambda^*_1r^*_1+\lambda_2^*r_2^*-\lambda^*_{2}v_{2}^0\right)/\lambda^*_1),p_1(v_2^0)$. Then, $p_1(v_2^0)>v_1^0$ by construction and $p_2(v_1^0)>v_2^0$ because condition 1 is satisfied for the function $p_2(v_1)$. Hence, $(v_1^0,v_2^0)$ is a point in $W_0(p)$ that lies above $\lambda^*_1v_1+\lambda^*_2v_2=\lambda^*_1r^*_1+\lambda_2^*r_2^*$. Contradiction.

To establish condition 3, suppose to the contrary that there exist points $v_1'$ and $v_1''$ such that $r_1^*\leq v_1'<v_1''$ and $p_2(v_2'')<p_2(v_1')$. Consider any $v_2^0\in(p_2(v_1''),p_2(v_1')).$ First, note that for any such $v_2^0$ we must have $p_1(v_2^0)\geq v_1''$. If not, then the point $(v_1'',v_2^0)$ proves that the functions $p$ violate the supply constraint. Second, for any such $v_2^0$ we must have $p_1(v_2^0)\leq v_1'$. If not, the point $(v_1',v_2^0)$ belongs to the set $W_0(p)$ and lies above the line $\lambda^*_1v_1+\lambda^*_2v_2=\lambda^*_1r^*_1+\lambda_2^*r_2^*$, which is impossible by condition 2. Summing up, we have $p_1(v_2^0)\leq v_1'<v_1''\leq p_1(v_2^0)$ -- a contradiction. Thus, $p_2(v_1)$ has to be weakly increasing. Analogously for $p_1(v_2)$.

Finally, take any interval $(v_1',v_1'')$, $v_1'\geq r_1^*$ such that $p_2(v_1)$ is strictly increasing on it and any point $v_1^0 \in (v_1',v_1'')$. Suppose to the contrary that $p_1(p_2(v_1^0))\neq v_1^0$. Either $p_1(p_2(v_1^0))> v_1^0$ or 
$p_1(p_2(v_1^0))< v_1^0$. If $p_1(p_2(v_1^0))> v_1^0$, take any point $v_1^1\in(v_1^0,p_1(p_2(v_1^0)))$. Then, as $p_2(v_1)$ is weakly increasing overall and strictly increasing on $(v_1',v_1'')$, it must be that $p_2(v_1^1)>p_2(v_1^0)$. But then the point $(v_1^1,p_2(v_1^0))$ belongs to the set $W_0(p)$ and lies above the line $\lambda^*_1v_1+\lambda^*_2v_2=\lambda^*_1r^*_1+\lambda_2^*r_2^*$, which is impossible by condition 2. If, on the other hand, $p_1(p_2(v_1^0))< v_1^0$, take any point $v_1^1\in(p_1(p_2(v_1^0)),v_1^0)$. Again, it must be that $p_2(v_1^1)<p_2(v_1^0)$. But then the point $(v_1^1,p_2(v_1^0))$ proves that the functions $p$ violate the supply constraint. Contradiction. Analogously for $p_1(v_2)$.
\eop}

\pfof{Proposition \ref{setmhigh}}{The proof of sufficiency is the same as in the proof of proposition \ref{setmlow}. The proof of necessity is as follows.

Take any optimal mechanism $M$. By lemma \ref{lambdasame}, $\lambda^*_i(M)$ is given by \eqref{lambdahighmeans}. By the same logic as in the proof of proposition \ref{setmlow}, some $\hat{r}$ that satisfies the conditions \eqref{priceshighmeans} is a fixed point of $\tilde{p}(M)$. If $\hat{r}>0$, this, along with $\lambda^*_i$, pins down $\tilde{p}(M)$ to be the functions given in the condition 1 in the proposition (regardless of $\hat{r}$). If $\hat{r}_i=0$ for some $i$, we have $0=\max\{\lambda^*_{-i}\hat{r}_{-i}+b_i,0\}$. If $\lambda^*_{-i}\hat{r}_{-i}+b_i<0$, then $\inf\limits_{v_{-i}}(p_i(v_{-i})-\lambda^*_{-i}v_{-i})=b_i<-\lambda^*_{-i}\hat{r}_{-i}=\inf\limits_{v_{-i}}(\hat{p}_i(v_{-i})-\lambda^*_{-i}v_{-i})$ so $M$ is not optimal. Thus, we must have $\lambda^*_{-i}\hat{r}_{-i}+b_i=0$, so $\tilde{p}(M)$ are reconstructed unambiguously to be functions defined in the RHS of conditions 1 and 2 in the proposition. Thus, condition 2 holds. 

To establish condition 3, note that for any optimal LSA $\hat{p}$, revenue is equal to $\lambda^*m+\inf\limits_{W_0(\hat{p^*})}(-\lambda^*v)$ where $\hat{p^*}$ is the optimal LSA with the highest prices $r^*$. Thus, for $M$ to be optimal it must be that $\inf\limits_{W_0(p)}(-\lambda^*v)\geq \inf\limits_{W_0(\hat{p^*})}(-\lambda^*v)$.  Thus, the set $W_0(p)$ must lie weakly below the line $\lambda^*_1v_1+\lambda^*_2v_2=\lambda^*_1r^*_1+\lambda_2^*r_2^*$. Then, one obtains condition 3 similarly to proof of necessity of condition 2 in proposition \ref{setmlow}.

Finally, to establish condition 1, note that if $p_i(v_{-i})> r^*_i+\lambda^*_i(v_{-i}-r^*_{-i})$ for some $v_{-i}^0\geq r^*_{-i}$, then the point $(v_{-i}^0,v_i^0)$ where $v_i^0$ is any point in $(r^*_i+\lambda^*_i(v_{-i}-r^*_{-i}),p_i(v_{-i}))$ belongs to $W_0(p)$ lies above the line $\lambda^*_1v_1+\lambda^*_2v_2=\lambda^*_1r^*_1+\lambda_2^*r_2^*$, which is impossible by condition 3.
\eop}

\pfof{proposition \ref{wcdistr}}{
As noted in the main text, the support of any worst-case distribution must be contained in the set of minimizers of $t^{LSA}(v)-\lambda^*(LSA)v$. It remains to find  the optimal Lagrange multiplier $\lambda^*$ as a function of the generalized reserve prices $r$. $\lambda^*$ maximizes
\[R(r,\lambda)=m\lambda+\min\{v(1-\lambda_1-\lambda_2),r_1-r_2\lambda_2-\vmax\lambda_1,r_2-r_1\lambda_1-\vmax\lambda_2,-r\lambda\}.\]
It follows from a straightforward, but tedious analysis that 
\begin{itemize}
\item If $r_2<\overline{r}_2(r_1)$ and $r_2<\vmax-r_1$, $\lambda_1^*=\frac{\vmax-r_2}{\vmax-r_1}$, $\lambda^*_2=\frac{\vmax-r_1}{\vmax-r_2}$;
\item If $r_2>\overline{r}_2(r_1)$ but $r_2<\vmax-r_1$, $\lambda_i^*=\frac{r_i}{\vmax-r_i}$, $i=1,2$;
\item If $r_2>\vmax-r_1$, either $\lambda_1^*=\frac{r_1}{\vmax-r_1}$ and $\lambda_2^*=\frac{\vmax-r_1}{\vmax-r_2}$ or $\lambda_1^*=\frac{\vmax-r_2}{\vmax-r_1}$ and $\lambda_2^*=\frac{r_2}{\vmax-r_2}$.
\end{itemize}
From this, the statement of the proposition follows. 
\eop}

\pfof{theorem \ref{main2}}{
The proof is the same as for theorem \ref{main} with a few modifications. Lemma \ref{getridofO} does not hold, as the inequality \eqref{cangetridofo} does not necessarily hold for $i=2$, $j=1$. Define $\tilde{W}_2:=\{v: v_2\geq p_2(v_1)\mbox{ and }v_1\leq p_1(\vmax_2)\}$. We have $\inf\limits_{v\in W_2}(p_2(v_1)-\lambda v)=\inf\limits_{v\in \tilde{W}_2}(p_2(v_1)-\lambda v)$. For bidder 1, we still have $\inf\limits_{v\in W_1}(p_1(v_2)-\lambda v)=\inf\limits_{v\in W_1^{\geq}(p)}(p_1(v_2)-\lambda v)$. 

The consideration of \textbf{Grand case II} goes identically to theorem \ref{main}, so consider \textbf{Grand case I}.  We have 
$\inf\limits_{v\in \tilde{W}_2}(p_2(v_1)-\lambda^* v)=\inf\limits_{v_1\in[0,\tilde{v}_1]}(p_2(v_1)-\lambda^*_1v_1)-\lambda^*_2\vmax_2$, where $\tilde{v}_1:=p_1(\vmax_2)$. For bidder 1, nothing changes.
Thus, at the first step of transforming a given pair functions $(p_1(v_2),p_2(v_1))$ we transform $p_1(v_2)$, as before, to $\tilde{p}_1(v_2)$, given by \eqref{tildepdef}, and transform $p_2(v_1)$ to 
\beq\label{tildepdef2}
\tilde{p}_2(v_1):=
\begin{cases}
\max\{\lambda^*_{1}v_{1}+\inf\limits_{w\in[0,\tilde{v}_1]}(p_2(w)-\lambda^*_{1} w),0\}, & v_1\leq \tilde{v}_1;\\
\vmax_2, & v_1>\tilde{v}_1.
\end{cases}
\eeq

By the same logic as in the proof of proposition \ref{tildebetter}, $R(\tilde{p},\lambda^*)\geq R(p,\lambda^*)$. 

Note that for all $v_2\in[0,\vmax_2]$, $\tilde{p}_1(v_2)\leq \tilde{p}_1(\vmax_2)\leq p_1(\vmax_2)=\tilde{v}_1$, as $\tilde{p}_1$ is nondecreasing and $\tilde{p}_1(v_2)\leq \tilde{p}_1(v_2)$ for all $v_2$. Thus, one may consider $\tilde{p}$ as a continuous map from $[0,\tilde{v}_1]\times[0, \vmax_2]$ to itself. Thus, a fixed point exists. 
The rest of the proof goes analogously to the proof of theorem \ref{main}, except the consideration of \emph{Subcase 2} of \textbf{Case 3}. There, one shows directly that $\frac{\lambda_1^*}{1+\lambda^*_1}+ \frac{\lambda_2^*}{1+\lambda^*_2}>1$ implies that either $b_1<0$ or $b_2=0$. Then, the point defined by $w^*=(\max\{b_1,0\},\max\{b_2,0\})$ will be a desired fixed point with at least one zero coordinate.

Note that the construction will identify a dominating linear score auction such that the boundary between $W_1$ and $W_2$ goes through the point $(\tilde{v}_1,\vmax_2)$. 

\eop}

\pfof{proposition \ref{optparam}}{
}{Again, because for any optimal mechanism $\lambda^*\geq 0$, one can restrict attention only for such $\lambda$. Solving the inner problems as in the proof of \ref{innerproblsol}, one obtains that $R(p,\lambda)$ is equal to 
\beq
\begin{cases}
\lambda m+\min\left\{\tilde{v}_1-\lambda_2\vmax_2-\lambda_1\vmax_1,\vmax_2-\lambda_1\tilde{v}_1-\lambda_2\vmax_2,\min\limits_i(r_i-\lambda_{-i}r_{-i}-\lambda_i\vmax_i),-\lambda r \right\}, & r>0,\\
\lambda m+\min\left\{\tilde{v}_1-\lambda_2\vmax_2-\lambda_1\vmax_1,\vmax_2-\lambda_1\tilde{v}_1-\lambda_2\vmax_2,\min\limits_i(r_i-\lambda_{-i}r_{-i}-\lambda_i\vmax_i)\right\}, & \mbox{o/w}.
\end{cases}
\eeq

The problem is to maximize this function over $r\in[0,\vmax_1]\times[0,\vmax_2]$, $\lambda\in \mathbb{R}^2_+$, and $\tilde{v}_1\in[0,\vmax_1]$. Again, it is more convenient to optimize first over $r$ and $\tilde{v}$, and then over $\lambda$. Optimizing over $\tilde{v}_1$, one does weakly better equalizing 
$\tilde{v}_1-\lambda_2\vmax_2-\lambda_1\vmax_1$ and $\vmax_2-\lambda_1\tilde{v}_1-\lambda_2\vmax_2$, thus setting 
$\tilde{v}_1^*=\frac{\lambda_1\vmax_1+\vmax_2}{\lambda_1+1}.$
The common value of the two expressions will be 
\beq\label{commonvalue}
\frac{\lambda_1\vmax_1+\vmax_2}{\lambda_1+1}-\lambda_2\vmax_2-\lambda_1\vmax_1.
\eeq

Then, defining $r^*_i(\lambda)=\frac{\lambda_i}{1+\lambda_i}\vmax_i$, the proofs of lemmas \ref{rstaroptimal} and \ref{geomconstraintexists} go through with \eqref{commonvalue} replacing $\vmax\left( 1-\sum_i\lambda_i\right)$. (Also, one lowers the value of $\lambda_2$, not any $\lambda$ to obtain a strict improvement. 

Thus, we arrive at a problem similar to \eqref{finalproblem}-\eqref{geomconstr}:
\begin{align}
\max\limits_{\lambda\geq 0} &\mbox{ } \sum_{i=1}^2 m_i\lambda_i-\frac{\lambda_i^2}{1+\lambda_i}\vmax_i\label{finalproblem1}\\
\mbox{s.t. } &\sum_{i=1}^2\frac{\lambda_i}{1+\lambda_i}\leq 1  \label{geomconstr1}
\end{align}

Its solution again differs depending on whether the constraint \eqref{geomconstr1} binds. It does not bind when means are low ($\sqrt{1-m_1/\vmax_1}+\sqrt{1-m_2/\vmax_2}>1$) and binds otherwise. If it does not bind, we obtain 
\[\lambda_i^*=\sqrt{\frac{\vmax_i}{\vmax_i-m_i}}-1,\]
and thus, recovering optimal prices using same logic as in the proof of proposition \ref{optimal prices}, one gets \eqref{optpriceslowmeans}. 
To recover all optimal slopes in this case (aside from an optimal slope corresponding to $\tilde{v}_1^*=\frac{\lambda_1\vmax_1+\vmax_2}{\lambda_1+1}.$), one uses proposition \ref{setmlow} (its logic is unchanged when upper bounds are different). Restricting condition 1 of proposition \ref{setmlow} to linear score auctions, one gets the condition \eqref{optgammameanslow}.

If the constraint \eqref{geomconstr1} binds, one obtains from first-order conditions that the equation \eqref{gamma} must hold for the optimal $\lambda^*_1$. But note that when the constraint binds, $\lambda*_1\lambda_2^*=1$, and thus, the optimal slope $\gamma^*$ is unique (as in Figure \ref{fig:sethighmeans}). Hence, $\gamma^*=\lambda_1^*$, so $\gamma^*$ satisfies \eqref{gamma}.To obtain the set of optimal prices, one again uses the same proof as in proposition \ref{optimal prices}.   
\eop}

\section{Worst-case distributions}\label{wcdistrsect}
The proof of theorem \ref{main} didn't consider worst-case distributions thanks to the representation afforded by linear programming duality. However, knowing them may also be useful to know how Nature reacts to a given mechanism and thus gain a deeper understanding of the game between the seller and Nature. In this section, we discuss worst-case distributions for corner-hitting linear score auctions (not necessarily optimal ones). 

The worst-case distributions can be deduced from a complementary slackness result in \cite{smith1995generalized}.  Namely, it follows that for any mechanism $M$ any distribution solving the inner problem \eqref{problem} also solves the \emph{unconstrained} problem
$$\inf_{F\in\Delta} \left\{t^{M}(v)-\lambda^*(M) v\right\},$$
where $\lambda^*(M)$ is the solution to the dual problem \eqref{biconj} for the mechanism $M$ and $\Delta$ is the set of all Borel distributions on $[0,\vmax]^n$. That is, if a worst-case distribution exists, its support must be contained in set of minimizers of the function $t^{M}(v)-\lambda^*(M)v$ over $[0,\vmax]^n$.

Consider $n=2$ and any corner-hitting linear score auction with parameters $r=(r_1,r_2)$. Before we proceed, we change the rules a little by stating that when $v_i\leq r_i$ for $i=1,2$, the object is unsold. This does not change the infimum over distributions but ensures that a worst-case distribution always exists. Then, when $\lambda^*_i>0$ for $i=1,2$ (it will happen whenever $r_i<m_i$) it follows immediately from the structure of the function $t^{LSA}(v)$ that the support of any worst-case distribution is contained in the set
$$S^*:=\{(r_1,r_2)\}\cup\{v:v_i\geq r_i,\max\{v_1,v_2\}=\vmax\}.$$
Depending on the means $m$ and prices $r$, different subsets of the set $S^*$ may be selected as the support of a worst-case distribution. It turns out that there may be three types of worst-case distributions when $r_i<m_i$, described in the following table: 

\begin{center}
\begin{tabular}{ |c|c| } 
 \hline
 Type & Support \\
 \hline
 I & Any subset of $\{v:v_i\geq r_i,\max\{v_1,v_2\}=\vmax\}$ \\ 
 II & $\{(r_1,r_2),(r_1,\vmax),(\vmax,r_2)\}$   \\ 
 III & $\{(r_1,r_2),(r_1,\vmax)\mbox{ or }(\vmax,r_2),(\vmax,\vmax)\}$   \\
 \hline
\end{tabular}
\end{center}

Figure \ref{fig:wc} shows how the support of the worst-case distribution depends on prices $r$ for fixed means $m$. 
Interestingly, if prices are relatively low, Nature chooses a distribution such that sale happens with probability 1 and does not try to ``undercut'' the seller. Only when prices rise beyond a certain boundary, Nature starts to induce no trade with positive probability in the worst case. When prices rise even further, Nature switches from inducing negative correlation between values (so that the lowest bid is always very low) to inducing positive correlation (so that two bidders are effectively replaced by one).

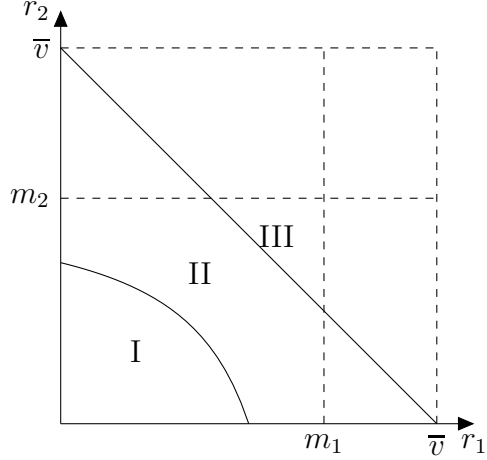
\begin{figure}[h!]
\centering
\begin{tikzpicture}[>=triangle 45, xscale=5,yscale=5]
\draw[->] (0,0) -- (1.1,0) node[below] {$r_1$};
\draw[->] (0,0) -- (0,1.1) node[left] {$r_2$};
\draw[dashed] (1,0)--(1,1);
\draw[dashed] (0,1)--(1,1);
\node[below] at (1,0){$\vmax$};
\node[left] at (0,1){$\vmax$};
\draw (0,1)--(1,0);
\draw[domain=0:0.5]
plot (\x, {(0.6*(1-\x)-0.3)/(0.7-\x)});
\node[below] at (0.2,0.25){I};
\node[left] at (0.43,0.4){II};
\node[left] at (0.65,0.5){III};
\draw[dashed] (0.7,0)--(0.7,1);
\draw[dashed] (0,0.6)--(1,0.6);
\node[below] at (0.7,0){$m_1$};
\node[left] at (0,0.6){$m_2$};
\end{tikzpicture}
\caption{Different types of distributions are worst-case for different prices $r$ }\label{fig:wc}
\end{figure}

Define $\overline{r}_2(r_1):=\frac{m_2(\vmax-r_1)-\vmax(\vmax-m_1)}{m_1-r_1}$. 
 
\begin{proposition}\label{wcdistr}
Suppose $n=2$ and and consider a corner-hitting linear score auction with parameters $r_1<m_1$ and $r_2<m_2$. The rules are modified so that when $v_i\leq r_i$ for $i=1,2$, the object is unsold. Then, the set of worst-case distributions is a subset $\Delta_{WC}$ of $\Delta(m,\vmax)$ such that for every $F\in \Delta(m,\vmax)$: 
\begin{enumerate}
\item If $r_2<\overline{r}_2(r_1)$ and $r_2<\vmax-r_1$, $F\in \Delta_{WC}$ if and only if $\supp (F)\subseteq \{v:v_i\geq r_i,\max\{v_1,v_2\}=\vmax\}$;
\item If $r_2>\overline{r}_2(r_1)$ but $r_2<\vmax-r_1$, $F\in \Delta_{WC}$ if and only if $\supp (F)=\{(r_1,r_2),(r_1,\vmax),(\vmax,r_2)\}$;
\item If $r_2>\vmax-r_1$, $F\in \Delta_{WC}$ if and only if $\supp (F)\subseteq \{(r_1,r_2)\}\cup\{v:v_1=\vmax,v_2\geq r_2\}$ or $\supp (F)\subseteq \{(r_1,r_2)\}\cup\{v:v_2=\vmax,v_1\geq r_1\}$.
\end{enumerate} 
\end{proposition}

Note that $F$ is pinned down from the means constraints and the support if the support contains three points.

\section{Mechanisms with negative $\lambda_i$}\label{appC}
\textbf{Example 1.}
Suppose there are two bidders with $v_i\in[0,1]$. The first bidder never gets the object, but the second bidder gets one iff 
\[v_2>p_1(v_1)=r+k\cdot v_1\]
where $r>0$, $k>0$, $r+k<1$. That is, the price that the second bidder pays upon getting the object depends positively on first bidder's report. 

If $m_2\in (r+km_1,1-(1-r-k)m_1)$, the worst-case distribution may be shown to be a ternary distribution with support $\{(0,r),(0,1),(1,r+k)\}$. Bidder 2 does not get the object when $v=(0,r)$ or $v=(1,r+k)$. So Nature is undercutting the seller by putting mass at these value profiles. 
The worst-case expected revenue if $m_2\in (r+km_1,1-(1-r-k)m_1)$ is \[-\frac{kr}{1-r}m_1+r\frac{m_2-r}{1-r},\] thus
$\lambda^*_1=-\frac{kr}{1-r}<0$. 
The seller would lose (and Nature win) from a higher mean of bidder's 1 value. 

Intuitively, with a higher $m_1$, Nature can put more mass on the point $(1,r+k)$ but because $r+k>r$, this would mean undercutting the seller with higher (on average) values of $v_2$. With a fixed $m_2$, this means that the total mass of no-sale value profiles $\{(0,r),(1,r+k)\}$ may also be increased -- and this harms the seller.

\textbf{Example 2.}
Suppose there are two bidders with different upper bounds and means such that $\vmax_1>m_1>\vmax_2$. 

Consider a corner-hitting LSA with $r_1\in (\vmax_2,m_1)$, $r_2<\vmax_2$. Suppose also that the parameters satisfy 
\[m_2<\vmax_2-\frac{\vmax_2}{\vmax_1-r_1}(\vmax_1-m_1).\]
It this case, the worst-case distributions is a three-point distribution on $(r_1,0)$, $(\vmax_1,0)$ and $(\vmax_1,\vmax_2)$ (if the tie breaking is such that the second bidder gets the object if $v=(\vmax_1,\vmax_2)$, otherwise this point can be approximated). But then, the higher $m_2$ is, the higher probability Nature can put on the point $(\vmax_1,\vmax_2)$, which harms the seller because her revenue is $\vmax_2$ if $v=(\vmax_1,\vmax_2)$ and it is $r_1>\vmax_2$ if $v=(\vmax_1,0)$. Indeed, the expected revenue is equal to 
$\left(1-\frac{r_1}{\vmax_2}\right)m_2+\frac{m_1-r_1}{\vmax_1-r_1}r_1$,
so $\lambda_2^*=1-\frac{r_1}{\vmax_2}<0$.

\end{document}